%% file: arxiv-main.tex
\documentclass{article}
\usepackage{waingarten}
\usepackage{graphicx}
\usepackage{tikz}
\usepackage{stmaryrd}

\algdef{SE}[SUBALG]{Indent}{EndIndent}{}{\algorithmicend\ }%
\algtext*{Indent}
\algtext*{EndIndent}

\usetikzlibrary{arrows.meta, positioning, shapes.geometric, fit, quotes, calc}

\newcommand{\AFN}{\textup{\texttt{AFN}}\xspace}

\newcommand{\Enc}{\mathrm{Enc}} 
\newcommand{\Dec}{\mathrm{Dec}}

\newcommand{\minrank}{\mathrm{minrank}} 
\newcommand{\interval}[1]{\llbracket #1 \rrbracket}

\providecommand{\email}[1]{\texttt{#1}}

\title{A Polynomial Space Lower Bound for Diameter Estimation in Dynamic Streams} 
\author{
Sanjeev Khanna\thanks{University of Pennsylvania. Supported by NSF award CCF-2402284 and AFOSR award FA9550-25-1-0107. \{\email{sanjeev@cis.upenn.edu}\}} \and
Ashwin Padaki \thanks{University of Pennsylvania. Supported by NSF Graduate Research Fellowship Program. \{\email{apadaki@seas.upenn.edu}\}} \and
Krish Singal \thanks{University of Pennsylvania. Supported by NSF award CCF-2402284. \{\email{ksingal@seas.upenn.edu}\} } \and
Erik Waingarten\thanks{University of Pennsylvania. Supported by the National Science Foundation (NSF) under Grant No. CCF-2337993. \{\email{ewaingar@seas.upenn.edu}\}}
}

\begin{document} 

\maketitle

\input{abstract} 
\input{intro}

\input{preliminaries} 
\input{streaming-to-sketching}
\input{sketch-lb}

\input{upper-bounds}

\bibliographystyle{alpha} 
\bibliography{waingarten} 

\appendix 

\input{appendix} 

\end{document}

%% file: abstract.tex
\begin{abstract}
We study the space complexity of estimating the diameter of a subset of points in an arbitrary metric space in the dynamic (turnstile) streaming model. The input is given as a stream of updates to a frequency vector $x \in \mathbb{Z}_{\geq 0}^n$, where the support of $x$ defines a multiset of points in a fixed metric space $\calM = ([n], \sfd)$. The goal is to estimate the diameter of this multiset, defined as $\max\{\sfd(i,j) : x_i, x_j > 0\}$, to a specified approximation factor while using as little space as possible.

In insertion-only streams, a simple $O(\log n)$-space algorithm achieves a 2-approximation. In sharp contrast to this, we show that in the dynamic streaming model, any algorithm achieving a constant-factor approximation to diameter requires polynomial space. Specifically, we prove that a $c$-approximation to the diameter requires $n^{\Omega(1/c)}$ space. Our lower bound relies on two conceptual contributions: (1) a new connection between dynamic streaming algorithms and linear sketches for {\em scale-invariant} functions, a class that includes diameter estimation, and (2) a connection between linear sketches for diameter and the {\em minrank} of graphs, a notion previously studied in index coding.
We complement our lower bound with a nearly matching upper bound, which gives a $c$-approximation to the diameter in general metrics using $n^{O(1/c)}$ space.
\end{abstract}

%% file: intro.tex
\newpage
\section{Introduction} \label{sec: intro}

We study streaming algorithms for estimating the diameter of a set of points in a metric space. Let $\calM = ([n], \sfd)$ be a metric space on $n$ elements, where $\sfd$ is a distance function. We encode (multi-)subsets of $\calM$ as vectors $x \in \Z^n_{\geq 0}$, where the entry $x_i$ denotes the multiplicity of element $i \in [n]$ in the subset. The goal is to estimate the diameter of the subset encoded by $x$ in a streaming model, where the vector $x$ is updated dynamically through insertions and deletions (i.e., increments and decrements to its coordinates), and the diameter of the set represented by $x$ is defined as:
\[
\diam_{\calM}(x) = \max\left\{ \sfd(i,j) : x_i, x_j > 0\right\}.
\]
We assume oracle access to the metric $\sfd$, without limiting the number of distance queries. In particular, we allow algorithms to depend arbitrarily on the structure of the metric. The central objective is to minimize the space complexity required to output a $c$-approximation to $\diam_{\calM}(x)$: a number $\boldeta\in\R$ satisfying \[\diam_{\calM}(x) / c < \boldeta \leq \diam_{\calM}(x)\] with high probability.\footnote{The fact that a $c$-approximation must be \emph{strictly} larger than the diameter over $c$ is non-standard (usually, an output which is exactly diameter over $c$ is considered a $c$-approximation). We opt for this notion of approximation, since it will make the upper and lower bounds slightly cleaner to state.}

In insertion-only streams (where updates only increment $x$), there is a simple $O(\log n)$-space algorithm achieving a $c$-approximation for any $c > 2$: store the first inserted element $i^* \in [n]$ and maintain the maximum distance $\sfd(i^*, j)$ over all subsequently inserted elements $j$. This quantity corresponds to the radius of the smallest enclosing ball centered at $i^*$, which is within a factor $2$ of the diameter.\footnote{By a reduction from the indexing problem in communication complexity, there exist metric spaces $\calM$ where achieving a $c$-approximation for $c\leq 2$ requires $\Omega(n)$ bits of space, even in insertion-only streams.}

However, this approach breaks down in the presence of deletions: if $i^*$ is later removed from the set, the algorithm loses all relevant information. This motivates the following fundamental question:

\begin{quote}
Do there exist small-space streaming algorithms that can approximate the diameter in the dynamic (turnstile) model, where both insertions and deletions are allowed?
\end{quote}

Note that, given {\em prior knowledge} of some $i \in [n]$ with $x_i \neq 0$ in the final vector, one could reduce the problem to maintaining distances $\sfd(i, j)$ for updated elements $j$, and use $\ell_0$-sampling~\cite{FIS05, JST11} to estimate the radius around $i$ under insertions and deletions. However, identifying such an $i$ {\em during} the stream is challenging: while an $\ell_0$-sampler can recover such an index in logarithmic space, its output only becomes available at the end of the stream—too late to maintain relevant distances involving $i$. Determining the space-complexity of streaming the diameter of a turnstile stream of points in a high-dimensional Euclidean space was asked by Krauthgamer~\cite{K22}. While the current manuscript does not resolve that problem (which currently remains open), it is a necessary first step.

\subsection{Our Results}

We establish nearly tight space-approximation tradeoffs for estimating the diameter of general metrics in dynamic streams. Our main result is a lower bound that shows any randomized $c$-approximation streaming algorithm which succeeds with high-constant probability must use polynomial space (as a function of $n$ and $c$), demonstrating a strong separation from the logarithmic-space 2-approximation in insertion-only streams.

\begin{theorem}\label{thm:lb}
For any $c \geq 1$, there exists a metric $\calM = ([n], \sfd)$ such that any dynamic streaming algorithm that computes a $c$-approximation to $\diam_{\calM}(x)$ for all $x \in \{0,1\}^n$ must use $\tilde{\Omega}(n^{1/(2\ceil{c} - 1)})$ bits of space.
%
\end{theorem}

As we discuss in Section~\ref{sec:overview}, the proof of Theorem~\ref{thm:lb} relies on two important conceptual developments. The first is further developing the connection between dynamic streaming algorithms and linear sketches~\cite{G08, LNW14, KP20}. We consider problems defined on partial functions $g \colon \Z^n \to \{0,1,*\}$ that are \emph{scale-invariant} (i.e. $g(a x) = g(x)$ for any scaling $a\in\Z_{+}$ of a vector $x \in \Z^n$)---note that $\diam_{\calM}(x)$ satisfies this property. We show that for this class of problems, dynamic streaming algorithms can be made linear sketches without any overhead in complexity. In contrast, the space complexity of the sketch in~\cite{LNW14} increases logarithmically in the maximum entry of $x$, which is necessary for general problems but would render the second part of our argument meaningless. 

The second is a connection between linear sketches for the diameter and the minrank of a particular graph associated with the metric $\calM$. The minrank of a graph was developed in~\cite{BBJK06} for the problem of index coding with side information, and we use a bound from~\cite{ABGMM20} on the minrank of a random graph in order to lower bound the dimension of linear sketches. Complementing this, we give a matching upper bound that achieves a $c$-approximation in $n^{O(1/c)}$ space for general metrics with polynomially bounded aspect ratio:

\begin{theorem}\label{thm:ub}
For any metric space $\calM = ([n], \sfd)$ and any $c > 6$, there is a dynamic streaming algorithm that computes a $c$-approximation to $\diam_{\calM}(x)$ with  $\tilde{O}(n^{1 / \lfloor (c-2)/4\rfloor})$ bits of space.\footnote{We assume, for convenience, that final frequency vectors have entries with magnitude at most $\poly(n)$ (to support frequency vectors of magnitude $m$, we pay a multiplicative $O(\log m)$ in the bit-complexity). Furthermore, we assume for convenience, that the metric $\calM$ has aspect ratio $\poly(n)$ (to support aspect ratio $\Delta$, we pay an additive $\polylog(n\Delta)$ factor in the bit-complexity).}
\end{theorem}

The proof of Theorem~\ref{thm:ub} makes use of low-dimensional $\ell_{\infty}$-embeddings and $\ell_0$-sampling. Specifically, a result of~\cite{M96b} shows that any metric on $n$ points admits a distortion-$(2k+1)$ embedding into $\ell_{\infty}^d$ with dimension $d = O(kn^{1/(k+1)} \log n)$. We show that for dynamic pointsets in $\ell_{\infty}^d$, a constant factor approximation to the diameter be accomplished in $\tilde{O}(d)$ space, by maintaining a constant number of $\ell_0$-samplers per coordinate.

Our techniques also yield new embedding lower bounds. We show that the same random metrics used in Theorem~\ref{thm:lb} require large dimension when embedded into $\ell_{\infty}$ with bounded distortion:

\begin{theorem}\label{thm:embed-lb}
For any $k\in\N$ and $c < 2k+1$, let $\calbM$ be the shortest path metric on a random bipartite graph $\calG(n,n,p)$ with $p=n^{-1 + 1/(2k-1)} / \polylog(n)$. Then with high probability, any embedding of $\calbM$ into $\ell_{\infty}$ with distortion $c$ requires dimension $\tilde{\Omega}(n^{1/(2k-1)})$.
\end{theorem}

The benefit of Theorem~\ref{thm:embed-lb} is that it gives precise constants for embedding a simple-to-describe metric into low-dimensional $\ell_{\infty}$. To best of our knowledge, comparable theorems for dimension-distortion tradeoffs for embeddings into normed spaces use one of two routes. They show that subsets of low-dimensional norms cannot host shortest path metrics on expanders (of the form in~\cite{M97, LMN05, N17}, which are lossy with respect to constant-factors in the dimension-distortion tradeoffs), or use a counting argument in Matousek~\cite{M96b} (attributed to~\cite{B85,JLS87}) which relates the dimension-distortion tradeoff to existence of dense graphs with no short cycles. Even though the counting argument in~\cite{M96b} gives precise constants and can obtain a better tradeoff than Theorem~\ref{thm:embed-lb} when instantiated with the currently-known densest cycle-free graphs, the appeal of Theorem~\ref{thm:embed-lb} is the simplicity of the hard metric, and the technique, which identifies the minrank of a graph associated with the metric as the relevant quantity in dimension-distortion tradeoffs.

\subsection{Related Work} 

There has been significant work on streaming algorithms for geometric problems, with a particular focus on upper bounds and for points lying in Euclidean space. In the insertion-only model, the current best algorithms for estimating the diameter of $n$-point subsets of $d$-dimensional Euclidean space are in~\cite{AS15}, who give a $(\sqrt{2} + \eps)$-approximation using $d \cdot \poly(1/\eps)$ space; furthermore, they show that any $\sqrt{2} (1 - 2/d^{1/3})$-approximation requires $\Omega(\min\{ n, e^{d^{1/3}}\})$ space.\footnote{Their reduction to indexing also implies $(\sqrt{2} - \eps)$-approximation requires space $\min\{ \Omega(n), 2^{\Omega(\eps^2 d)}\}$. For $n = 2^{\Omega(\eps^2 d)}$, find a collection $\calC \subset \mathbb{S}^d$ of size $n$ with pairwise distances $\sqrt{2}\pm\eps$; inputs $x \in \{0,1\}^n$ are associated with subsets of $\calC$, and queries consisting of antipodal points have diameter $2$ or $\sqrt{2}+\eps$.}  For low-dimensional Euclidean spaces (when $d$ is a constant) and in the insertion-only model, one can achieve $(1+\eps)$-approximation in $\poly(1/\eps)$ space using coresets~\cite{AHV04, BC08}. For dynamic streaming algorithms, the best space-approximation tradeoff for (high-dimensional) Euclidean spaces appears in~\cite{I03}, where they give an insertion-only algorithm for $c$-approximation using $O(dn^{1/(c^2-1)})$ space, which can be slightly modified (using $\ell_0$-sampling) to obtain a $c(1+\eps)$-approximation in dynamic streams using $\tilde{O}(dn^{2/(c^2-1)}/\eps)$ space.

In general, there are many geometric problems whose space-approximation exhibits a qualitatively similar tradeoff (i.e. for any fixed constant approximation $c > 1$, a polynomial space complexity which decays with $c$). We have already mentioned Euclidean diameter above; the Euclidean minimum spanning tree of a set of $n$ points admits a $c$-approximation in $n^{O(1/\sqrt{c})}$ space~\cite{CJLW21, CCJLW23}; Euclidean facility location admits a $c$-approximation in $n^{O(1/c)}$ space~\cite{CJKVY22}; and the Earth mover's distance over $[\Delta]^2$ admits a $c$-approximation in $\Delta^{O(1/c)}$ space. The one technique for proving lower bounds on large approximation factors reduces the streaming lower bound to the lower bounding the distortion of $\ell_1$-embeddings (for example,~\cite{AIK08, AKR15, CJLW21}); however, all such arguments lower bound the space complexity \emph{times} the approximation, and thus cannot match the tradeoffs of Theorem~\ref{thm:lb} and Theorem~\ref{thm:ub}.

\begin{remark}
In some results cited above (namely,~\cite{I03, CJLW21, CCJLW23, CJKVY22}), the meaning of the parameter $n$ differs slightly. In this paper, $n$ is the size of the underlying metric space $\calM$, and in those works, $n$ is the size of the Euclidean subset (since the metric $(\R^d, \ell_2)$ is infinite). However, in the above cited papers, one may often perform a dimension reduction to $O(\log n)$ dimensions and discretization step. This effectively reduces the problem to one where the $n$-point subset comes from a fixed $n^{O(1)}$-sized sub-metric of $\ell_2$. For these cases, the size of the subset and the size of the metric become polynomially related, which is why we claim these are qualitatively similar.
\end{remark}

\subsection{Technical Overview}\label{sec:overview}

\paragraph{Communication Game and the Minrank.} A natural approach to lower bounding the space complexity of a streaming algorithm for diameter is to consider the following one-way communication game. There is a fixed public metric $\calM$, which in our case is the shortest path metric on a random bipartite graph $\bG = (U, V, \bE) \sim \calG(n,n,p)$ (Definition~\ref{def:random-metric}), along with a public source of randomness.

\begin{itemize}
\item Alice receives a uniformly random vector $\bx \sim \{0,1\}^n$ encoding a subset of $U$, by associating the $n$ coordinates with the elements $u_1,\dots, u_n \in U$.
\item Bob receives a uniformly random index $\bi \in [n]$ corresponding to a vertex $v_{\bi} \in V$, as well as a vector $\by \in \{0,1\}^n$ such that $\by_j = \bx_j$ for all $u_j \notin N(v_{\bi}) \cup \{ u_{\bi} \}$, and $\by_j = 0$ otherwise (here, $N(v_{\bi})$ denotes the neighborhood of $v_{\bi}$).
\end{itemize}

Both Alice and Bob have oracle access to the metric $\calM$, and Alice must send a message that enables Bob to recover $\bx_{\bi}$ with high probability. Assuming the existence of a dynamic streaming algorithm for approximating diameter, Alice simulates insertions of every $u_j$ with $\bx_j = 1$ and sends the algorithm’s state to Bob. Bob then simulates deletions of all $u_j$ with $\by_j = 1$ and queries the streaming algorithm. We set $p \lesssim n^{1/(2k-1)}/n$, so that $\sfd(u_i, v_i) \geq 2k+1$ with high probability. Hence, if $\bx_{\bi} = 0$, the vector $\bx - \by$ is supported entirely within $N(v_{\bi})$ and the diameter is at most 2; if $\bx_{\bi} = 1$, the resulting set has diameter at least $2k$ (except for the unlikely case every $u_j \in N(v_{\bj})$ has $\bx_j = 0$). For any $c\leq k$, a $c$-approximation to $\diam_{\calM}(\cdot)$ can thus distinguish between these two cases. This setup corresponds to an instance of {\em index coding} with side information~\cite{BBJK06}, where Bob knows almost everything about Alice's subset (all but the membership of elements in $N(v_{\bi}) \cup \{ u_{\bi} \}$). The communication complexity of this problem remains open in general, but progress has been made under the restriction to linear protocols over finite fields. Over $\F_2$, the complexity is characterized by the \emph{minrank} of the knowledge graph (i.e., the index $i \in [n]$ has directed edges to everything except $N(v_i)$). The minrank is the minimum rank of an $n \times n$ matrix that is nonzero on the diagonal and zero at every $(i,j)$ that is not an edge of the knowledge graph. For linear protocols over $\F_2$, the message length (i.e. the sketch dimension) must be at least the minrank, which is $\tilde{\Omega}(n^{1/(2k-1)})$ for our distribution over graphs~\cite{ABGMM20}. This completes the lower bound for linear sketches over $\F_2$, and motivates our approach to general streaming lower bounds via reductions to linear sketches~\cite{G08, LNW14, KP20}, which bypasses the need to analyze the communication problem in full generality.

\paragraph{General Streaming Lower Bounds.} 
A subtlety arises when generalizing to a field $\F$ beyond $\F_2$, since there are two things to generalize. First, the communication problem may consider inputs $\bx$ drawn from $\F^n$ instead of $\F_2^n = \{0,1\}^n$; second, linear protocols may be over $\F$ instead of $\F_2$. If the domain of the communication game and linear protocols use the same field $\F$, bounds easily generalize.\footnote{See Remark 2.3 in  \cite{LS09}.} However, in a general streaming lower bound, we must first fix an input distribution (which is tied to the domain of the communication game), and then allow for a streaming algorithm to be tailored to that distribution.

This issue appears concretely in~\cite{LNW14}: the reduction to linear sketches produces a sketch over $\R$ whose dimension and entries grow with the magnitude of the final frequency vectors (see $m$ in Theorem~16 of~\cite{LNW14}). Consider the following faulty approach toward a streaming lower bound. Fix an input distribution and apply~\cite{LNW14} to obtain an $s \times n$ sketching matrix. Then, correctness under the distribution implies that frequency vectors $x$ that “fool” the sketch (e.g., those supported only on $N(v_i) \cup \{ u_i \}$ with $x_i\neq 0$) must not sketch to $0$. A duality argument (Lemma~\ref{lem:build-min-rank}) then yields a rank-$s$ matrix satisfying the minrank conditions. The minrank lower bound~\cite{ABGMM20} (which holds for arbitrary fields, including $\R$) implies a lower bound on $s$.

However, the quantitative parameters in~\cite{LNW14} introduce a circular dependency that leads to the downfall of this approach. The duality argument requires the kernel of the sketch to avoid the convex hull of fooling inputs. Since streaming algorithms operate on integer vectors, a “continuous-to-discrete” argument shows that the convex hull intersection must contain an integer point (Lemma~\ref{lem:build-min-rank}). However, the coordinates of this point, which is the solution to a system of $s$ linear equations, may be as large as $\exp(s)$. To apply the duality argument and rule out these integer-valued fooling vectors, the input distribution must support vectors with entries up to $\exp(s)$. But the logarithmic loss in~\cite{LNW14} means the dimensionality of the sketch, $s$, will suffer a loss least logarithmic in $\exp(s)$, which leads to the circularity and gives a vacuous lower bound.

\paragraph{Revisiting the Reduction in~\cite{LNW14}.} The reduction from streaming algorithms to linear sketches proceeds in two steps. The first step reduces dynamic streaming algorithms to \emph{path-independent} algorithms~\cite{G08, LNW14}, whose memory contents depend only on the final frequency vector, not the update sequence. This step is lossless (Theorem~\ref{thm:lnw}) and is the starting point of our lower bound argument. Path-independent algorithms have a natural linear structure over $\Z^n$: the \emph{kernel} $K$ (the set of frequency vectors mapping to the initial state) is closed under addition, i.e. $x, y \in K \Rightarrow x + y \in K$.\footnote{If $x, y \in K$, they individually map the algorithm to the initial state, then $x + y$ can first update $x$ (resulting in the initial state) and then update $y$ (again, going back to initial state), so $x+y \in K$.} 
While this linearity does not imply that path-independent algorithms are linear sketches, the second step in~\cite{LNW14} shows how to simulate them using linear sketches, at the cost of a logarithmic overhead in sketch dimension, and sketch entries that depend on the magnitude of the final frequency vectors.

Our first main technical result (Theorem~\ref{thm:lin-sketch-red}, proved in Section~\ref{sec:stream-to-sketch}) eliminates the logarithmic loss incurred in~\cite{LNW14} (and therefore avoids the aforementioned circularity) for a class of problems we call \emph{scale-invariant decision problems}. These are partial functions $g \colon \Z^n \to \{0,1,*\}$ satisfying $g(ax) = g(x)$ for all scalings $a\in \Z_{+}$ (Definition~\ref{def:scale-inv}). The diameter decision problem falls into this class, since scaling a frequency vector does not change its support. Our second main result (Theorem~\ref{thm: Diam linear sketch dim lb}, proved in Section~\ref{sec:ln-sketch}) gives a lower bound on the sketch dimension required to approximate the diameter on random bipartite graph metrics. The statement of Theorem~\ref{thm:lb} follows directly from combining Theorem~\ref{thm:lin-sketch-red} and Theorem~\ref{thm: Diam linear sketch dim lb}, both of which we outline below.

\paragraph{Overview of Theorem~\ref{thm:lin-sketch-red}: Streaming to Sketching for Scale-Invariant Problems.} 
In Section~\ref{sec:stream-to-sketch}, we prove Theorem~\ref{thm:lin-sketch-red}, which gives a reduction from dynamic streaming algorithms to linear sketches for scale-invariant problems. The dimensionality of the resulting sketch is at most $\calS^+(\calA, 1)$, the final space used by the streaming algorithm $\calA$ on frequency vectors in $\{0,1\}^n$ (see Definition~\ref{def:space}).

\begin{restatable*}{theorem}{linearsketchthm}\label{thm:lin-sketch-red} 
Suppose $\calA$ is a randomized streaming algorithm that computes a scale-invariant function $g$ over all dynamic streams with probability at least $1-\delta$. Then, for any distribution $\calD$ supported on $\Z^n\cap g^{-1}(\{0, 1\})$, there exists a matrix $T$ with the following guarantees:
\begin{itemize}
\item $T$ is an $s \times n$ integer matrix where $s \leq \calS^+(\calA, 1)$ and $\|T\|_{\infty} \leq O(n^{n/2})$.
\item There exists a function $\Dec \colon \Z^s \to \{0,1\}$ such that 
\begin{align}
\Prx_{\bx \sim \calD}\left[ \Dec(T\bx) = g(\bx) \right] \geq 1 - O(\delta). \label{eq:correctness}
\end{align}
\end{itemize} 
\end{restatable*}
The novelty of Theorem~\ref{thm:lin-sketch-red} is that the complexity of the sketch $T$ (i.e. its dimensionality and the size of its entries) is completely independent of the magnitude of the vectors in the support of $\calD$. To prove Theorem~\ref{thm:lin-sketch-red}, we start similarly to~\cite{G08,LNW14} with the (lossless) reduction to path-independent algorithms. Path-independent algorithms are specified by encoding and decoding functions $(\Enc, \Dec')$, where $\Enc \colon \Z^n \to \calW$ maps frequency vectors to a state\footnote{The function $\Enc \colon \Z^n \to \calW$ only depends on the frequency vector and not the history of updates, since it is path-independent.} and $\Dec' \colon \calW \to \{0,1\}$ produces the output (we use $\Dec'$ to distinguish from $\Dec$ in Theorem~\ref{thm:lin-sketch-red}). As mentioned, the kernel $K$ of $\Enc$, consisting of frequency vectors mapping to the initial state, is a submodule of $\Z^n$; if $K$ behaved like a subspace of $\R^n$, the linear map which projects to $K^{\perp}$ would give the desired matrix $T$. However, submodules of $\Z^n$ are not subspaces, so~\cite{LNW14} employs a ``structure theorem'' (Theorems~1 and 2 in~\cite{LNW14}) to characterize $K$ on all inputs within the required support. This increases the sketch dimension by a factor that is logarithmic in the magnitude of input vectors.

To bypass this logarithmic overhead, the proof of Theorem~\ref{thm:lin-sketch-red} sets $T$ to a matrix which projects vectors onto $\Span(K \cap \{0,1\}^n)^{\perp}$. The intersection $K \cap \{0,1\}^n$ gives the desired complexity of $T$ (i.e. dimensionality at most $\calS^+(\calA, 1)$ and entries that are at most exponential in $n$). The main challenge lies in guaranteeing correctness: the components of $x \in \Z^n$ along $\Span(K \cap \{0,1\}^n)$ are ``forgotten'' by $Tx$, and the fact that submodules are not subspaces means $\Span(K \cap \{0,1\}^n)$ may include vectors outside of $K$. This presents an obstacle to decoding $Tx$, because many $x' \in \Z^n$ with different encodings $\Enc(x')\neq \Enc(x)$ may ``collapse'' to the same sketch $Tx' = Tx$. In light of this ambiguity, we let $\Dec \colon \Z^s \to \{0,1\}$ be the next best thing: for $y \in \Z^s$, we consider all inputs $x \in \Z^n$ with $Tx = y$ and take a majority vote on $\Dec'(\Enc(x))$, where the ``majority'' is taken with respect to $\bx \sim \calD$. This alone does not work, but we have yet to take advantage of scale-invariance. By drawing a random scaling $\bZ$ from a suitable distribution over $\Z_+$, we use the scale-invariance of $g$ to show that, conditioned on $T\bx = y$, the majority of $\Dec'(\Enc(\bZ \bx))$ (for $\bx \sim \calD$) coincides with the majority of $g(\bx)$, whenever $(\Enc, \Dec')$ is correct. This gives the desired correctness guarantee of Theorem~\ref{thm:lin-sketch-red}. This step, showing the error of $\Dec$ does not increase even though states of $\Enc(\cdot)$ are ``collapsed'', appears in Subsection~\ref{sec:proof-of-ln-sketch} and is the most technically involved part of Theorem~\ref{thm:lin-sketch-red}.

\ignore{Nevertheless, we  by adapting the decoding function $\Dec$, that the additional error incurred is bounded.\footnote{The argument in~\cite{G08} also adapts decoding and avoids the logarithmic loss; however, it is specific to frequency estimation and only works for deterministic algorithms.} In particular, we define the output $\Dec(Tx)$ to be a ``majority vote'' on outputs produced by $(\Enc, \Dec')$ on inputs from $\calD$ which are consistent with 

We argue that any low-space randomized streaming algorithm $\calA$ computing a scale-invariant function $g \colon \Z^n \to \{0,1,*\}$ with error probability $\delta$ induces a low-dimensional linear sketch $T$ to compute $g$ with error probability $O(\delta)$ over inputs drawn from an arbitrary distribution $\calD \subset \Z^n \cap g^{-1}(\{0,1\})$. The matrix $T$ will be integer-valued, and unlike prior reductions~\cite{LNW14}, neither the dimension nor the size of its entries depends on the distribution $\calD$.

At a high level, our proof proceeds as follows. We begin by defining a new distribution $\calD'$ obtained by randomly rescaling each $\bx \sim \calD$ via $\bZ \sim \calZ$, where $\calM$ is a carefully designed distribution over integers. Since $g$ is scale-invariant, this rescaling preserves the function value. We then apply Theorem~\ref{thm:lnw} to convert $\calA$ into a deterministic, path-independent algorithm $\calB = (\Enc, \Dec)$ that computes $g$ over $\calD'$ with error probability at most $2\delta$ and space complexity no larger than that of $\calA$.

The algorithm $\calB$ defines a kernel $K$ of frequency vectors that map to the same memory state. Using Lemma~\ref{lem:basis-from-enc}, we construct a basis of $\R^n$ that separates this kernel from its orthogonal complement. The sketch matrix $T$ consists of the basis vectors spanning the orthogonal complement and therefore has dimension $s\leq \calS(\calA,1)$ and entry size $\norm{T}_\infty = O(2^n)$. We then define a decoding algorithm $\Dec'$ for our sketch using a majority vote over the outputs of $\Dec(\Enc(\bZ \bx))$. To analyze its error, we couple pairs of inputs sharing the same sketch and show that any disagreement must imply either a decoding error by $\calB$ or a violation of scale-invariance. This establishes that the linear sketch $T$ has error probability $O(\delta)$.}

\paragraph{Overview of Theorem~\ref{thm: Diam linear sketch dim lb}: Sketching Lower Bounds for Diameter.} 

We complete the proof of Theorem~\ref{thm:lb} by establishing Theorem~\ref{thm: Diam linear sketch dim lb}, which gives a lower bound on the dimension of linear sketches for diameter estimation. In particular, we consider the distribution $\calG(n, n, p)$ over random bipartite graphs with vertices $U = \{ u_1,\dots, u_n \}$ on the left-hand side, vertices $V = \{ v_1,\dots, v_n\}$ on the right-hand side, and each edge $(u_i, v_j)$ appearing independently with probability $p$. Our hard metric is given by the shortest path distance on $\bG \sim \calG(n,n,p)$, and our lower bound holds against linear sketches which compute the scale-invariant function $\diam_{\bG}^{2,c}$ (i.e. determine whether a pointset has diameter $\leq 2$ or $\geq 2c$), as defined formally in Definition~\ref{def: diameter}.

\begin{restatable*}{theorem}{linsketchlb}\label{thm: Diam linear sketch dim lb} 
    For any $c > 1$ and $n \in \N$, the following holds with high probability 
    over $\bG\sim\calG(n,n,p)$ where $p = p(n,c) \in (0,1)$. There exists a distribution $\calD=\calD(\bG)$ supported on $\Z^{2n}$ such that any linear sketch $(T, \Dec)$ where 
    \begin{itemize}
        \item $T$ is an $s \times 2n$ integer matrix with $\|T\|_{\infty} \leq O((2n)^{n})$, and
        \item The function $\Dec \colon \Z^s \to \{0,1\}$ satisfies \[ \Prx_{\bx\sim\calD}[\Dec(T\bx) = \diam^{2,c}_{\bG}(\bx)] \geq 1 - 1/(100 n),\]
    \end{itemize}
    must have dimension $s \geq \tilde{\Omega}(n^{1/(2\ceil{c}-1)})$.  
\end{restatable*} 
It is not too hard to show (Lemma~\ref{lem:necessary-properties}) that letting $k = \lceil c \rceil$ and $p = n^{-1+1/(2k-1)} /\log n$, a draw of $\bG \sim \calG(n,n,p)$ will simultaneously satisfy (i) there are $\Omega(n)$ indices $I^*(\bG)$ where $\sfd_{\bG}(u_i, v_i) \geq 2k+1$ and the neighborhood $\bN(v_i)$ is nonempty for all $i\in I^*(\bG)$, and (ii) the minrank of the directed knowledge graph corresponding to any $\Omega(n)$-sized induced subgraph of $\bG$ is at least $\tilde{\Omega}(np)$. In what follows, we design the distribution $\calD$, so that from the assumptions of Theorem~\ref{thm: Diam linear sketch dim lb}, we can violate the lower bound on the minrank. Roughly speaking, we construct $\calD$ as follows:
\begin{itemize}
    \item Draw $\bi \sim I^*(\bG)$, and let $\bx_{\bi}$ be zero or positive with equal probability (the specific choice of the positive value will be a subtle consideration)
    \item For $j\neq i$, the entry $\bx_j$ is positive whenever $u_j \in \bN(v_{\bi})$, and zero otherwise
\end{itemize}

When $\bx_{\bi} = 0$, the subset encoded by $\bx$ is contained within $\bN(v_{\bi})$ and has diameter at most $2$. If $\bx_{\bi} > 0$, the subset encoded by $\bx$ contains $u_{\bi}$ and a neighbor of $v_{\bi}$, and has diameter at least $2c$ by property (i). Thus, any linear sketch computing $\diam_{\bG}^{2,c}(\bx)$ must determine whether or not $\bx_{\bi} = 0$.

Consider a draw from $\calD$ for a fixed $\bi = i$. There is a natural set of (integer) vectors $z$ which may ``fool'' the sketch, namely those which satisfy $z_i \neq 0$, and $z_j = 0$ whenever $j\neq i$ and $u_j \notin \bN(v_i)$. If $Tz = 0$, then for $\bx_i=0$, the sketch will err on one of $\{ \bx, \bx+z \}$, because $T\bx = T(\bx + z)$ but $0=\bx_i\neq (\bx+z)_i$. On the other hand, if no such fooling vectors $z$ exist, we obtain a strong constraint on the columns of $T$: \[0 \neq Tz = z_i T^{i} + \sum_{u_j \in \bN(v_i)} z_j T^j.\] We would like to argue by duality that there exists $h_i \in \R^s$ for which $\langle h_i, T^i\rangle\neq 0$ and $\langle h_i, T^j\rangle = 0$ for $u_j \in \bN(v_i)$. By stacking vectors $h_i$ into an $n \times s$ matrix $H$, we would obtain a rank-$s$ matrix $HT$ satisfying the minrank condition for the knowledge graph induced by $I^*(\bG)$. Our lower bound would follow from the minrank lower bound (property (ii) above). 

But there is a slight nuance: the duality argument for finding $h_i$ requires that no real-valued ``fooling'' vectors $z$ exist. However, linear sketches (which operate on integer vectors) need not be correct on non-integer inputs. Separately, even if $z$ was an integer vector, $\bx + z$ might appear with negligible probability under $\calD$. The final argument (Claim~\ref{claim: forbidden-kernel-vectors} and Lemma~\ref{lem:build-min-rank}) addresses both of these issues and proceeds as follows. Suppose any (potentially real-valued) fooling vector $z$ satisfies $Tz = 0$. Then, since $T$ is an integer $s \times 2n$ matrix with entries $O((2n)^n)$, we can find an integer fooling vector $z$ whose entries have magnitude $\exp(O(sn \log n + s\log s))$ (coming from the determinant of an $s\times s$ matrix with bounded entries). If we knew $z_i$ exactly (and not just an upper bound), then we would define $\calD$ so that $\bx_i \in \{0, z_i\}$. Then, the mistake made on $\{\bx, \bx + z\}$ would contribute close to $1/(2n)$ towards the failure probability of the linear sketch, since $\bx+z$ remains in the support of $\calD$ with high probability. Recall, however, that $z_i$ depends on $T$, which in turn depends on $\calD$, so we cannot define $\calD$ in terms of $z_i$. This motivates the actual definition of $\calD$, where a non-zero $\bx_i$ is set to $\sfP = \exp(O(sn \log n + s\log s))!$. An integer fooling vector $z$, if one exists, has $z_i\leq \exp(O(sn \log n + s\log s))$, so the scalar $\chi = \sfP / z_i$ is an integer, and the inputs, $\bx$ and $\bx + \chi z$, are (i) indistinguishable to the sketch, (ii) require different outputs, and (iii) collectively appear often under $\calD$. This completes the duality argument, from which we conclude the statement of Theorem~\ref{thm: Diam linear sketch dim lb}.

\subsection{Organization of the Paper}

Section~\ref{sec: prelims} contains preliminary definitions that will be used throughout the paper, and Section~\ref{sec: streaming-defns} provides a formalization of streaming algorithms and space complexity. Section~\ref{sec:stream-to-sketch} presents a novel reduction from streaming algorithms to linear sketches for scale-invariant problems. In Section~\ref{sec:ln-sketch}, we prove a dimension lower bound on linear sketches for diameter estimation and conclude a streaming lower bound (Theorem~\ref{thm:lb}), the main result of our paper. Finally, in Section~\ref{sec: embedding-ubs}, we show a matching upper bound for diameter estimation in general metrics (Theorem~\ref{thm:ub}) and present an application to dimension-distortion tradeoffs for metric embeddings into $\ell_\infty$ (Theorem~\ref{thm:embed-lb}).

\ignore{
Note, we already conditioned on $\bi = i$, so in order to guarantee no fooling vector $z \in \R^s$ exists, we show that existence of a single such $z$ results in a constant

and produces a vector $\bx \in \Z^{2n}$ which is non-zero in every index $j$ for $u_j \in \bN(v_{\bi})$ (i.e., all vertices of $\bN(v_{\bi})$ are included in the multi-set), every other $j \neq \bi$ is zero (every vertex in $V \cup (U \setminus \bN(v_{\bi}) \setminus \{ u_{\bi} \})$ is not in the multi-set

We prove that for any $k\in\N$, linear sketches $T\in\Z^{s\times n}$ achieving a better-than-$k$ approximation to the diameter must have dimension $s\geq \tilde{\Omega}(n^{1/(2k-1)})$. Our lower bounds hold in metric spaces given by random bipartite graphs under the shortest path distance. In particular, consider $\bG=(U,V,\bE)$ with $|U|=|V|=n$, where each edge $(u_i,v_j)$ is included in $\bE$ independently with probability $p$ (see Definition~\ref{def:random-metric}). We will set $p \lesssim n^{-1 + 1/(2k-1)}$ so that with high probability, there is a set $I^\ast\subset[n]$ of size $\Omega(n)$ such that $\sfd_\bG(u_i,v_i)\geq 2k+1$ whenever $i\in I^\ast$. 

We define a distribution $\calD$ over inputs $\bx\in\Z^{n}$ which will encode subsets of $U$, the left side of the bipartition. We first sample $\bi\sim I^\ast$, which can be interpreted as a hidden index. We set $x_{j}$ to be a (random) positive number if $(u_j,v_{\bi})\in\bE$, encoding the inclusion of the neighborhood of $v_{\bi}$. Then, we set $x_{\bi}$ to be positive or zero with probability $1/2$, encoding the inclusion or exclusion of the vertex $u_{\bi}$. The idea is that if $x_{\bi} > 0$, then the subset encoded by $x$ contains both $u_{\bi}$ and a vertex adjacent to $v_{\bi}$, and therefore has diameter $\geq 2k$. However, if $x_{\bi} = 0$, then $x$ encodes a subset of the neighborhood of $v_{\bi}$, whose diameter is at most $2$. Informally, this means that any linear sketch hoping to approximate the diameter within a factor $k$ must recognize whether or not $x_{\bi} = 0$.

The above condition can be viewed as a restriction on the kernel of the linear sketch $T$. Fix $i \in I^\ast$ and suppose there is an integer vector $z\in\ker(T)$ supported on $\{i\} \cup \{ j : (u_j, v_i) \in \bE\}$ with $z_i\neq 0$. Since $z$ is an integer solution to a system of $s$ linear equations given by the rows of $T$, we can assume that its entries are at most exponential in $s$ and the magnitude of the entries of $T$. Now, consider inputs $x\in\supp(\calD)$ for which $x_i = 0$. Since $T(x+z) = Tx$ but $(x+z)_i \neq x_i = 0$, the linear sketch must fail on at least one of $\{x, x+z\}$. In order for this mistake to count towards the failure probability of $T$, we set the support of $\calD$ to be large enough (exponential in $s$ and the magnitude of the entries of $T$) so as to include $\bx + z$ with high probability over $\bx\sim\calD$. 

In other words, if $T$ succeeds with sufficiently high probability against $\calD$, then $\ker(T)$ avoids vectors $z\in\Z^n$ of the above form, which means the columns of $T$ indexed by $\{i\}\cup \{ j : (u_j, v_i)\in\bE \}$ are linearly independent. A duality argument then allows us to extract a hyperplane $h\in\R^s$ for which $(hT)_i = 1$ and $(hT)_j = 0$ for $\{ j : (u_j, v_i)\in\bE \}$. Repeating for every $i \in I^\ast$ gives a rank-$s$ matrix $M=HT$ that satisfies a minrank-like condition on the rows corresponding to $I^\ast$, essentially equivalent to the minrank condition for random graphs $\calG(n,1-p)$. A minor modification to the minrank lower bound in~\cite{ABGMM20} gives $s \geq \tilde{\Omega}(np) = \tilde{\Omega}(n^{1/(2k-1)})$.}



%% file: preliminaries.tex
\section{Preliminaries} \label{sec: prelims}

The following preliminaries consist of formal definitions that we will use throughout the paper. We begin with a formal definitions of the geometric estimation problems that we study, as well as the algorithmic primitive of $\ell_0$-sampling. Then, we give descriptions of (randomized) streaming algorithms, path-independent streaming algorithms. 

\subsection{Geometric Estimation Problems} 

Here, we present formal definitions for the geometric estimation problems (in finite metric spaces) studied in this paper. Both problems, the diameter and approximate furthest neighbor, are problems defined on multi-sets of a finite metric space. Importantly, both will fall in the category of \emph{scale-invariant} decision problems, i.e., those whose outputs does not change if one rescales the input by a positive integer (Definition~\ref{def:scale-inv}).

\begin{definition}[Diameter] \label{def: diameter}
Fix $n \in \N$ and let $\calM = ([n], \sfd)$ be a metric and $r > 0$ and $c>1$. The diameter function $\diam^{r,c}_{\calM} : \Z^n\to\{0,1,\ast\}$ is given by
\[\diam^{r,c}_{\calM}(x) = \begin{cases}
        1 & \textup{if $x \in \Z_{\geq 0}^{n}$ and } \max\left\{ \sfd(i,j) : x_i, x_j > 0 \right\} \geq cr \\
        0 & \textup{if $x \in \Z_{\geq 0}^n$ and } \max\left\{ \sfd(i,j) : x_i, x_j > 0 \right\} \leq r \\
        \ast & \textup{otherwise} 
    \end{cases}  \] 
\end{definition}

\begin{definition}[Approximate Furthest Neighbor] \label{def: AFN}
Fix $n \in \N$ and let $\calM = ([n], \sfd)$ be a metric and $r > 0$ and $c>1$. The approximate furthest neighbor function $\AFN^{r,c}_{\calM} : \Z^n\times[n]\to\{0,1,\ast\}$ is given by 
\[\AFN^{r,c}_{\calM}(x,q) = \begin{cases}
        1 & \textup{if $x \in \Z_{\geq 0}^{n}$ and } \max\left\{ \sfd(q,i) : x_i > 0 \right\} \geq cr \\
        0 & \textup{if $x \in \Z_{\geq 0}^n$ and } \max\left\{ \sfd(q,i) : x_i > 0 \right\} \leq r \\
        \ast & \textup{otherwise} 
    \end{cases}  \] 
\end{definition}

\subsection{Streaming Subroutine: $\ell_0$-sampling~\cite{JST11}} 

In Section~\ref{sec: embedding-ubs}, we will use the following streaming primitive for our upper bounds.

\begin{definition}[$\ell_0$-sampler~\cite{JST11}] \label{def: L0 sampler} 
For $\delta > 0$, a $\delta$-error $\ell_0$-sampler over $\Z^n$ is a randomized streaming algorithm with the following behavior: on input $x \in \Z^n$, it fails with probability at most $\delta$; conditioned on not failing, it either returns a uniformly random index $\bi \sim \supp(x)$, where $\supp(x) := \{i \in [n] : x_i \neq 0\}$, or declares $x = 0$ if $\supp(x) = \emptyset$.
\end{definition}

\begin{lemma}[Theorem 2 of~\cite{JST11}]\label{thm: L0 sampler}
For any $\delta > 0$, there exists a $\delta$-error $\ell_0$-sampler that uses $O(\log n \log(1/\delta) \log(mn))$ bits of space, where $m$ is an upper bound on the magnitude of the coordinates of the input vector. Moreover, the algorithm can be implemented as a linear sketch using an $s \times n$ matrix with $s = O(\log n \log(1/\delta))$ whose entries are in $\{-1,0,1\}$.
\end{lemma}

%% file: streaming-to-sketching.tex

\subsection{Randomized Streaming Algorithms and Path-Independent Algorithms} \label{sec: streaming-defns}
Throughout this subsection, we will use the notation $\interval{a,b} := \{a, a+1,\ldots, b\}$ for integers $a\leq b\in\Z$.

\paragraph{Randomized Streaming Algorithm.} We follow the standard formalization of (randomized) streaming algorithms from~\cite{G08, LNW14} for dynamic streams. One (minor) difference between the definition here and~\cite{LNW14} is that we differentiate between the working tape and the random tape;~\cite{LNW14} evaluates the space complexity by adding the randomness complexity, and we keep them separate. 


\begin{definition}[Dynamic Stream]
A dynamic stream $\sigma$ for $\Z^n$ of length $u \in \N$ is a sequence of $u$ updates of the form $\textsc{Add}(i)$ or $\textsc{Remove}(i)$ for $i \in [n]$, defining an implicit frequency vector $x \in \Z^n$. The vector $x$ is initialized to zero, an $\textsc{Add}(i)$ update increments $x_i$ by 1, while a $\textsc{Remove}(i)$ update decrements $x_i$ by 1.
\end{definition}

\begin{definition}[Randomized Streaming Algorithm]\label{def:randomized-streaming-alg}
A randomized streaming algorithm $\calA$ is specified by a Turing machine with state space $S$ and two distinguished states $s, t \in S$ denoting the start and end of an update's processing. The machine has access to three tapes:
\begin{itemize}
    \item \textbf{Input Tape}: A read-only, one-way tape containing a dynamic stream $\sigma$.
    \item \textbf{Working Tape}: A two-way tape storing memory from $\calW \subset \{0,1\}^*$.
    \item \textbf{Random Tape}: A read-only, two-way tape containing a fixed string $\rho \in \{0,1\}^{r}$ of random bits. The parameter $r$ is the randomness complexity of $\calA$ (see Definition~\ref{def:randomness}). 
\end{itemize}
A configuration after processing an update is a tuple $(w, \rho)$, where $w \in \calW$ is the working tape and $\rho$ is the random tape.

\textbf{Output Function:} There is a function $\Dec \colon \calW \times \{0,1\}^r \to \{0,1\}$ that produces the output from the final configuration.
\end{definition}

\paragraph{Correctness, Space Complexity, and Randomness Complexity.}
We now define the correctness, the space complexity of randomized streaming algorithms (that which we will lower bound), and the randomness complexity of a streaming algorithm. 
The computation starts with the input tape initialized to a dynamic stream $\sigma$ of length $u$; the work tape starts empty; the random tape $\rho$ is filled with $r$ uniformly random bits; and the state $q$ is initialized to $s$.
The algorithm proceeds in $u$ rounds where in each round $h = 1, \dots, u$:
\begin{itemize}
    \item The machine begins in state $s$, reads the $h$-th update of $\sigma$.
    \item It updates $w$ and the state $q$, while $\rho$ remains fixed.
    \item When it reaches state $t$, the processing of the update is complete, and the next round begins.
\end{itemize}

Each update defines a transition from $(w, \rho)$ to $(w', \rho)$. Following~\cite{G08, LNW14}, we denote this operation by $\oplus$: for example, $w \oplus e_i = w'$ indicates the result of applying $\textsc{Add}(i)$ to $(w,\rho)$. After processing all updates, the algorithm outputs $\Dec(w, \rho)$. Let $\calA(\sigma)$ denote this output.

\begin{definition}[Correctness]
For a function $g \colon \Z^n \to \{0,1,*\}$, a randomized streaming algorithm $\calA$ computes $g$ with error probability at most $\delta$ if, for any stream $\sigma$ with frequency vector $x \in \Z^n$,
\[
\Pr_{\rho}[g(x) \in \{0,1\} \text{ and }\calA(\sigma) \neq g(x) ] \leq \delta.
\]
\end{definition}

\begin{definition}[Space Complexity]\label{def:space}
For $m \in \N$, the space complexity of a streaming algorithm $\calA$ on (nonnegative) frequency vectors $x \in \interval{0,m}^n$ is
\[
\calS^+(\calA, m) = \log |\calW^+(\calA, m)|,
\]
where $\calW^+(\calA, m)$ is the set of all possible final working tapes $w$ that may result from processing any stream $\sigma$ with final frequency vector $x \in \interval{0,m}^n$ and any random tape $\rho \in \{0,1\}^r$.\footnote{In Appendix~\ref{sec:stream-to-path-ind}, we also (analogously) consider $\calW(\calA, m)$ to be the set of all final working tapes $w$ that may result from streams $\sigma$ with final frequency vector $x \in \interval{-m,m}^n$ and any random tape $\rho \in \{0,1\}^r$. Then, $\calS(\calA, m) = \log |\calW(\calA, m)|$, and since $\calW^+(\calA, m) \subset \calW(\calA, m)$, we have $\calS^+(\calA, m) \leq \calS(\calA, m)$.}
\end{definition}

We will lower bound $\calS^+(\calA, 1)$, which captures the number of bits needed to represent the final working tape on any stream with final frequency vector in $\{0,1\}^n$. Note that this is not the maximum memory used during computation as intermediate frequencies may be very large even if the final vector is in $\{0,1\}^n$. 
For example, an algorithm using a counter to maintain $\sum_{i=1}^n x_i$ would use $O(\log n)$ bits in the final configuration when $x\in\{0,1\}^n$, even though intermediate states use space complexity which grows with the magnitude of $x_i$'s. In a randomized streaming algorithm meant to work on arbitrarily long streams, one should think of the algorithm as dynamically allocating and de-allocating memory in the working tape. 

\begin{definition}[Randomness Complexity]\label{def:randomness}
The randomness complexity $r = r(n, m)$ of a randomized streaming algorithm $\calA$ is the maximum number of random bits that the algorithm uses over all streams with final frequency vector in $\interval{-m, m}^n$. Equivalently, the random tape $\rho$ of $\calA$ is assumed to lie in $\{0,1\}^{r(n,m)}$. While $r$ can depend arbitrarily on $n$ and $m$, it will be important that $r$ does not depend on the length of the stream. 
\end{definition}

\paragraph{Path-Independent Algorithms.} An important notion will be that of a (deterministic) path-independent algorithm, which we define below. Then, we state one of the main lemmas of~\cite{LNW14}, which shows how, given a distribution $\calD$, one may obtain a deterministic path-independent algorithm from a randomized streaming algorithm.  

\begin{definition}[Path-Independent Algorithm]\label{def:path-ind}
For any $n,m \in \N$, a deterministic, path-independent streaming algorithm $\calB$ with bound $m$ and dimension $n$ is specified by a pair of functions $(\Enc,\Dec)$ along with a set $\calW$ of states and transitions.
\begin{itemize}
    \item \textbf{Encoding.} $\Enc \colon \interval{-m,m}^n \to \calW$, maps frequency vectors to states.
    \item \textbf{Decoding.} $\Dec \colon \calW \to \{0,1\}$, maps a state to a final output.
    \item \textbf{Transitions.} Let $w \in \calW$ for which $\Enc^{-1}(w)$ contains a vector in $\interval{-(m-1), m-1}^n$. There is a well-defined transition function $\oplus$ that specifies, for each of the $2n$ possible stream updates $\xi e_i$ (with $\xi \in \{-1,1\}$ and $i\in[n]$), the new state of the algorithm: \[w' = w\oplus \xi e_i.\] One can view $\calW$ as vertices in a directed graph, with the edge $(w,w')$ labeled by $\xi e_i$.
\end{itemize}
$\calB$ is called \emph{path-independent} if it satisfies the following property:
\begin{itemize}
\item \textbf{Path-Independence}. For all $x \in \interval{-(m-1),m-1}^n$, $i \in [n]$, and $\xi \in \{-1,1\}$, one has \[\Enc(x) \oplus \xi e_i = \Enc(x+\xi e_i).\] 
\end{itemize}

Two relevant properties of streaming algorithms are space complexity and correctness.
\begin{itemize}
\item \textbf{Space Complexity}. The space complexity of $\calB$, and the nonnegative space complexity of $\calB$ are specified functions $\calS(\calB, \cdot), \calS^+(\calB, \cdot) \colon \N \to \R$, given by:
\begin{align*}
\calS(\calB, \ell) &= \log\left|\left\{ \Enc(x) \in \calW : x \in \interval{-\ell, \ell}^n \right\} \right| \qquad \calS^+(\calB, \ell) = \log\left|\left\{ \Enc(x) \in \calW : x \in \interval{0, \ell}^n \right\} \right|.
\end{align*}
\item \textbf{Correctness}. We say that $\calB$ computes a partial function $g \colon \interval{-m,m}^n \to \{0,1,*\}$ over a distribution $\calD$ supported on $\interval{-m,m}^n$ with probability at least $1-\delta$, if
\[ \Prx_{\bx \sim \calD}\left[ g(x) \in \{ 0,1 \} ~\mathrm{and}~\Dec(\Enc(\bx)) \neq g(\bx) \right] \leq \delta \]
\end{itemize}
\end{definition}

The first result in~\cite{LNW14} (Theorem~5) is a (lossless) reduction from randomized streaming algorithms to path-independent streaming algorithms. The theorem is stated to emphasize a generalization of the results from~\cite{LNW14} that we will need; namely, that the bound on the space complexity $\calS(\calB, \ell) \leq \calS(\calA, \ell)$ and $\calS^+(\calB,\ell) \leq \calS^+(\calA, \ell)$ holds for all $\ell \leq m$. We include a proof in Appendix~\ref{sec:stream-to-path-ind} (following the arguments in~\cite{LNW14}, and expanding on their exposition).

\begin{restatable*}{theorem}{lnwthm}\label{thm:lnw}
For any $n \in \N$, let $\calA$ be any randomized streaming algorithm with randomness complexity $r \in \N$ that computes a function $g \colon \Z^n \to \{0,1, *\}$ with error at most $\delta$, and let $\calD$ be a finitely-supported distribution on $\Z^n$. Then, for any large enough $m \in \N$, there exists a deterministic path-independent algorithm $\calB = (\calW, \Enc, \Dec, \oplus)$ with bound $m$ such that:
\begin{itemize}
\item The space complexity satisfies $\calS(\calB, \ell) \leq \calS(\calA, \ell)$ and $\calS^{+}(\calB, \ell) \leq \calS^{+}(\calA, \ell)$ for all $\ell \leq m$. 
\item $\calB$ computes $g$ on $\calD$ with probability at least $1-8\delta$.
\end{itemize}
\end{restatable*}

\subsection{A Special Basis from a Path-Independent Algorithm}

Theorem~\ref{thm:lnw} allows one to define a distribution $\calD$ supported on frequency vectors in $\Z^n$ and, assuming a randomized streaming algorithm $\calA$, obtain a path-independent streaming algorithm $\calB$ for $\calD$ with bounded space complexity. The following lemma (also essentially appearing in~\cite{G08, LNW14}) will allow us to connect the space complexity of a path-independent algorithm $\calB$ with the sketching matrix we construct in Section~\ref{sec:stream-to-sketch}.

\begin{definition}[Zero-Set and the Spanning Set]\label{def:zero-and-spanning}
Let $\calB = (\calW, \Enc, \Dec, \oplus)$ be a deterministic path-independent algorithm with bound $m$ and dimension $n$. For any $\ell \leq m$, we define
\begin{itemize}
\item \textbf{Zero-Set}. The zero-set of $\Enc(\cdot)$ within $\interval{-\ell, \ell}^n$ is given by
\[ K(\ell) = \left\{ x \in \interval{-\ell, \ell}^n : \Enc(x) = \Enc(0) \right\}.\]
\item \textbf{Zero-Set Subspace}. The zero-set subspace $M(\ell) \subset \R^n$ is $M(\ell) = \Span(K(\ell))$, where the span is taken over $\R$.
\end{itemize}
We let $t(\ell) \in \N$ be the dimension of $M(\ell)$ and $b_1,\dots, b_{t(\ell)} \in K(\ell)$ be an arbitrary basis of $M(\ell)$. 
\end{definition}

In the language of~\cite{G08,LNW14}, for any $\ell \leq m$, the zero-set $K(\ell) \subset \interval{-\ell,\ell}^n$ is the intersection of the ``kernel'' of the encoding function $\Enc(\cdot)$ with the interval $\interval{- \ell, \ell}^n$. We choose to name this the zero-set (as opposed to referencing kernels or sub-modules) because Definition~\ref{def:path-ind} and the proof of Theorem~\ref{thm:lnw} only requires encoding be defined for vectors in $\interval{-m,m}^n$. The proof of Lemma~\ref{lem:basis-from-enc} appears in Appendix~\ref{sec:path-ind-to-matrix}.

\begin{lemma}\label{lem:basis-from-enc}
Let $\calB = (\calW, \Enc,\Dec, \oplus)$ be a deterministic path-independent algorithm. There exists a basis $b_1,\dots, b_n$ of $\R^n$ and an integer $t$ with $n - \calS^{+}(\calB, 1) \leq t \leq n$ such that:
\begin{itemize}
\item The first $t$ vectors $b_1,\dots, b_t \in K(1)$ span the zero-set subspace $M = M(1)$.
\item The remaining vectors $b_{t+1}, \dots, b_n$ span $M^{\perp}$, and are integer vectors with $\|b_{t+1}\|_{\infty}, \dots, \|b_{n}\|_{\infty} \leq t^{t/2}$. 
\end{itemize}
\end{lemma}

The benefit of path-independent algorithms will come from the following lemma, which will allow us to conclude that the encoding state of a frequency vector $x$ which is not in $K(1)$, but may be represented as a (bounded) integer combination of vectors in $K(1)$, has its encoding being $\Enc(0)$. This will be a crucial property used in the proof of Theorem~\ref{thm:lin-sketch-red}.

\begin{lemma}\label{lem:path-ind-prop}
    Let $\calB = (\calW, \Enc,\Dec,\oplus)$ be a deterministic path-independent algorithm of bound $m$ and dimension $n$. For any $x \in \Z^n$ which can be represented as
    \[ x = \sum_{i=1}^t \alpha_i b_i,\]
    with $\alpha \in \Z^n$ and $t |\alpha_i| \leq m-1$ for the basis $b_1, \dots, b_t$ of $M(1)$ from Lemma~\ref{lem:basis-from-enc}, $\Enc(x) = \Enc(0)$.
\end{lemma}

\begin{proof}
    We note that we may decompose the vector $x$ as a sum of vectors from the standard basis by letting
    \[ x = \sum_{i=1}^t \sum_{\ell_1=1}^{|\alpha_i|} \sum_{j=1}^n \sum_{\ell_2=1}^{|b_{ij}|} \sign(\alpha_i b_{ij}) \cdot e_j,\]
    where $b_{ij}$ is the $j$-th coordinate of $b_i$, and $e_j$ is the standard basis vector. Thus, using path-independence of $\calB$, as long as $m-1 \geq |\alpha_i| \cdot t$, all partial sums above lie in $\interval{-(m-1), m-1}^n$, and we may conclude
    \[ \Enc(0) \oplus \bigoplus_{i=1}^t \bigoplus_{\ell_1=1}^{|\alpha_i|}\left( \bigoplus_{j=1}^n \bigoplus_{\ell_2=1}^{|b_{ij}|} \sign(\alpha_i b_{ij}) \cdot e_j \right) = \Enc(x).\]
    However, for each fixed $i$ and $\ell_1$, the updates $\sign(\alpha_i b_{ij}) \cdot e_j$ for each $j\in[n]$ (repeated $|b_{ij}|$ times) results in the frequency vector $\sign(\alpha_i) \cdot b_{i}$. We claim that both $\Enc(b_i) = \Enc(-b_i) = \Enc(0)$, so that we may inductively conclude $\Enc(x) = \Enc(0)$. The claim that $\Enc(b_i) = \Enc(0)$ is immediate from $b_i \in K(1)$. The fact that $\Enc(-b_i) = \Enc(0)$ follows from writing:
    \begin{align*} 
    \Enc(0) &= \Enc(0) \oplus \left(\bigoplus_{j=1}^n \bigoplus_{\ell=1}^{|b_{ij}|} \sign(b_{ij}) e_j \right) \oplus \left(\bigoplus_{j=1}^n \bigoplus_{\ell=1}^{|b_{ij}|} -\sign(b_{ij}) e_j \right) \\
            &= \Enc(b_i) \oplus \left(\bigoplus_{j=1}^n \bigoplus_{\ell=1}^{|b_{ij}|} -\sign(b_{ij}) e_j \right) = \Enc(0) \oplus \left(\bigoplus_{j=1}^n \bigoplus_{\ell=1}^{|b_{ij}|} -\sign(b_{ij}) e_j \right) = \Enc(-b_i).
    \end{align*}
\end{proof}

\section{Streaming Algorithms to Sketching Matrices}\label{sec:stream-to-sketch}

Our first main theorem shows how to extract a low-complexity sketching matrix from a randomized streaming algorithm, for the class of {\em scale-invariant} promise problems (i.e., problems whose answer does not change if one rescales the frequency vector).

\begin{definition}\label{def:scale-inv}
A problem specified by a partial function $g \colon \Z^n \to \{0,1, *\}$ is scale invariant if $g(a x) = g(x)$ for all $x \in \Z^n$ and $a\in \Z_{+}$. 
 \end{definition}

\linearsketchthm

The matrix may be used as a linear sketch for solving instances sampled from $\calD$ with high probability.
We note that in a crucial contrast to a similar transformation given in~\cite{LNW14}, the complexity of the matrix $T$ we find (i.e., the dimensionality and granularity of the entries of $T$) does not depend on the support of the frequency vectors in $\calD$. The remainder of this section is devoted to the proof of above theorem. 

\ignore{We start with some technical preliminaries.

\subsection{Preliminaries on Streaming Algorithms}
\enote{this section may be able to be replaced with appropriate pointers to the appendix.}
The proof of our theorem will crucially rely on an argument of~\cite{LNW14}, which shows how to transform a randomized streaming algorithm to a deterministic path-independent algorithm for any fixed distribution. We present a self-contained proof of this statement in Appendix~\ref{sec:stream-to-path-ind} and Appendix~\ref{sec:path-ind-to-matrix}. We present the relevant definitions and results here, deferring all proofs to the Appendix.

\begin{definition}[Path-Independent, abridged version of Definition~\ref{def:path-ind}]\label{def:path-ind-abridged}
A (deterministic) streaming algorithm is \emph{path-independent} if the final working tape depends only on the frequency vector $x$, not on the specific stream $\sigma$. That is, there exists:
\begin{itemize}
    \item \textbf{Encoding:} $\Enc \colon \Z^n \to \calW$, mapping frequency vectors to final work tapes.
    \item \textbf{Decoding:} $\Dec \colon \calW \to \{0,1\}$, producing the final output, given a work tape configuration.
\end{itemize}
Moreover, for all $x \in \Z^n$, $i \in [n]$, and $w \in \{-1,1\}$, we have $\Enc(x) \oplus we_i = \Enc(x + we_i)$.
\end{definition}

\begin{theorem}[Abridged version of Theorem~\ref{thm:lnw}]\label{thm:lnw-preview}
Let $\calA$ be any randomized streaming algorithm with randomness complexity $r \in \N$ that computes a function $g \colon \Z^n \to \{0,1, *\}$ with error at most $\delta$, and let $\calD$ be a finitely-supported distribution on $\Z^n$. Then, for any large enough $m \in \N$, there exists a deterministic path-independent algorithm $\calB = (\calW, \Enc, \Dec, \oplus)$ with bound $m$ such that:
\begin{itemize}
\item The space complexity satisfies $\calS^{+}(\calB, \ell) \leq \calS^{+}(\calA, \ell)$ for all $\ell \leq m$. 
\item $\calB$ computes $g$ on $\calD$ with probability at least $1-8\delta$.
\end{itemize}
\end{theorem}


\paragraph{Kernel of a Path-Independent Algorithm.}
Following~\cite{G08, LNW14}, the space complexity of a path-independent algorithm is governed by the kernel of the encoding map. Let $\calB = (\Enc, \Dec)$ be a deterministic path-independent algorithm. Then we define its kernel as:
\[
K = \{x \in \Z^n : \Enc(x) = \Enc(0)\},
\]
which forms a submodule of $\Z^n$ (namely, for any vectors $x, y \in K$ and any scalars $a, b \in \Z_{\geq 0}$ have $a x + by \in K$). For any $m \in \N$, it is shown in Lemma 4 of~\cite{G08} that
\[
\dim(\Span(K \cap \interval{0,m}^n)) \geq n - \frac{\calS(\calB, m)}{\log_2(m+1)}.
\]

In particular, if the space $\calS(\calB, m)$ of the streaming algorithm is {\em small}, the kernel has a {\em large} dimension.

\begin{lemma}[Same as Lemma~\ref{lem:basis-from-enc}]\label{lem:basis-from-enc-preview}
Let $\calB = (\calW, \Enc,\Dec, \oplus)$ be a deterministic and path-independent algorithm. There exists a basis $b_1,\dots, b_n$ of $\R^n$ and an integer $t$ with $n - \calS^{+}(\calB, 1) \leq t \leq n$ such that:
\begin{itemize}
\item The first $t$ vectors $b_1,\dots, b_t$ span the zero-set subspace $M = M(1)$ with $\ell=1$.
\item The remaining vectors $b_{t+1}, \dots, b_n$ span $M^{\perp}$, and are integer vectors with $\|b_{t+1}\|_{\infty}, \dots, \|b_{n}\|_{\infty} \leq t^{t/2}$. 
\end{itemize}
\end{lemma}}

\subsection{Proof of Theorem~\ref{thm:lin-sketch-red}} \label{sec:proof-of-ln-sketch}

We will now prove Theorem~\ref{thm:lin-sketch-red} by showing that any streaming algorithm for computing a scale-invariant function can be simulated by a linear sketch of comparable dimension that satisfies the guarantees of the theorem. At a high-level, our proof proceeds by first lifting the given input distribution via random scaling and then applying Theorem~\ref{thm:lnw} to obtain a path-independent algorithm. We then construct a sketching matrix using the basis obtained from Lemma~\ref{lem:basis-from-enc}, and this ensures that the kernel of the matrix contains the zero-set (Definition~\ref{def:zero-and-spanning}). The main technical step is defining a decoder that operates on the image of the sketch and essentially maintains the correctness guarantee of the streaming algorithm. Towards this end, we define a suitable majority-based decoder, and utilize scale-invariance and a coupling argument to prove that the error probability of our sketch is similar to that of the original streaming algorithm.

We start by defining a randomized distribution over scalings of frequency vectors, which allows us to leverage scale-invariance in the proof. Let $\sfR, \sfL \in \N$ be parameters (which should be read as very large numbers) to be set later. Let $P \subset \N$ denote the set of all primes less than $\sfR$. Define the distribution $\calZ$ over $\Z$ as follows:
\begin{itemize}
    \item For each $p \in P$, sample $\boldeta(p) \sim \{0, \dots, \sfL\}$ independently;
    \item Define $\bZ = \prod_{p \in P} p^{\boldeta(p)}$.
\end{itemize}
Let $\calD'$ be the distribution over $\Z^n$ obtained by sampling $\bx \sim \calD$ and $\bZ \sim \calZ$, and setting $\bz = \bZ \cdot \bx$. Furthermore, up to an (negligible) additional error probability, we may assume that $\calD$ is finitely supported and let $m_0$ such that $\supp(\calD) \subset \interval{-m_0, m_0}^n$. In the case $\supp(\calD)$ is unbounded, we set $m_0$ so as to ensure a $(1-\delta)$-fraction of $\calD$ lies in $\interval{-m_0,m_0}^n$ and incur an additional $\delta$ error for cases when $\bx \sim \calD$ does not lie in $\interval{-m_0,m_0}^n$.

Apply Theorem~\ref{thm:lnw} to the scale-invariant function $g$ and the distribution $\calD'$, for any large enough bound $m$. We emphasize that $m$ should be treated as a large number which we specify later; the fact that Theorem~\ref{thm:lnw} was proven for any large enough $m$ means it is helpful to consider this as a yet-unset parameter which can be set to be arbitrarily large. This yields a deterministic, path-independent algorithm $\calB = (\calW, \Enc, \Dec, \oplus)$ with a large enough bound $m$ and dimension $n$ such that:
\begin{align}
\Pr_{\bz \sim \calD'}[\Dec(\Enc(\bz)) = g(\bz)] \geq 1 - 8\delta, \label{eq:guarantee} 
\end{align} and $\calS^+(\calB, \ell) \leq \calS^{+}(\calA, \ell)$ for all $\ell \leq m$, and in particular, $\calS^+(\calB, 1) \leq \calS^+(\calA, 1)$.

By Lemma~\ref{lem:basis-from-enc}, let $b_1,\dots,b_n$ be a basis of $\R^n$ with the first $t \geq n - \calS^+(\calB,1)$ vectors spanning $M = M(1)$ from $K(1) \subset \interval{-1,1}^n$, and the remaining vectors spanning $M^\perp$ (Definition~\ref{def:zero-and-spanning}). Define $T \in \Z^{s \times n}$ to be the matrix whose rows are the last $s = n-t \leq \calS^+(\calB,1)$ vectors $b_{t+1}, \dots, b_n$. Then:
\begin{itemize}
    \item $T$ is an integer matrix with $\|T\|_{\infty} \leq O(n^{n/2})$, since $t \leq n$;
    \item The kernel of $T$ is $M$.
\end{itemize}

We now define the decoding function $\Dec' \colon \R^s \to \{0,1\}$ to satisfy Equation~\eqref{eq:correctness}. It suffices to define $\Dec'$ on the image of $T$ restricted to $\supp(\calD)$, since correctness is only required over $\calD$. For any $y$ in the image of $T$, define:
\[
\Dec'(y) = 
\begin{cases}
1 & \text{if } \mathbb{E}_{\bx \sim \calD}\left[ \Pr_{\bZ \sim \calZ}[\Dec(\Enc(\bZ \bx)) = 1] \mid T\bx = y \right] > \mathbb{E}_{\bx \sim \calD}\left[ \Pr_{\bZ \sim \calZ}[\Dec(\Enc(\bZ \bx)) = 0] \mid T\bx = y \right], \\
0 & \text{otherwise}.
\end{cases}
\]
In other words, $\Dec'(y)$ returns the majority outcome of $\Dec(\Enc(\bZ \bx))$ conditioned on $T \bx = y$. Let $\gamma$ denote the probability over $\bx \sim \calD$ that $\Dec'(T\bx) \neq g(\bx)$. Our goal now is to bound $\gamma$ in terms of $\delta$.
Since $\Dec'$ is a function of $T\bx$, we can equivalently express $\gamma$ by conditioning on the value of $y = T\bx$. Let $\calY$ be the distribution over $\Z^s$ induced by $\calD$ obtained by setting $y = T\bx$ for $\bx \sim \calD$. Then we may write:
\[
\gamma = \mathbb{E}_{y \sim \calY} \left[
\Pr_{\bx \sim \calD \mid T\bx = y}[\Dec'(\by) \neq g(\bx)]
\right],
\]
where the key will be to switch the order in how one samples. Namely, the above sampling experiment should be read as: sample $y$ in the image of $T$ on inputs drawn from $\calD$, and then, draw an input $\bx\sim \calD$ conditioned on it resulting in the output $y = T \bx$.

To analyze $\gamma$ in terms of the error probability of the deterministic path-independent algorithm $\calB = (\calW,\Enc, \Dec,\oplus)$, we define a coupling distribution $\calC(y)$ over pairs $((\bZ_1, \bx_1), (\bZ_2, \bx_2)) \in (\Z \times \Z^n)^2$ such that:
\begin{itemize}
    \item Both $\bx_1$ and $\bx_2$ satisfy $T\bx_1 = T\bx_2 = y$;
    \item For each $k \in \{1,2\}$, the marginal distribution of $(\bZ_k, \bx_k)$ is given by $\bx_k \sim \calD$ and $\bZ_k \sim \calZ$, conditioned on $T\bx_k = y$;
    \item For at least one choice of $k \in \{ 1,2 \}$, we have $\Dec(\Enc(\bZ_k \bx_k)) = \Dec'(y)$.
\end{itemize}
We emphasize that the choice of letting $\Dec'(y)$ be the ``majority rule'' is precisely so that such a coupling distribution $\calC(y)$ is guaranteed to exist. A distribution $\calC(y)$ of this form could be constructed in the following way: we consider all pairs $(Z, x)$ where $x \in \supp(\calD)$ and $Z \in \supp(\calZ)$ with $Tx = y$, and we let $\pi$ be any ordering of these pairs where pairs with $\Dec(\Enc(Z \cdot x)) = \Dec'(y)$ come before those with $\Dec(\Enc(Z \cdot x)) \neq \Dec'(y)$; a draw of the pair $(\Z \times \Z^n)^2$ is then given by sampling a real number $r \in [0,1]$ then taking the pairs corresponding to $r$ from the CDF of $\pi$ and the reverse order of $\pi$.

Using this coupling, and noting that the marginal distribution of $\bx_1$ under $\calC(y)$ matches the conditional distribution $\calD \mid T\bx = y$, we can equivalently write:
\[
\gamma = \mathbb{E}_{y \sim \calY}
\left[
\Pr_{((\bZ_1, \bx_1), (\bZ_2, \bx_2)) \sim \calC(y)}[\Dec'(y) \neq g(\bx_1)]
\right] = \mathbb{E}_{y \sim \calY}
\left[
\Pr_{((\bZ_1, \bx_1), (\bZ_2, \bx_2)) \sim \calC(y)}[\Dec'(y) \neq g(\bZ_1 \bx_1)]
\right];
\]

where the second equality uses the fact that $g$ is scale-invariant. Now whenever $((\bZ_1, \bx_1), (\bZ_2, \bx_2)) \sim \calC(\by)$ has $\Dec'(\by) \neq g(\bZ_1\bx_1)$, at least one of the following cases must occur:
\begin{itemize}
\item \textbf{Case} (a): $\Dec(\Enc(\bZ_1 \bx_1)) = \Dec'(\by)$, which means $\Dec(\Enc(\bZ_1 \bx_1)) \neq g(\bZ_1 \bx_1)$.
\item \textbf{Case} (b): $\Dec(\Enc(\bZ_1 \bx_1)) = g(\bZ_1 \bx_1) \neq \Dec'(\by)$ and $\Dec'(\by) \neq g(\bZ_2 \bx_2)$. From the second property of $\calC(\by)$, $\Dec(\Enc(\bZ_2 \bx_2)) = \Dec'(\by)$, so $\Dec(\Enc(\bZ_2 \bx_2)) \neq g(\bZ_2 \bx_2)$.
\item \textbf{Case} (c): $\Dec(\Enc(\bZ_1 \bx_1)) = g(\bZ_1 \bx_1) \neq \Dec'(\by)$, $\Dec(\Enc(\bZ_2 \bx_2))= \Dec'(\by) = g(\bZ_2 \bx_2)$, and this means that $g(\bZ_1 \bx_1) \neq g(\bZ_2 \bx_2)$.
\end{itemize}
Thus, at least one of the three cases occurs with probability at least $\gamma/3$. 

\paragraph{(Easier) Cases (a) and (b).} In both cases, we observe that one of $\bZ_1 \bx_1$ or $\bZ_2 \bx_2$ is a failure point for the deterministic algorithm $\calB = (\calW, \Enc, \Dec,\oplus)$ on an input drawn from $\calD'$. In particular, if case (a) occurs with probability $\ge \gamma/3$, we have:
\[
\frac{\gamma}{3} \leq \Pr_{\bz \sim \calD'}[\Dec(\Enc(\bz)) \neq g(\bz)] \leq 8\delta,
\]
by the guarantee of Theorem~\ref{thm:lnw}.

Similarly, suppose that case (b) occurs with probability $\geq \gamma/3$. We again find that $\bz = \bZ_2 \bx_2 \sim \calD'$ is a point where $\calB$ errs, so the same bound applies:
\[
\frac{\gamma}{3} \leq \Pr_{\bz \sim \calD'}[\Dec(\Enc(\bz)) \neq g(\bz)] \leq 8\delta.
\]

Combining both, we get $\gamma \leq 24\delta$ if either case (a) or (b) occur with probability $\geq \gamma/3$.

\paragraph{(Harder) Case (c).} 

In this case, we have:
\[
\Dec(\Enc(\bZ_1 \bx_1)) = g(\bZ_1 \bx_1), \quad \Dec(\Enc(\bZ_2 \bx_2)) = g(\bZ_2 \bx_2), \quad \text{but } g(\bZ_1 \bx_1) \neq g(\bZ_2 \bx_2).
\]
Since $T\bx_1 = T\bx_2$, it follows that $\bx_1 - \bx_2 \in \ker(T) = M$, the kernel of $T$, which is spanned by $K(1)$. So $\bx_2 - \bx_1$ can be written as a (real) linear combination $\sum_{i=1}^t \alpha_i b_i$ of the $t$ basis vectors of $M$ (Definition~\ref{def:zero-and-spanning}), and recall these are all in $\interval{-1,1}^n$ and the encoding function evaluates each $\Enc(b_i) = 0$. In particular, letting $W$ be the $n \times t$ matrix whose columns are $b_1, \dots, b_t$, we have $\bx_2 - \bx_1 = W\alpha$ for some $\alpha \in \R^t$.

Furthermore, $\bx_1, \bx_2$ are integer-valued and all entries of $W$ are in $\interval{-1,1}^n$. We consider the following: (i) first, let $W'$ be a $t\times t$ matrix which selects a subset of $t$ linearly independent rows of $W$, (ii) we set the parameter $\sfD = \det(W')$, which is an integer bounded by $t!$ (since entries are from $\interval{-1,1}$) and is independent of $\bx_2 - \bx_1$, and (iii) by Cramer's rule, an integer $\alpha \in \Z^t$ where $\sfD (\bx_2 - \bx_1) = W \alpha$ exists, and where $|\alpha_i| \leq m_0 \cdot t!$ (recall, $m_0$ is such that $\bx_1, \bx_2 \in \interval{-m_0,m_0}^n$).

Using Lemma~\ref{lem:path-ind-prop}, this has the consequence that for any $a \in \Z$ where $|a| \leq (m-1)/(2t\cdot m_0 t!)$, 
\[ \Enc(a \sfD (\bx_2 - \bx_1)) = \Enc(0). \]
As long as $a \sfD \bx_1 \in \interval{-(m-1)/2, (m-1)/2}^n$, 
\begin{align*}
    \Enc(a \sfD \bx_2) &= \Enc(a \sfD(\bx_2 - \bx_1)) \oplus \left(\bigoplus_{j=1}^n \bigoplus_{\ell=1}^{|a \sfD \bx_{1j}|} \sign(a\sfD \bx_{1j}) e_j \right) \\
    &= \Enc(0) \oplus \left(\bigoplus_{j=1}^n \bigoplus_{\ell=1}^{|a \sfD \bx_{1j}|} \sign(a\sfD \bx_{1j}) e_j \right) = \Enc(a \sfD \bx_1).
\end{align*}
\ignore{so there exists a fixed integer $\sfD$ (depending only on $K(1)$) such that $\sfD(\bx_2 - \bx_1)$ is an integer linear combination of $n$ vectors in $K(1)$, where the coefficients of each of these vectors has magnitude at most .... The reason is that (irrespective of $\bx_2 - \bx_1$,) we may find a set $k \leq n$ linearly independent vectors in $K(1)$ which span $M$. Letting $W$ be the $n \times k$ Boolean matrix whose columns form a linearly independent set of vectors in $K(1)$, the assumption that $\bx_2 - \bx_1 \in M$ means that there exists $\alpha \in \R^k$ with $W \alpha = \bx_2 - \bx_1$. In particular, finding a set of $k$ linearly independent rows and finding the inverse of the $k \times k$ submatrix of $W$ would give a desired $\alpha \in \Q^k$, where every coordinate is an integer multiple of $\sfD$, where $\sfD$ is the determinant of the $k \times k$ submatrix with entries in $\{-1,0,1\}$, and has magnitude at most $k!$ (recall, the multilinear representation of the determinant). {\color{red}Erik: we should now ensure that we only apply path-independence on vectors of magnitude at most $m$ -- note, however, that $m$ does not matter, so we can choose it to be large enough once we figure out what the right value should be.} Therefore $\sfD (\bx_2 - \bx_1)$ is an integer linear combination of vectors in $K(1)$, and hence for every $a \in \Z$ such that $a \sfD (\bx_2 - \bx_1) \in \interval{-(m-1),m-1}^n$, we have $\Enc(a \sfD (\bx_2 - \bx_1)) = 0$.}
Thus, for any integer $a \in \Z$ where $|a| \leq (m-1) / (2t\cdot m_0 t!)$ and $a\sfD \bx_1 \in \interval{-(m-1)/2, (m-1)/2}^n$, we have $\Dec(\Enc(a \sfD \bx_1)) = \Dec(\Enc(a \sfD \bx_2))$, and scale-invariance implies:
\[
g(a \sfD \bx_1) = g(\bx_1), \quad g(a \sfD \bx_2) = g(\bx_2),
\]
which are distinct by the assumption for case (c). Therefore, for every $a \in \Z$ with $|a| \leq (m-1)/(2t\cdot m_0 t!)$ at least one of the two must be incorrect, that is, either $\Dec(\Enc(a \sfD \bx_1)) \neq g(a \sfD \bx_1)$ or $\Dec(\Enc(a \sfD \bx_2)) \neq g(a \sfD \bx_2)$.

We can now set $m$ (as a function of $\sfR, \sfL$) to be $2t\cdot m_0 \cdot t! \cdot \sfR^{\sfR\sfL} + 1$. This ensures that for every $a = \bZ' \sim \calZ$, we satisfy $|a| \leq (m-1) / (2t\cdot m_0 t!)$ and $a \sfD \bx_1 \in \interval{-(m-1)/2, (m-1)/2}^n$. Then by averaging over $a = \bZ' \sim \calZ$, we obtain:
\begin{align}
\frac{\gamma}{3} &\leq \Ex_{\bZ' \sim \calZ} \left[
\Pr_{\substack{y \sim \calY \\ \calC(y)}}[\Dec(\Enc(\bZ' \sfD \bx_1)) \neq g(\bZ' \sfD \bx_1)]
+ \Pr_{\substack{y \sim \calY \\ \calC(y)}}[\Dec(\Enc(\bZ' \sfD \bx_2)) \neq g(\bZ' \sfD \bx_2)]
\right] \nonumber \\
&\leq 2 \cdot \Pr_{\bz \sim \calD'}[\Dec(\Enc(\sfD \bz)) \neq g(\sfD \bz)]. \label{eq:almost}
\end{align}

We now bound the error probability in Equation~\eqref{eq:almost}. The difficulty is that $\sfD \bz$ does not follow the same distribution as $\calD'$, since $\bz \sim \calD'$ implies $\bz = \bZ \bx$ for $\bZ \sim \calZ$ and $\bx \sim \calD$, whereas here we are analyzing $\sfD \bz = \sfD \bZ \bx$. However, this is precisely where the design of the distribution $\calZ$ allows us to control the discrepancy between these two distributions. In particular, we will show that for appropriate choice of the parameters $\sfL, \sfR$, we can ensure that the scaling by $\sfD$ maps almost all of $\calD'$ back into itself, modulo a small {\em additive error} of $O(\delta)$.

Now to bound this source of error, let us consider the prime factorization of $\sfD = \prod_{p \in P(\sfD)} p^{\pi(p)}$. Then for $\sfD \bZ$ to lie within the support of $\calZ$, it suffices that, for each $p \in P(\sfD)$, the sampled exponent $\boldeta(p)$ in $\bZ$ satisfies $\boldeta(p) \leq \sfL - \pi(p)$. The probability that this fails is at most:
\[
\sum_{p \in P(\sfD)} \frac{\pi(p)}{\sfL},
\]
since each $\boldeta(p)$ is drawn uniformly from $\{0,\dots,\sfL\}$. Each $\pi(p) \leq O(\log \sfD)$ and $|P(\sfD)| \leq O(\log \sfD)$ as well, so the total failure probability is at most $O(\log^2 \sfD / \sfL)$.

In these boundary cases, we can no longer ensure that $\sfD \bz$ lies in the support of $\calD'$, and this is the only source of error in connecting the two distributions. We now choose $\sfL = \Omega(n^2 \log^2 n / \delta)$, and 
$\sfR \ge n! \ge \sfD$, to achieve an additive error of $O(\delta)$, as desired. This yields:
\[
\Pr_{\bz \sim \calD'}[\Dec(\Enc(\sfD \bz)) \neq g(\sfD \bz)] \leq \Pr_{\bz \sim \calD'}[\Dec(\Enc(\bz)) \neq g(\bz)] + O\left( \frac{\log^2 \sfD}{\sfL} \right).
\]

By the earlier guarantee on $\calB$ from Theorem~\ref{thm:lnw}, the first term is at most $8\delta$, and by our choice of  $\sfL$ and $\sfR$, the second term is at most $O(\delta)$. Hence we conclude that case (c) also satisfies:

\[
\gamma \leq 3 \cdot O(\delta) = O(\delta),
\]
completing the proof of Theorem~\ref{thm:lin-sketch-red}. 

%% file: sketch-lb.tex
\section{A Lower Bound for Linear Sketches Approximating Diameter}
\label{sec:ln-sketch}


We will show lower bounds for linear sketches that compute $\diam^{2,c}_{\calM}$ (Definition~\ref{def: diameter}) where $\calM$ is derived from the random graph metrics specified as follows.

\begin{definition}\label{def:random-metric} 
For $n \in \N$ and $p \in (0, 1)$, $\bG = (U,V,\bE) \sim \calG(n,n,p)$ is a random bipartite graph with $U=\{ u_1,\ldots,u_n \}$ and $V = \{ v_1,\ldots, v_n\}$, where every edge $(u_i,v_j)$ is included in $\bE$ independently with probability $p$. We denote the neighborhood of a vertex $v_i\in V$ as $\bN(v_i)\subset U$.

For convenience, we often associate the left-hand side, $U$ with the set of indices $[n]$, and $V$ with the set of indices $[n+1,2n]$. $\bG$ refers to both a graph and a metric space $([2n], \sfd_{\bG})$, where $\sfd_{\bG}(i,j)$ is the length of a shortest path between $i$ and $j$ in $\bG$, or $\infty$ if $i$ and $j$ are disconnected.
\end{definition} 

\linsketchlb

\subsection{Proof of Theorem~\ref{thm: Diam linear sketch dim lb}} \label{sec:proof-diameter}

We begin by stating the key properties that our random graph metric $\bG\sim\calG(n,n,p)$ from Definition~\ref{def:random-metric} will satisfy. We begin by defining the minrank of a graph, which will be a useful notion in our analysis.

\begin{definition}[Minrank~\cite{BBJK06, ABGMM20}]
    The real minrank of a graph on $m$ vertices is the smallest rank of an matrix $M\in\R^{m\times m}$ with nonzero entries on the diagonal, and $M_{ij} = 0$ whenever $(i,j)$ is not an edge of the graph.
\end{definition}

To obtain our lower bound, it will be sufficient for the following properties to hold in our graph $\bG$. For the remainder of the section, it will be useful to fix an integer $k = \lceil c \rceil \leq \log n / (20\log \log n)$. Notice that for larger values of $c \geq \log n / (20\log \log n)$, the desired lower bound becomes vacuous. We will set a suitable $p = p(n, k)$ to derive the desired lower bound on graph metrics $\bG$.

\begin{lemma}\label{lem:necessary-properties}
For any $k \leq \log n / (20 \log \log n)$, set $p = n^{-1+1/(2k-1)}/\log n$ and let $\bG \sim \calG(n,n,p)$. Then, the following three properties hold with probability $1-o(1)$:
\begin{itemize}
    \item The set $I^*(\bG)$ of coordinates $i \in [n]$ where $\sfd_{\bG}(u_i, v_i) \geq 2k+1$ has size $|I^*(\bG)| \geq 0.9n$.
    \item For every $i \in [n]$, the neighborhood $\bN(v_i) \neq \emptyset$.
    \item Letting $\tilde{\bG} = (I^*(\bG), \tilde{\bE})$ be the (directed) graph where $(i, j) \in \tilde{\bE}$ whenever $(u_j, v_i) \notin \bE$, we have $\minrank(\tilde{\bG}) \geq \tilde{\Omega}(np)$.
\end{itemize}
\end{lemma}

\begin{proof}
    For the first bullet point, we first upper bound the probability that $\sfd_{\bG}(u_i,v_i)\leq 2k-1$ for any fixed $i \in [n]$ using a union bound; then, we lower bound $|I^*(\bG)|$ by Markov's inequality. In particular, fix $i$ and note that in order for $\sfd_{\bG}(u_i, v_i)$ to have a path of length at most $2k-1$, there must be some $\ell \leq k-1$ and two sequences $u^{(1)}, \dots, u^{(\ell+1)}$ and $v^{(1)}, \dots, v^{(\ell+1)}$ where $u^{(1)} = u_i$ and $v^{(\ell+1)} = v_i$ such that an edge connecting $u^{(j)}$ to $v^{(j)}$ always exists, and an edge connecting $v^{(j)}$ to $u^{(j+1)}$ always exists.  Hence, we can upper bound the probability that such a path exists by:
\begin{align*}
\sum_{\ell=1}^{k-1} n^{2\ell} \cdot p^{2\ell +1} = p \cdot \sum_{\ell=1}^{k-1} (np)^{2 \ell} \leq p\cdot \frac{(np)^{2k-2}}{1 - (1/np)^2} \leq 2 n^{2k-2} p^{2k-1} \leq \frac{2}{\log n}.
\end{align*}
Hence, $I^*(\bG)$ is expected to have size at least $n - o(n)$, and we obtain the desired bound by Markov's inequality. The second bullet point follows from a simpler bound. The probability that there exists some $i \in [n]$ where $\bN(v_i) = \emptyset$ is at most $n(1-p)^{n} \leq 1/n$ whenever $p \geq 2\log n / n$ (note, our setting of $p$ is significantly larger, so we expect $\bN(v_i)$ to have size $\Omega(n^{1/(2k-1)})$).

The final bullet point has its proof in Appendix~\ref{apx: minrank proof} and uses an argument of~\cite{ABGMM20}. Roughly speaking,~\cite{ABGMM20} lower bounds the minrank of a random (directed) graph $\calG(m,q)$ by $\Omega(m\log(1/q)/\log m)$. Their lower bound proceeds by a union bound over sub-matrices with a special structure (which they show are very unlikely to appear). For us, the only step is showing the argument can incorporate the fact we consider a sub-graph defined by vertices in $I^*(\bG)$ whose has size $m$ is at least $\Omega(n)$, and we set $q = 1-p$, so $\log(1/q) \geq p$.
\end{proof}

\ignore{
\begin{claim}\label{claim: far w.h.p.} 
For $k\in\N$ and $p = n^{1/(2k-1) - 1} / \log n$, let $\bG\sim\calG(n,n,p)$. Then, for any $i\in[n]$,
\[ \Prx_{\bG}\left[ \sfd_{\bG}(u_i, v_i) > 2k-1 \right] \geq 1 - o_n(1). \] 
Moreover, since $\bG$ is bipartite, $\sfd_{\bG}(u_i, v_i) \geq 2k+1$ with probability $1 - o_n(1)$. 
\end{claim} 

\begin{proof}
We upper bound the probability that $\sfd_{\bG}(u_i,v_i)\leq 2k-1$ by a union bound. In particular, in order for $\sfd_{\bG}(u_i, v_i)$ to have a path of length at most $2k-1$, there must be some $\ell \leq k-1$ two sequences $u^{(1)}, \dots, u^{(\ell+1)}$ and $v^{(1)}, \dots, v^{(\ell+1)}$ where $u^{(1)} = u_i$ and $v^{(\ell+1)} = v_i$ such that an edge connecting $u^{(j)}$ to $v^{(j)}$ always exists, and an edge connecting $v^{(j)}$ to $u^{(j+1)}$ always exists.  Hence, we can upper bound the probability that such a path exists by:
\begin{align*}
\sum_{\ell=1}^{k-1} n^{2\ell} \cdot p^{2\ell +1} = p \cdot \sum_{\ell=1}^{k-1} (np)^{2 \ell} \leq p\cdot \frac{(np)^{2k-2}}{1 - (1/np)^2} \leq 2 n^{2k-2} p^{2k-1} \leq \frac{2}{\log n}
\end{align*}
\end{proof}

\begin{lemma}\label{lem: many far dists}
    For $k\in\N$ and $p = n^{1/(2k-1) - 1} / \log n$, let $\bG\sim\calG(n,n,p)$, and define \[I^\ast(\bG) := \{i\in[n] : \sfd_{\bG}(u_i,v_i)\geq 2k+1\}.\] With probability $1-o_n(1)$ over the draw of $\bG$, it holds that $|I^\ast(\bG)| \geq 0.9n$.
\end{lemma}
\begin{proof}
    By Claim~\ref{claim: far w.h.p.}, we have $\Ex_{\bG}\left[n-|I^\ast(\bG)|\right] = o(n)$, so by Markov's inequality, $n-|I^\ast(\bG)| \leq 0.1n$ with probability $1-o_n(1)$.
\end{proof}

While $I^\ast(\bG)$ is defined with respect to the metric $\sfd_\bG$, the crux of our argument will relate it to a combinatorial property of $\bG$ resembling the minrank \cite{ABGMM20}.

\begin{definition}\label{def: minrank completion}
    Let $G = (U,V,E)$ be a bipartite graph with $|U|=|V|=n$. For $1\leq t\leq n$, we say that a matrix $M\in\R^{n\times n}$ is a $t$-completion of $G$ if there is a subset $I\subset [n]$ of size $|I| = t$ for which: 
    \begin{itemize}
        \item $M_{ii} \neq 0$ for every $i \in I$,
        \item $M_{ij} = 0$ for every $i\in I$ and $j\in[n]$ with $(u_j,v_i)\in\bE$
    \end{itemize}
    
\end{definition}

\begin{lemma}[Minrank Lower Bound] \label{lem: minrank lower bound with bad rows}
    Let $n\in\N$ and $p\in[0,1]$. With probability $1 - o_t(1)$ over $\bG \sim \calG(n,n,p)$, any matrix $M=M(\bG)$ that is a $t$-completion of $\bG$ has
    \[ \rank(M) \geq \Omega\del{\frac{t p}{\log t}} \]
\end{lemma}
\begin{proof}
    See Appendix~\ref{apx: minrank proof}.
\end{proof}

As we will see, we require the following three properties from the graph $\bG \sim \calG(n,n,p)$: (i) the set $I^*(\bG)$ should have size at least $0.9n$ (Lemma~\ref{lem: many far dists}), (ii) the minrank of any $t$-completion of $\bG$ if at least $\tilde{\Omega}(tp)$, and (iii) the neighborhoods $\bN(v_i)$ are non-empty for all $i \in [n]$. This final property also holds with high probability from a simple union bound whenever $p \geq 2\log n / n$. In the remainder of the section, it is useful to consider any fixed bipartite graph $G = (U, V, E)$ from the support of $\calG(n,n,p)$ satisfying the three properties above.}


\paragraph{Description of $\calD(\bG)$.} We now describe the hard distribution $\calD(\bG)$ referred to in Theorem~\ref{thm: Diam linear sketch dim lb}, and we will show that a linear sketch cannot correctly compute $\diam^{2,k}_{\bG}(\bx)$ where $\bx$ is drawn from this distribution. We first define $\calD(G)$ for any fixed bipartite graph $G = (U, V, E)$ with $|U|=|V|=n$, and then pick $\bG$ satisfying the conditions of Lemma~\ref{lem:necessary-properties}. A draw from $\calD(G)$ is an input $x \in \Z^{2n}$ where each coordinate corresponds to a point in $G$. 
As it will turn out, it suffices to consider frequency vectors supported entirely on the first $n$ coordinates corresponding to the $n$ vertices $u_1,\dots, u_n\in U$. Specifically, for two large parameters $\sfP, \sfU \in \N$ (which we specify later), an input $\bx \sim \calD(G)$ is drawn according to the following procedure:
\begin{enumerate}
\item\label{en:step-1} We draw $\bi \sim I^*(G)$. Let $\bx_i = 0$ with probability $1/2$ and $\bx_i = \sfP$ with probability $1/2$. 
\item\label{en:step-2} Every $j \in [n]$ such that $(u_j, v_{\bi}) \in E$ has $\bx_j \sim [\sfP \cdot \sfU]$, and $\bx_j = 0$ if $(u_j, v_{\bi}) \notin E$.
\item The final $n$ coordinates $\bx_{n+1}, \dots, \bx_{2n}$ are all set to $0$.
\end{enumerate}
It is also useful to introduce the notation $(\bx, \bi) \sim \calD(G)$ which reveals the index $\bi$ sampled in Step~\ref{en:step-1}, and this allows one to connect our desired function $\diam^{2,k}_{G}(\cdot)$ to the distribution.

\begin{claim}\label{claim: index determines diameter}
    Let $G = (U, V, E)$ satisfy conditions of Lemma~\ref{lem:necessary-properties}. Then, for any $(x, i)$ sampled from $\calD(G)$, $\diam_G^{2,k}(x) = \ind\{ x_i > 0 \}$.
\end{claim}
\begin{proof}   
    Fix any $(x, i)$ drawn from $\calD(G)$ and let $X\subset [2n]$ denote the support of the vector $x$. If $x_i=0$, then $X= N(v_i)$ and therefore has diameter at most $2$. If $x_i > 0$, then $X = N(v_i)\cup \{u_i\}$. Choosing $u'\in N(v_i)$ arbitrarily, the diameter of $X$ is at least $\sfd_G(u_i,u')\geq \sfd_G(u_i,v_i)-  \sfd_G(v_i,u') \geq 2k$, where we used the fact that $i\in I^\ast(G)$.
\end{proof}

Given $\calD(G)$, we let $(T, \Dec)$ be an integer $s \times 2n$ matrix with $\|T\|_{\infty} \leq O((2n)^{n})$, and $\Dec\colon \Z^s \to \{0,1\}$ which we assume computes $\diam_{G}^{2,k}(\bx)$ with probability at least $1-1/(100n)$ over $\bx \sim \calD(G)$. Note that vectors $x \in \supp(\calD(G)) \subset \Z^{2n}$ are always zero in the last $n$ coordinates (those corresponding to vertices in $V$), so it suffices to consider the first $n$ columns of $T$, and we henceforth refer to $T$ as an $s \times n$ matrix. 

\begin{claim} \label{claim: forbidden-kernel-vectors} 
    Suppose $G$ satisfies the condition in Lemma~\ref{lem:necessary-properties}. Let $\sfQ := \sfP \cdot O(s^2 (2n)^{2n})^{s+1}$ and for every $i \in I^\ast$, we define $Z_i(\sfP)$ to be the set of vectors $z\in\Z^n$ satisfying:
    \begin{enumerate}[label=(\roman*)]
        \item $z_i = \sfP$ 
        \item $\norm{z}_\infty\leq \sfQ$
        \item $z_j = 0$ for all $j\in[n]$ for which $(u_j,v_i)\notin E$.
    \end{enumerate}
    When $\sfU = \omega(n (s^2 (2n)^{2n})^{s+1})$, we have that $\ker(T) \cap Z_i(\sfP) = \emptyset$. 
    
\end{claim} 
\begin{proof} 
    
    Fix an $i \in I^\ast(G)$ and let $z\in Z_i(\sfP)$. Suppose, for the sake of contradiction that $Tz = 0$. We will lower bound the failure probability of $(T, \Dec)$. We consider a draw of $\bx\sim \calD(G)$ conditioned on $\bi = i$ and $\bx_{\bi} = 0$ in Step~\ref{en:step-1} (which occurs with probability $1/(2n)$). Whenever $\bx_j \in \{ \sfQ + 1, \dots, \sfP \sfU - \sfQ\}$ for all $j \in [n]$ with $(u_j, v_i) \in E$, the fact $\|z\|_{\infty} \leq \sfQ$ implies $\bx + z$ lies in the support of $\calD(G)$. Thus, $\bx + z$ is drawn from $\calD(G)$ with the same probability as $\bx$ is. Namely, $i$ is chosen and $(\bx + z)_i = \sfP$ with probability $1/(2n)$, and the distribution is uniform over its support on the remaining coordinates. However, $Tz = 0$ implies $T(\bx + z) = T\bx$, which means that $\Dec(T(\bx + z)) = \Dec(T\bx)$. By Claim~\ref{claim: index determines diameter}, $\diam_{G}^{2,k}(\bx)\neq \diam_{G}^{2,k}(\bx + z)$, meaning $(\Dec, T)$ is incorrect on one of $\{\bx, \bx+z\}$. In other words, 
    \begin{align*}
    \Prx_{\bx \sim \calD(G)}\left[ \Dec(T\bx) \neq \diam_{G}^{2,k}(\bx) \right] \geq \frac{1}{2n} \cdot \left(1 - \frac{2\cdot \sfQ}{\sfP \cdot \sfU} \right)^n \geq \frac{1}{50 n}
    \end{align*}
    whenever $\sfU \gg n \cdot (s^2 (2n)^{2n})^{s+1}$. 
\end{proof} 

Crucially, the fact that $\ker(T)\cap Z_i(\sfP) = \emptyset$ for all $i \in I^\ast$ allows us to construct a low-rank matrix $M\in\R^{n\times n}$ which satisfies the conditions of the minrank of $\tilde{G}$ (the third bullet point in Lemma~\ref{lem:necessary-properties}). 

\begin{lemma}\label{lem:build-min-rank} 
Let $T$ be an integer $s \times n$ matrix satisfying $\| T\|_{\infty} \leq O((2n)^n)$ such that $\ker(T) \cap Z_i(\sfP) = \emptyset$ for all $i \in I^\ast$. If $\sfP := (O(s^{2} (2n)^{2n})^s)!$, there exists an $n\times n$ matrix $M$ of rank at most $s$ with ones on the diagonal and $M_{ij} = 0$ whenever $i \in I^*(G)$ and $(u_j, v_i) \in E$. 
\end{lemma} 
\begin{proof}
We will let $M = HT$ where we describe the $n \times s$ matrix $H$ by describing each of its rows $h_i \in \R^s$ for $i \in [n]$. If $i \notin I^*(G)$, we let $h_i$ be the $T^{(i)} / \|T^{(i)}\|_2^2$, where $T^{(i)}$ is the $i$-column of $T$, which satisfies the constraint that $M_{ii} = 1$.

On the other hand, suppose $i \in I^{*}(G)$, and let $A$ be the $s \times |N(v_i)|$ submatrix of $T$ which considers the columns $j \in [n]$ lying in the neighborhood $u_j \in N(v_i)$, and $W$ be the $s \times k$ matrix with $k \leq \min\{ s, |N(v_i)|\}$ which contains a maximal set of linearly independent columns of $A$. We will first show that there are no vectors $w \in \R^{|N(v_i)|}$ with $A w = T^{(i)}$. If such a vector $w$ existed, there would also exist $y \in \R^{k}$ which satisfies $Wy = T^{(i)}$, and because $T^{(i)}$ and $W$ have integer coordinates of magnitude at most $O((2n)^n)$, there exists an integer vector $\tilde{y} \in \Z^k$ which satisfies $W \tilde{y} = \sfD \cdot T^{(i)}$, where $\sfD := \det(W^{\t} W)$ is an integer of magnitude at most $O(s^{2} (2n)^{2n})^s$, and $\| \tilde{y} \|_{\infty} \leq \sfD \cdot O(s^2 (2n)^{2n})^{s+1}$. Thus, by setting of $\sfP = (O(s^{2} (2n)^{2n})^s)!$, the fact $\tilde{y}$ exists means $\sfP / \sfD$ is an integer, and $\tilde{y}' = \sfP / \sfD \cdot \tilde{y}$  satisfies $W \tilde{y}' = \sfP \cdot T^{(i)}$ and $\| \tilde{y}'\|_{\infty} \leq \sfP\cdot O(s^2 (2n)^{2n})^{s+1}$. Let $z \in \mathbb{Z}^n$ be defined as $z_i = \sfP$ and $z_j = -\tilde{y}_{j'}'$ for each $j$ such that $u_j \in N(v_i)$, where $\tilde{y}_{j'}'$ is the coordinate of $\tilde{y}'$ associated with the column of $T$ indexed by $u_j$.
This would imply the existence of a $z \in \ker(T)$ that is also in $Z_i(\sfP)$, which we have ruled out by Claim~\ref{claim: forbidden-kernel-vectors}.
Since no such vectors exists, this means that 
\[\left\{ T^{(i)} - A w : w \in \R^{|N(v_i)|} \right\}, \]
is a convex set that does not contain $0$. Therefore, there is a vector $\tilde{h} \in \R^s$ which satisfies $\langle \tilde{h}, T^{(i)}\rangle - \langle \tilde{h}, A w \rangle > 0$ for all $w \in \R^{|N(v_i)|}$. This implies, in particular, that $\langle \tilde{h}, T^{(i)}\rangle > \langle \tilde{h}, T^{(j)}\rangle \cdot \sfW$, for all $j \in N(v_i)$ and all $\sfW \in \R$. Note this can only happen if $\langle \tilde{h}, T^{(j)}\rangle = 0$. Rescaling $\tilde{h}$ by $1 / \langle \tilde{h}, T^{(i)}\rangle$ gives the desired vector $h_i$. 
\end{proof}

\begin{proof}[Proof of Theorem~\ref{thm: Diam linear sketch dim lb}]
    Fix $k = \ceil{c} \leq \log n / (20 \log \log n)$ and take $p = n^{-1+1/(2k-1)}/\log n$ (recall, if $c \geq \log n / (8 \log \log n)$ then the lower bound holds vacuously), and a draw $\bG\sim \calG(n,n,p)$ satisfying conditions of Lemma~\ref{lem:necessary-properties} which hold with probability $1-o(1)$. 
    Observe that in the graph $G$, all paths between $U$-side vertices have even length, since paths alternate between vertices in $U$ and $V$.
    So any $x \in \supp(\calD(\bG))$ which satisfies $\diam_{\bG}(x) \geq 2c$ will also satisfy $\diam_{\bG}(x) \geq 2k$, and this means that the linear sketch $(T, \Dec)$ must satisfy $\Dec(T\bx) = \diam^{2,k}_{\bG}(\bx)$ with probability $1 -1/(100 n)$. Therefore, we can apply Claim~\ref{claim: forbidden-kernel-vectors} and Lemma~\ref{lem:build-min-rank} and obtain an $n \times n$ matrix $M$ of rank at most $s$, with ones on the diagonal and which satisfies $M_{ij} = 0$ whenever $i \in I^*(\bG)$ and $(u_j, v_i) \in \bE$. Consider the sub-matrix $M'$ which considers rows and columns in $I^*(\bG)$; and note this satisfies conditions of the minrank on the knowledge graph $\tilde{\bG}$. The matrix $M'$ has rank at most $s$ (since it is a sub-matrix of $M$), and satisfies $M'_{ij} = 0$ whenever $(i, j) \notin \tilde{\bE}$ (equivalent to $(u_j, v_i) \in \bE$). By the third bullet point of Lemma~\ref{lem:necessary-properties}, $s$ is at least $\tilde{\Omega}(np)$, and we obtain the desired bound.
\end{proof}

\begin{proof}[Proof of Theorem~\ref{thm:lb}]
    Let $G$ satisfy the condition of Theorem~\ref{thm: Diam linear sketch dim lb} and let $\calA$ be a dynamic streaming algorithm computing a $c$-approximation for $\diam_G$ with success probability $1 - o(1/n)$ (note that we can afford to lose a $O(\log n)$ multiplicative factor to ensure the high success probability). The algorithm $\calA$ computes $\diam^{2,c}_{\calM}(x)$ for any stream with frequency vector $x$ with probability at least $1-o(1/n)$ (over the internal randomness of the algorithm), where $\calM$ is the shortest path metric induced by $G$. Applying Theorem~\ref{thm:lin-sketch-red} on the distribution $\calD(G)$, which is supported on $(\diam^{2,c}_{G})^{-1}(\{0,1\})$ by Claim~\ref{claim: index determines diameter}, we obtain a linear sketch $(T,\Dec)$ of dimension $s \leq \calS^+(\calA, 1)$ satisfying the conditions of Theorem~\ref{thm: Diam linear sketch dim lb}, which implies \[\calS^+(\calA,1) \geq s\geq \tilde{\Omega}(n^{1/(2\ceil{c} - 1)}).\]
\end{proof}

\subsection{Linear Sketching Lower Bound for $\AFN$} 

Theorem~\ref{thm: Diam linear sketch dim lb} proves a lower bound on the dimension of any linear sketch estimating the diameter on random graph metrics. We now show that, as a consequence, a similar lower bound holds for the related problem of approximate furthest neighbor. 

We will show lower bounds for computing $\AFN_{1,c}^{\calM}$ (Definition~\ref{def: AFN}) that hold for the same random graph metric $\bG=(U,V,\bE)\sim\calG(n,n,p)$ given in Definition~\ref{def:random-metric}, albeit with a slightly different setting of $p$. We will require two additional properties of $\bG$.

\begin{claim}\label{claim: neighborhood size bounds}
    For $n,k\in\N$ with $k\leq\frac{\log n}{20 \log \log n}$, set $p=n^{-1 + 1/(2k-1)}/\log^5 n$. With probability $1-o(1)$ over the draw of $\bG=(U,V,\bE)\sim\calG(n,n,p)$, the following properties hold: 
    \begin{enumerate}[label=(\roman*)] 
        \item $|\bN(v_i)| \geq pn/2$ for all $i \in [n]$ 
        \item $|\bN(v_i) \cap \bN(v_j)| \leq \log n\cdot \max\{np^2, 1\}$ for all distinct $i,j \in [n]$ 
    \end{enumerate} 
\end{claim}
\begin{proof} 
    We will show that each property holds with probability $1-o(1)$. Notice that for any fixed $i \in [n]$, the probability that $|\bN(v_i)| \leq pn/2$ is at most $e^{-pn/8} \leq 1/n^2$ by a Chernoff bound (this holds whenever $p\geq 16\log n / n$, which is true for our setting of $k$). Union bounding over $i\in[n]$, we conclude that (i) holds with high probability. 
    Next, take any distinct $i,j\in[n]$, and observe that the distribution of $|\bN(v_i)\cap \bN(v_j)|$ is precisely $\Bin(n,p^2)$. If $k=1$, then $np^2=\tilde{\Omega}(n)$ and standard concentration bounds imply $|\bN(v_i) \cap \bN(v_j)|\leq np^2 \log n$ except with probability $1/n^{10}$. If $k\geq 2$, then $np^2\ll n^{-0.1}$ and it can be directly shown that \[\sum_{\ell = \log n}^n \binom{n}{\ell}\cdot p^{2\ell}\cdot (1-p^2)^{n-\ell} < \sum_{\ell = \log n}^n (np^2)^\ell = n^{-\omega(1)}.\] In either case, taking a union bound over all pairs $i, j \in [n]$ concludes (ii). 
\end{proof} 

\begin{lemma}\label{lem:afn-lb}
Fix $k, n \in \N$, and let $c \leq 2k+1$. For $\bG \sim \calG(n,n,p)$ with $p = n^{1/(2k-1) - 1}/\log^5 n$, and $\calD = \calD(\bG)$ (from Subsection~\ref{sec:proof-diameter}), the following holds with probability $1-o(1)$: any linear sketch $(T, \Dec)$ which satisfies
\begin{itemize}
\item The matrix $T \in \Z^{s \times 2n}$ has $\|T\|_{\infty} \leq O((2n)^n)$, and
\item $\Dec \colon \Z^s \times [2n] \to \{0,1\}$ is an output function such that the probability over $\bx \sim \calD$ that there exists $q \in [2n]$ satisfying $\Dec(T\bx, q) \in \{0,1\} \setminus \{ \AFN^{1, c}_{\bG}(\bx, q)\}$ is at most $1/(200 n)$,
\end{itemize}
must have $s \geq \tilde{\Omega}(n^{1/(2k-1)})$.
\end{lemma}

\begin{proof}
Notice that the lower bound is vacuous when $k = \Omega(\log n / \log \log n)$. Thus, for the remainder of the proof, we may assume $k \leq \log n / 20 \log \log n$.  

The proof is via a direct reduction to Theorem~\ref{thm: Diam linear sketch dim lb} using $\ell_0$-samplers (Lemma~\ref{thm: L0 sampler}) and Claim~\ref{claim: neighborhood size bounds}. In particular, we design a linear sketch $(T', \Dec')$ which will satisfy the conditions of Theorem~\ref{thm: Diam linear sketch dim lb} with the approximation factor $c' := k$. The sketch $(T', \Dec')$ does the following:
\begin{itemize}
\item It initializes an assumed sketch $(T, \Dec)$ from the hypothesis of this lemma.
\item It initializes $L$ independent $\ell_0$-samplers from Lemma~\ref{thm: L0 sampler} for vectors in $\Z^{2n}$ and $\delta = 0.1$, for a parameter $L$ (which will be set to $100b\log n$ for $b :=\log n\cdot \max\{np^2, 1\} = o(np)$).
\end{itemize}
The matrix $T'$ is obtained by stacking the $s$ rows of $T$, as well as the $L \cdot O(\log n)$ rows of the $\ell_0$-samplers, for a total dimensionality of $s + O(L\log n)$. Note that the bound $\|T' \|_{\infty} \leq O((2n)^n)$ follows from the fact $\|T\|_{\infty} \leq O((2n)^n)$, and the $\ell_0$-samplers have entries in $\{-1, 0, 1\}$ (from Lemma~\ref{thm: L0 sampler}). The decoding function $\Dec' \colon \Z^{n} \to \{0,1\}$ proceeds as follows:
\begin{enumerate}
\item Extract from the $L$ many $\ell_0$-samplers, coordinates $\bell_1,\dots, \bell_{L} \in [2n] \cup \{ \bot\}$ such that any non-failing index $t \in [L]$ (i.e., one whose $\bell_t \neq \bot$) satisfies $\bx_{\bell_t} \neq 0$.
\item Second, we find the unique $i \in [n]$ such that $u_{\bell_1}, \dots, u_{\bell_L} \in N(v_{i}) \cup \{ u_{i} \}$, which we show exists as long as $| \{ \bell_1, \dots, \bell_L \} | \geq b+2$ from the description of $\calD$ and Claim~\ref{claim: neighborhood size bounds}. We output $\Dec(T\bx, v_{i})$.
\end{enumerate}
We now show that $\Dec'$ satisfies the conditions of Theorem~\ref{thm: Diam linear sketch dim lb}, that is, that the probability over $\bx \sim \calD$ that $\Dec'(T'\bx) = \diam_{\bG}^{2,c'}(\bx)$ is at least $1-1/(100n)$. In particular, we consider any draw of $\bx \sim \calD$ which was generated with $\bi$, and we first claim that, as long as $L = 100 b \log n$, the $\ell_0$-samplers return satisfy $|\{ \bell_1,\dots, \bell_{L}\}| \geq b + 2$ with probability $1 - o(1/n)$. This occurs for two reasons. First, (i) each $\ell_0$-sampler fails independently with probability at most $0.1$, so by a standard Chernoff bound, there are at most $L/2$ $\ell_0$-samplers which fail. Second, (ii) whenever at least $L/2$ many $\ell_0$-samplers succeed, each is distributed uniformly among indices $j \in [n]$ where $\bx_j \neq 0$; if $(\bx, \bi)$ was the input generated by $\calD$, the non-zero entries are contained within $\bN(v_{\bi}) \cup u_{\bi}$,  whose size is at least $np/2 = \omega(b)$.

Since $L/2$ $\ell_0$-samplers succeed and take on entries uniformly from a set of size $\omega(b)$, the probability $|\{ \bell_1, \dots, \bell_L \}| < b + 2$ is at most $o(1/n)$, and therefore, there are at least $b$ many indices among $\bell_1,\dots, \bell_L$ within $\bN(v_{\bi})$ (the additive $+2$ comes from $\bot$ and potentially $\bi$); when this occurs, the unique coordinate $i$ from the second step is exactly $\bi$, and $\AFN_{\calM}^{1,c}(\bx, v_{\bi}) = \ind\{ x_{\bi} > 0 \}$. In other words, we have:
\begin{align*}
\Prx_{(\bx, \bi) \sim \calD}\left[ \Dec'(T'\bx) = \diam_{\bG}^{2,c'}(\bx) \right] \geq \Prx_{(\bx, \bi) \sim \calD}\left[ \Dec(T\bx, v_{\bi}) = \AFN_{\bG}^{1,c}(\bx, v_{\bi})\right] - o(1/n) \geq 1 - 1/(100 n).
\end{align*}
The resulting lower bound implies $s + O(L\log n) \geq \tilde{\Omega}(n^{1/(2\lceil c'\rceil-1)})$, which is the same as $\tilde{\Omega}(n^{1 / (2k-1)})$ since $c' = k$ is already an integer. The conclusion that $s\geq \tilde{\Omega}(n^{1/(2k-1)})$ follows from $L\log n = O(b\log^2 n) = o(np/\log n)$, completing the proof.\footnote{Here, we are using the fact that the lower bound in Theorem~\ref{thm: Diam linear sketch dim lb} is exactly $\Omega(np / \log n)$ even for our new setting of $p = n^{1/(2k-1) - 1} / \log^5 n$}
\end{proof}

%% file: upper-bounds.tex
\section{Embedding Metrics in Low-Dimensional $\ell_{\infty}$} \label{sec: embedding-ubs}

In this section, we give algorithms for solving the approximate furthest neighbor and diameter problems in $\ell^{k}_\infty$. We then use embeddings of arbitrary metrics into low-dimensional $\ell_\infty$ to give a dynamic streaming algorithm for diameter estimation in general metrics, matching the lower bound of Theorem~\ref{thm:lb} up to constant factors in the exponent. By the same token, we show how our sketching lower bound for approximate furthest neighbor from Lemma~\ref{lem:afn-lb} translates to dimension-distortion tradeoffs for embeddings into $\ell_\infty$.

\subsection{Streaming Approximate Furthest Neighbor in $\ell_{\infty}^k$} 

We begin by generalizing the approximate furthest neighbor problem (originally defined in Definition~\ref{def: AFN}) to potentially infinite metric spaces. 

\begin{definition}[Approximate Furthest Neighbor in arbitrary metrics] 
Let $\calM = (Y, \sfd)$ be a metric space and fix an $n$-point submetric $U = \{u_1,\dots, u_n\}\subset \calM$ and parameters $r > 0$ and $c > 1$. The approximate furthest neighbor function $\AFN_{U, \calM}^{r, c} : \Z^n\times Y \to\{0,1,\ast\}$ is given by 
\[\AFN_{U, \calM}^{r, c}(x,q) = \begin{cases}
        1 & \textup{if $x \in \Z_{\geq 0}^{n}$ and } \max\left\{ \sfd(q,u_i) : x_i > 0 \right\} \geq cr \\ 
        0 & \textup{if $x \in \Z_{\geq 0}^n$ and } \max\left\{ \sfd(q,u_i) : x_i > 0 \right\} \leq r \\ 
        \ast & \textup{otherwise} 
    \end{cases}  \] 
In dynamic streaming, the underlying frequency vector $x\in\Z^n$ (encoding a multiset $X\subset U$) is updated via increments and decrements. The query $q\in Y$ is given at the end of the stream.
\end{definition}

It will be convenient to relabel the range of $\AFN_{U,\calM}^{r,c}$ as $\{\text{``close''},\text{``far''},\ast\}$.

\begin{lemma}\label{lem: streaming AFN in l_infty}
Fix any $r,\Delta > 0$, $\eps, \delta \in (0, 1)$. There exists a dynamic streaming algorithm computing $\AFN_{U, \ell_\infty^k}^{r, (1+\epsilon)}$ with probability $1-\delta$, with the following guarantees:
\begin{itemize}
\item The algorithm is a linear sketch with $\tilde{O}(k \log n\log(1/\delta) / \eps)$ rows, $n$ columns, and matrix entries in $\{-1,0, 1\}$.
\item The bit-complexity is $\tilde{O}((k \log n\log(mn) + \log \Delta) \log(1/\delta) / \eps)$, where $\Delta$ upper bounds the aspect ratio, and $m$ upper bounds the magnitude of the underlying frequency vector.
\end{itemize}
\end{lemma}

The proof of Lemma~\ref{lem: streaming AFN in l_infty} relies on suitably constructing a small number of $\ell_0$-samplers (Definition~\ref{def: L0 sampler}) for each coordinate of the input space.

\subsubsection{Proof of Lemma~\ref{lem: streaming AFN in l_infty}}

Here, we present a linear sketching algorithm for $\AFN_{U, \ell_\infty^k}^{r, (1+\epsilon)}$ matching the description of Lemma~\ref{lem: streaming AFN in l_infty}. First, consider the simple algorithm that maintains an $\ell_0$-sampler along each dimension and returns ``Far" if it finds a sample which is $> r$ far from the query $q$, and ``Close" otherwise. In the case that the furthest neighbor from $q$ is at distance $\leq r$, the algorithm will always return ``Close". However, in the case there exists a neighbor of $q$ which is far, the probability of sampling it is at most $k / |X|$ where $X$ denotes the surviving set of points. 

To improve the probability of sampling a far point (when it exists), our algorithm discretizes each dimension into buckets of width $\frac{1}{\epsilon r}$ and `activates' each bucket (using a pairwise independent hash function) with probability roughly $\epsilon$. We update the $\ell_0$ sampler along a particular dimension only if the point's bucket along that dimension is active. Now, the probability that a given `far' bucket is active while all `close' buckets are inactive is $\Theta(\epsilon)$. Finally, repeating $\Theta(\log (1/ \delta) / \epsilon)$ times boosts the success probability to  $1 - \delta$.  

Let $D_{\min}$ and $D_{\max}$ denote the minimum and maximum distances in $U$, respectively. By shifting and scaling, we may assume $D_{\min} = 1$ and that $U \subseteq [0,\Delta]^k$. Thus, if $r < 1$ we can safely return ``far'', and we henceforth assume $r \geq 1$.

    \begin{algorithm}[H]\label{alg: streaming-afn}  
        \caption{Streaming $\AFN_{U, \ell_{\infty}^{k}}^{r, (1+\epsilon)}$}
        Define $\phi: [0, \Delta] \to \{0, \ldots, \Delta / \epsilon \}$ by $\phi(z) = \lfloor z / (\eps r) \rfloor$ \\
        Let $K:= \lceil 4 / \epsilon + 2 \rceil$ 
        \vspace{10pt}

        Draw $\bh: \{0, \ldots, \Delta / \epsilon\} \to \{0, 1, \ldots, K-1\}$ from a pairwise independent hash family \\
        Initialize $k$ independent $\ell_0$ samplers $S_1, \ldots, S_k$ on frequency vectors $z^1, \ldots, z^k \in \Z^n$ with failure probability $\leq 1/(2k)$

        \vspace{10pt}
        \textbf{Stream Processing:}\\
        On insertion/deletion of $u_i \in U$, \ForEach{$\ell \in [k]$}{
            \If{$\bh(\phi(u_{i}^{\ell})) = 0$}{
                Increment/decrement $z_i^{\ell}$ in $S_{\ell}$ accordingly
            }
        } 
        \vspace{10pt}
        \textbf{Query Processing:} \\
        On query $q \in \R^k$, \ForEach{$\ell \in [k]$}{
            Draw sample $\bi_{\ell} \in \supp(z^{\ell})$ from $S_{\ell}$ (or $S_{\ell}$ declares that $z^{\ell}$ is zero) \\
            \If{$\norm{q - u_{\bi_{\ell}}}_{\infty} > r$}{
                Return Far
            }
        } 
        Return Close
    \end{algorithm}

    We will analyze the algorithm that performs $\Theta(\log(1 / \delta) / \epsilon)$ independent repetitions of Algorithm~\ref{alg: streaming-afn}. It outputs ``Far'' if any run outputs ``Far'', and ``Close'' otherwise.

\textbf{Analysis.} 
Let $X\subset U$ denote the final multiset of points. Suppose all $u \in X$ satisfy $\|u - q\|_{\infty} \le r$. Then as long as the $\ell_0$-samplers $S_1,\dots, S_k$ do not fail (which occurs with probability at least $1/2$), any $\bi_{\ell} \in [n]$ returned by $S_{\ell}$ will satisfy $x_{\bi_{\ell}} \neq 0$, and hence $\|q - u_{\bi_{\ell}}\|_{\infty} \le r$, so algorithm \ref{alg: streaming-afn} outputs ``close''. Suppose, on the other hand, that there exists an index $i^* \in [n]$ such that, at the end of the stream, $x_{i^*} > 0$ and $\|q - u_{i^*}\|_{\infty} \geq (1+\eps)r$. Let $\ell \in [k]$ be a coordinate such that $|u_{i^*}^{\ell} - q^{\ell}| \geq (1+\eps)r$, and define the set 
\[
Q = \left\{\bigfloor{q^{\ell} / \epsilon r} - \frac{1}{\eps}, \dots, \bigfloor{q^{\ell} / \epsilon r} + \frac{1}{\eps} \right\},
\]
a set of $2/\eps + 1$ integers. Notice that, any point $y \in U$ with $|y^{\ell} - q^{\ell}| \le r$ must have $\phi(y^{\ell}) \in Q$, and any $y \in U$ with $|y^{\ell} - q^{\ell}| \ge (1+\eps)r$ has $\phi(y^{\ell}) \notin Q$. In particular, $\phi(u_{i^*}^{\ell}) \notin Q$. Let $\calbE$ be the event that $\bh(\phi(u_{i^*}^{\ell})) = 0$ and $\bh(\beta) \ne 0$ for all $\beta \in Q$. Then, we use pairwise independence of $\bh$ to say
\begin{align*}
\Prx_{\bh}[\calbE] &= \Prx_{\bh}[\bh(\phi(u_{i^*}^{\ell})) = 0] \cdot \Prx_{\bh}[\bh(\beta) \ne 0\ \forall \beta \in Q \mid \bh(\phi(u_{i^*}^{\ell})) = 0] \\
&\ge \frac{1}{K} \left(1 - \sum_{\beta \in Q} \Prx_{\bh}[\bh(\beta) = 0]\right) = \frac{1}{K} \cdot \left(1 - \frac{|Q|}{K} \right) \ge \frac{1}{2K} \ge \frac{\eps}{12},
\end{align*}
using the setting of $K = \lceil 4/\eps + 2 \rceil$ and the fact that $\eps \in (0,1)$. Conditioned on $\calbE$, the following occurs. First, the sampler $S_{\ell}$ has $z^{\ell}_{i^*} > 0$ since $\bh(\phi(u_{i^*}^{\ell})) = 0$, which implies $\supp(z^{\ell}) \neq \emptyset$. Assume $S_{\ell}$ does not fail, so it outputs a coordinate $\bi_{\ell} \in [n]$ satisfying $x_{\bi_{\ell}} > 0$ and $\bh(\phi(u_{\bi_{\ell}}^{\ell})) = 0$. Since $\bh(\phi(u_{\bi_{\ell}}^{\ell})) = 0$, it must be the case that $\phi(u_{\bi_{\ell}}^{\ell}) \notin Q$ by $\calbE$. This, in turn, means $\|q - u_{\bi_{\ell}} \|_{\infty} \geq |q_{\ell} - u_{\bi_{\ell}}^{\ell}| > r$, and the algorithm outputs ``far.'' The success probability in this case is at least $\epsilon/12 \cdot 1/2 \geq \epsilon/24$. After $\Theta(\log (1 / \delta) / \epsilon)$ many independent repetitions, the success probability becomes at least $1 - \delta$. 


\textbf{Space Complexity.} 
By Lemma~\ref{thm: L0 sampler}, each $\ell_0$-sampler can be implemented as an $s \times n$ linear sketch with $s = O(\log n \log k)$ and using $O(\log n \log k \log(nm))$ bits of space where $m$ is an upper bound on the magnitude of each coordinate of $x \in \Z^n$ in the input vector. The hash function $\bh$ requires $O(\log(\Delta/\eps))$ bits to store. The sketch for one independent run of algorithm \ref{alg: streaming-afn} is an $O(k \log n \log k) \times n$ matrix, using $O(k \log n \log k \log(mn)) + O(\log(\Delta/ \eps))$ bits of space. Thus, the entire sketch for $O(\log(1 / \delta) / \epsilon)$ independent runs has dimensions $O(k \log n \log k \cdot \log(1 / \delta) / \epsilon) \times n$, using 
\[
O\left( (k \log k \log n \log(nm) + \log(\Delta/\eps)) \log(1/\delta) / \eps \right).
\]
bits of space. 


\begin{corollary} \label{cor: diameter Linfty algorithm} 
Fix any $r,\Delta > 0$ and any $c \ge 2(1 + \eps)$ for $\eps, \delta \in (0,1)$, there is a dynamic streaming algorithm for $\diam_{U}^{r, c}$ over $\ell_{\infty}^k$ that uses a linear sketch with $\tilde{O}(k \log n \log(1/\delta)/\eps)$ rows and $n$ columns. The total bit-complexity is $\tilde{O}( k \log n \log(nm) + \log \Delta) \log(1/\delta) / \eps)$, where $\Delta$ upper bounds the aspect ratio, and $m$ upper bounds the magnitude of the final frequency vector.
\end{corollary} 

\begin{proof} 
As before, we let $D_{\min}$ and $D_{\max}$ denote the minimum and maximum distances in $U$, respectively. By shifting and scaling, we may assume $D_{\min} = 1$ and that $U \subseteq [0, \Delta]^k$. Thus, if $r < 1$ we can safely return ``far'', and we henceforth assume $r \geq 1$. We create an $\ell_0$-sampler $S$ for the entire stream to sample an index $i \in [n]$ with $x_i \neq 0$. 
We condition on the sampler $S$ succeeding from here on, and we let $q = x_{i}$ for the sample $i$ generated by $S$ (note that if $\supp(x) = \emptyset$, the diameter is trivially small). We also run in parallel the $\smash{\AFN_{U, \ell_\infty^k}^{r, (1+\epsilon)}}$ algorithm of Lemma~\ref{lem: streaming AFN in l_infty} over the entire stream, also with error probability $\delta$. At the end of the stream, we give $q$ as the query point to the $\smash{\AFN_{U, \ell_\infty^k}^{r, (1+\epsilon)}}$ algorithm. 

If there exists a pair of points in $X$ with $\ell_\infty$ distance at least $2(1+\eps)r$, then for any query point $q \in X$, there exists another point $x^*$ that is at least $(1+\eps)r$ far; otherwise by triangle inequality the diameter of $X$ can be bounded by $2(1 + \eps)r$. Hence in this case, the $\smash{\AFN_{U, \ell_\infty^k}^{r, (1+\epsilon)}}$ algorithm will output ``far'' with probability at least $1-\delta$. On the other hand, if the diameter is at most $r$, then for any choice of $q \in X$, the $\smash{\AFN_{U, \ell_\infty^k}^{r, (1+\epsilon)}}$ algorithm will return ``close'' with probability at least $1 - \delta$. Hence, we get an algorithm for $\diam_{U}^{r, c}$ over $\ell_{\infty}^k$ that succeeds with probability at least $1-\delta$ and the space complexity follows accordingly.
\end{proof}

\subsection{Embedding Upper and Lower Bounds}

\begin{lemma}[Theorem~15.6.2 in~\cite{M02}] \label{lemma: frechet embedding}
    For any metric $\calM = ([n], \sfd)$ and any odd integer $\sfD = 2q - 1 \geq 3$, there exists an embedding of $\calM$ into $\ell_{\infty}^k$ of distortion $\sfD$ with $k = O(q n^{1/q} \log n)$. In other words, a function $g \colon [n] \to \R^k$ satisfying
    \[ \sfd(i,j) \leq \|g(i) - g(j) \|_{\infty} \leq \sfD \cdot \sfd(i,j) \]
    for all $i,j \in[n]$.
\end{lemma}

We are now ready to prove Theorem~\ref{thm:ub} and Theorem~\ref{thm:embed-lb}.

\begin{proof}[Proof of Theorem~\ref{thm:ub}]
Given a metric space $\calM = ([n],\sfd)$ and $c > 6$, we find the largest value of $q \geq 2$ and $\eps > 0$ such that $c - 2\eps \geq 4q - 2$, which corresponds to $q = \lfloor (c - 2) / 4\rfloor$ and $\eps >0$ a small enough constant so $c - 2\eps \geq 6$. We use Lemma~\ref{lemma: frechet embedding} to obtain an embedding $g\colon [n]\to\ell_\infty^k$ of distortion $\sfD = 2q-1 \geq 3$. We use Corollary~\ref{cor: diameter Linfty algorithm} with threshold parameter $\sfD r$, approximation $c = 2(1+\eps)$, $\delta = 1/10$, and let $U = g([n])$. The result is a streaming algorithm $\Alg$ for $\diam_{U}^{\sfD r, 2 (1+\epsilon)}$ succeeding with probability $0.9$. Our algorithm simply returns the output of $\Alg$.

Suppose the frequency vector $x \in \Z^n$ is such that $\diam_{\calM}(x) \leq r$, so that every $i, j\in [n]$ with $x_i, x_j > 0$ satisfies $\sfd(i, j) \leq r$. Then, we will have $\norm{g(i) - g(j)}_{\infty} \leq \sfD r$, and our algorithm will output ``close'' with probability at least $0.9$. On the other hand, if $\diam_{\cal{M}}(x) \geq 2(1+\eps) \sfD r$, let $i, j \in [n]$ be the indices where $x_i, x_j > 0$ and $\sfd(i, j) \geq 2(1+\eps) \sfD r$. Here, we will have $\norm{g(i) - g(j)}_{\infty} \geq 2(1+\eps) \sfD r$, and the algorithm will output ``far'' with probability at least $0.9$. The space complexity comes from Corollary~\ref{cor: diameter Linfty algorithm}, where note that $k = \tilde{O}(n^{1/q})$, so the linear sketch is a matrix of dimension $\tilde{O}(n^{1/q})\times n$, and total bit complexity $\tilde{O}(n^{1/q}) + \poly(\log(nm\Delta))$, where $m$ upper bounds the magnitude of the final frequency vector, and $\Delta$ upper bounds the aspect ratio. In terms of the approximation factor $c = 2(1+\eps) (2q-1) = 4q-2 +2\eps$ for $q \geq 2$, the resulting dimensionality of the sketch is $\tilde{O}(n^{1/ \lfloor (c - 2)/4\rfloor})$. 
\end{proof}

\begin{proof}[Proof of Theorem~\ref{thm:embed-lb}]
Let $\calbM = ([2n], \sfd)$ be the shortest path metric on a random bipartite graph $\calG(n,n,p)$ where \smash{$p = n^{-1+1/(2k-1)} / \log^5 n$}, and assume that the conclusions of Lemma~\ref{lem:afn-lb} hold for $\calbM$ (with probability at least $1-o(1)$). Now, suppose that $g \colon [2n] \to \ell_{\infty}^s$ is an embedding of $\calbM$ into $\ell_{\infty}^s$ with distortion $c < 2k+1$, and let $\eps > 0$ be a small enough parameter such that $c(1+\eps) \leq 2k+1$, and let $U = g([2n])$. We use Lemma~\ref{lem: streaming AFN in l_infty} to obtain a linear sketch for $\smash{\AFN_{U,\ell_\infty^s}^{c, (1+\epsilon)}}$ with failure probability $1/200n$ using a matrix of $\tilde{O}(s\log^2 n / \eps)$ rows, $n$ columns, and entries in $\{-1, 0, 1\}$, and directly compose it with our embedding to obtain the following linear sketch for $\smash{\AFN_{\calbM}^{1, 2k+1}}$. Upon receiving a query $q \in [2n]$ where $\sfd(q, i) \leq 1$ for all $i \in [2n]$ with $x_{i} > 0$, we have $\|g(q) - g(i) \|_{\infty} \leq c$ and our algorithm will output ``close'' with probability at least $1-1/200n$. On the other hand, if $q \in [2n]$ has some $i \in [2n]$ with $\sfd(q, i) \geq 2k+1$, then $\| g(q) - g(i) \|_{\infty} \geq 2k+1 \geq c(1 + \eps)$, hence, the algorithm will output ``far'' with probability at least $1-1/200n$. In other words, we obtain an algorithm which satisfies the conditions of Lemma~\ref{lem:afn-lb}, and therefore $\tilde{O}(s \log^2 n / \eps)$ is at least $\tilde{\Omega}(n^{1/(2k-1)})$, which implies $s = \tilde{\Omega}(n^{1 / (2k-1)})$.    
\end{proof}

%% file: appendix.tex
\section{The Minrank Lower Bound} \label{apx: minrank proof}

In this appendix, we sketch the argument of~\cite{ABGMM20} so as to give the lower bound on the minrank of a (directed) graph on $n$ vertices where each edge $(i,j)$ appears independently with probability $p$. Note that the third bullet point of Lemma~\ref{lem:necessary-properties}, involves showing that the minrank of a knowledge graph is large. Recall, $\bG = (U, V, \bE)$ is the random bipartite graph where $U = \{ u_1,\dots, u_n\}$ denotes vertices on the left-hand side, $V = \{ v_1,\dots, v_n \}$ denotes the vertices on the right-hand side, and every edge $(u_i, v_j)$ is present in $\bE$ independently with probability $p$. 

We may associate to $\bG$ the following knowledge (directed) graph $\bG' = ([n], \bE')$. The vertices are in $[n]$, and the edge $(i, j) \in \bE'$ if and only if $(u_i, v_j) \notin \bE$. Note, our construction of $\bG \sim \calG(n,n,p)$ means that $\bG'$ includes each edge independently with probability $q$. The main theorem of~\cite{ABGMM20} is quoted below, and then, we state the (minor) modification that we need to conclude Lemma~\ref{lem:necessary-properties}.
\begin{theorem}[Theorem 1.2 in~\cite{ABGMM20}]
    Let $\F = \F(n)$ be a field and assume $q = q(n)$ satisfies $1/n \leq q \leq 1$. Then, with high probability over $\bH \sim \calG(n,q)$,
    \[ \minrank_{\F}(\bH) \geq \frac{n \log(1/q)}{80 \log n}. \]
\end{theorem} 
We note that we will take $\F$ to be the field of real numbers, and we will use $q = 1-p$ (where $p$ comes from our parameter setting of $p = n^{1/(2k-1) - 1}/\log n$). The only minor modification is that, even though $\bG'$ defined above is drawn from $\calG(n, q)$, we require $\tilde{\bG}$ which is the sub-graph obtained from considering vertices in $I^*(\bG)$ (which depends on the randomness) have large minrank. The claim, which we show below, is that the argument of~\cite{ABGMM20}, which performs a union bound over potential sub-matrices, can give the following, from which our required bound will follow.

\begin{theorem}[Minor Modification to Theorem~1.2 of~\cite{ABGMM20}]\label{thm:minor-mod-abgmm}
Let $\F = \F(n)$ and $q = q(n)$ with $1/n \leq q \leq 1$. Then, with high probability over $\bH \sim \calG(n,q)$ any induced subgraph $\tilde{\bG}$ of $\bH$ on $t \leq n$ vertices satisfies
\[ \minrank_{\F}(\tilde{\bG}) \geq \Omega\left(\frac{t \log(1/q)}{\log n}\right).\]
\end{theorem}

\ignore{
the final part of our argument involves incorporating the fact we focus on a sub-matrix corresponding to entries in $I^*(\bG)$

While $I^\ast(\bG)$ is defined with respect to the metric $\sfd_\bG$, the crux of our argument will relate it to a combinatorial property of $\bG$ resembling the minrank \cite{ABGMM20}.

\begin{definition}\label{def: minrank completion}
    Let $G = (U,V,E)$ be a bipartite graph with $|U|=|V|=n$. For $1\leq t\leq n$, we say that a matrix $M\in\R^{n\times n}$ is a $t$-completion of $G$ if there is a subset $I\subset [n]$ of size $|I| = t$ for which: 
    \begin{itemize}
        \item $M_{ii} \neq 0$ for every $i \in I$,
        \item $M_{ij} = 0$ for every $i\in I$ and $j\in[n]$ with $(u_j,v_i)\in\bE$
    \end{itemize}
    
\end{definition}

\begin{lemma}[Minrank Lower Bound] \label{lem: minrank lower bound with bad rows}
    Let $n\in\N$ and $p\in[0,1]$. With probability $1 - o_t(1)$ over $\bG \sim \calG(n,n,p)$, any matrix $M=M(\bG)$ that is a $t$-completion of $\bG$ has
    \[ \rank(M) \geq \Omega\del{\frac{t p}{\log n}} \]
\end{lemma}}

\subsection{Proof of Theorem~\ref{thm:minor-mod-abgmm}}

\begin{definition}[\cite{ABGMM20}]
    The zero pattern of a sequence $(y_1,\ldots,y_m)\in\F^m$ is the sequence $(z_1,\ldots,z_m)\in\{0,\ast\}^m$ where $z_i = 0$ if $y_i = 0$ and $z_i = \ast$ if $y_i\neq 0$. The zero-pattern of a matrix $M\in\F^{n\times n}$ is the zero-pattern of its sequence of entries in $\F^{n^2}$.
\end{definition}

\begin{definition}[Definition 2.1 of~\cite{ABGMM20}]
    An $(n,k,s)$-matrix over a field $\F$ is a matrix $M\in\F^{n\times n}$ of rank $k$ with $s$ nonzero entries and containing rows $R_{i_1},\ldots, R_{i_k}$ forming a row basis for $M$ and columns $C_{j_1},\ldots,C_{j_k}$ forming a columns basis for $M$, such that the number of nonzero entries in the $2k$ vectors $R_{i_1},\ldots, R_{i_k}, C_{j_1},\ldots, C_{j_k}$ is at most $4ks/n$.
\end{definition}

\begin{lemma}[Lemma 2.2 of~\cite{ABGMM20}] \label{lem: principal submat}
    Let $\F$ be a field and $M\in\F^{n\times n}$ be a matrix of rank $k$. There are integers $n',k',s'$ with $k'/n'\leq k/n$ such that $M$ contains a $(n',k',s')$-principal submatrix.
\end{lemma}

\begin{lemma}[Lemma 2.4 of~\cite{ABGMM20}] \label{lem: zero patterns}
    For any field $\F$, the number of zero-patterns of $(n,k,s)$-matrices over $\F$ is at most $\binom{n}{k}^2 \cdot n^{20ks/n}$.
\end{lemma}

\begin{lemma}[Lemma 2.5 of~\cite{ABGMM20}] \label{lem: low rank implies dense} 
    Any matrix $M\in\F^{n\times n}$ of rank $k$ with nonzero entries on the diagonal must have at least $n^2 / (4k)$ nonzero entries.
\end{lemma}

We are now ready to prove Theorem~\ref{thm:minor-mod-abgmm}. The argument is essentially the same as the proof of Theorem 1.2 in ~\cite{ABGMM20}.

\begin{proof}[Proof of Theorem~\ref{thm:minor-mod-abgmm}]
    Let $\bH \sim \calG(n,q)$ and suppose there exists an induced subgraph $\tilde{\bG}$ of $\bH$ with $t$ vertices and a matrix $M$ of rank $k$ satisfying $M_{ii} \neq 0$ for all vertices $i$ in $\tilde{\bG}$ and $M_{ij} = 0$ whenever $i$ and $j$ are vertices in $\tilde{\bG}$ and $(i,j)$ is an edge of $\bH$. Then, we use Lemma~\ref{lem: principal submat} to find a $(n',k', s')$-principal submatrix $M'$ of $M$ for non-negative $n',k', s'$ which satisfy $k' / n' \leq k / t$ (note that $k' \geq 1$ since every principal submatrix of $M$ has ones on the diagonal). Note that $M'$ is also a principal submatrix of $\bH$, so we will argue that unless $k$ is large, the probability that any such principal submatrix $M'$ exists is small.

    Now, fix $U'\in\binom{[n]}{n'}$ and a potential zero-pattern $Z\in\{0,\ast\}^{U'\times U'}$ of a potential principal submatrix $(n',k',s')$-principal submatrix with $k'/n' \leq k / t$. Consider the probability taht the entries of $\bH \sim \calG(n,q)$ within $U'$ are consistent with $Z$. Observe that, for each entry $Z_{ij} \neq 0$, we must have $(i,j)$ in $\bH$, which occurs independently with probability $q$. Since $M'$ contains $s'-n'$ nonzero off-diagonal entries, the probability that $\bH$ is consistent with any fixed $Z$ is at most $q^{s'-n'}$. Moreover, by Lemma~\ref{lem: low rank implies dense}, $s'$ is at least $(n')^2 / (4k')$ which is at least $(n'/4) \cdot t/k$, since $n' / k' \geq t / k$. Hence, the probability that there exists such $(n',k',s')$-principal submatrices is at most
    \begin{align}
        \sum_{n'=1}^t \sum_{k'=1}^k \sum_{s' \geq n'/4 \cdot t/k}\binom{n}{n'} \binom{n'}{k'}^2 (n')^{20k's'/n'} \cdot q^{s'-n'}, \label{eq:haha}
    \end{align}
    which we show is small. First, we can upper bound the count on the number of choices for sub-matrices and zero patterns using the fact $n' \leq n$, and $k' / n' \leq k / t$:
    \[ \binom{n}{n'} \cdot \binom{n'}{k'}^2 (n')^{20k's'/n'} \leq n^{n' + 2k' + 20 k s' / t}.\]
    Therefore, we may re-write (\ref{eq:haha}) and use the fact that $n^{20k/t}$ is at most $1/\sqrt{q}$ whenever $k \leq t\log(1/q) / (160\log n)$, which implies
    \begin{align*}
        \sum_{n'=1}^t \sum_{k'=1}^k n^{n' + 2k'} / q^{n'} \sum_{s \geq n'/4 \cdot t/k} \left(q n^{20k/t} \right)^{s'} \leq \sum_{n'=1}^t \sum_{k'=1}^k n^{n' + 2k'} / q^{n'} \cdot (q^{t/k})^{n'/8} \left(\dfrac{1}{1 - \sqrt{q}} \right)
    \end{align*}
    Once more, we use the setting of $k$ to note 
    \[ q^{t/k} = \exp\left(-\frac{t \log(1/q)}{k} \right) \leq \exp\left( - 160 \log n \right) = n^{-160},\]
    and the upper bound on (\ref{eq:haha}) becomes at most
    \begin{align*}
        \sup_{n' \in \N} \left\{ n^{3n' + 2} \cdot n^{-20 n'} \cdot \left(\frac{1}{1 - \sqrt{q}} \right) \cdot (1/q)^{n'} \right\} \leq \sup_{n' \in \N} \left\{ n^{4n' + 4 - 20n'} \right\} \leq 1/n^{12}.
    \end{align*}
    where we used the fact $1/n \leq q \leq 1 - 1/n$ to first upper bound $(1/q)^{n'} \leq n^{n'}$, and then upper bound $1 / (1-\sqrt{q}) \leq n^2$. Note that, whenever $q \geq 1 - 1/n$, the desired bound is trivially true, as $t \leq n$.

\end{proof}


    

\newcommand{\rep}{\mathrm{rep}}
\newcommand{\Add}{\textsc{Add}}
\newcommand{\Remove}{\textsc{Remove}}

\section{Reduction to Path-Independent Algorithm}\label{sec:stream-to-path-ind}

The goal of this section is to provide a self-contained proof that, given any randomized streaming algorithm $\calA$ computing some function $g$, one can extract a deterministic, path-independent algorithm $\calB$ which computes $g$ over any distribution $\calD$ supported on finitely many frequency vectors. Reductions of this form were introduced in~\cite{G08,LNW14} as a core component of strengthening lower bounds for linear sketches into lower bounds for general streaming algorithms. The main theorem, relating general randomized streaming algorithms $\calA$ to path-independent algorithms $\calB$, is given below and corresponds to Theorem~5 of~\cite{LNW14}---the crucial definitions of randomized streaming algorithms and path-independent algorithms, appear in Definition~\ref{def:randomized-streaming-alg} and Definition~\ref{def:path-ind}.

For $a,b\in\Z$ with $a\le b$, we will use the notation $\interval{a,b} := \{a, a+1,\ldots, b\}$. We will also use $\log x$ to denote the base-$2$ logarithm of $x$.


\ignore{
\begin{definition}[Path-Independent Algorithm]\label{def:path-ind}
For any $n,m \in \N$, a deterministic, path-independent streaming algorithm $\calB$ with bound $m$ and dimension $n$ is specified by a pair of functions $(\Enc,\Dec)$ along with a set $\calW$ of states and transitions.
\begin{itemize}
    \item \textbf{Encoding.} $\Enc \colon \interval{-m,m}^n \to \calW$, maps frequency vectors to states.
    \item \textbf{Decoding.} $\Dec \colon \calW \to \{0,1\}$, maps a state to a final output.
    \item \textbf{Transitions.} Let $w \in \calW$ for which $\Enc^{-1}(w)$ contains a vector in $\interval{-(m-1), m-1}^n$. There is a well-defined transition function $\oplus$ that specifies, for each of the $2n$ possible stream updates $\xi e_i$ (with $\xi \in \{-1,1\}$ and $i\in[n]$), the new state of the algorithm: \[w' = w\oplus \xi e_i.\] One can view $\calW$ as vertices in a directed graph, with the edge $(w,w')$ labeled by $\xi e_i$.
\end{itemize}
$\calB$ is called \emph{path-independent} if it satisfies the following property:
\begin{itemize}
\item \textbf{Path-Independence}. For all $x \in \interval{-(m-1),m-1}^n$, $i \in [n]$, and $\xi \in \{-1,1\}$, one has \[\Enc(x) \oplus \xi e_i = \Enc(x+\xi e_i).\] 
\end{itemize}

Two relevant properties of streaming algorithms are space complexity and correctness.
\begin{itemize}
\item \textbf{Space Complexity}. The space complexity of $\calB$, and the nonnegative space complexity of $\calB$ are specified functions $\calS(\calB, \cdot), \calS^+(\calB, \cdot) \colon \N \to \R$, given by:
\begin{align*}
\calS(\calB, \ell) &= \log\left|\left\{ \Enc(x) \in \calW : x \in \interval{-\ell, \ell}^n \right\} \right| \qquad \calS^+(\calB, \ell) = \log\left|\left\{ \Enc(x) \in \calW : x \in \interval{0, \ell}^n \right\} \right|.
\end{align*}
\item \textbf{Correctness}. We say that $\calB$ computes a partial function $g \colon \interval{-m,m}^n \to \{0,1,*\}$ over a distribution $\calD$ supported on $\interval{-m,m}^n$ with probability at least $1-\delta$, if
\[ \Prx_{\bx \sim \calD}\left[ g(x) \in \{ 0,1 \} ~\mathrm{and}~\Dec(\Enc(\bx)) \neq g(\bx) \right] \leq \delta \]
\end{itemize}
\end{definition}

We note that, for our construction of path-independent streaming algorithms $\calB$, we will need to specify a finite bound $m \in \N$, which means that the path-independence assumption is only for vectors $x$ whose coordinates have magnitude less than $m$. While the specific bound on $m$ may affect the description of $\calB$, it does not affect the underlying space complexity bound $\calS(\calB, \ell)$ for $\ell \leq m$. 

\begin{remark}[Nonnegative Frequency Vectors] 
For streaming problems $g$ which are naturally defined on \emph{nonnegative} frequency vectors (e.g., the diameter), it becomes more natural to consider the space complexity as the logarithm of distinct number of final working tapes of nonnegative frequency vectors. In particular, for a randomized algorithm $\calA$, we may define $\calW^+(\calA, \ell)$ as the distinct number of final working tapes resulting from streams with frequency vectors in $\interval{0,\ell}^n$, 
and analogously, 
\[ \calS^{+}(\calA, \ell) = \log \left| \left\{ \Enc(x) \in \calW : x \in \interval{0,\ell}^n \right\} \right|. \]
While a lower bound on $\calS^+(\calA, \ell)$ implies the lower bound on $\calS(\calA, \ell)$, the opposite may not be true. Throughout the argument, we will consider working tapes and encodings resulting from streams with frequency vectors in $x \in \interval{-m, m}^n$, but we will obtain a final space lower bound on $\calS^{+}(\calA, \ell)$.
\end{remark} }

\lnwthm

To prove Theorem~\ref{thm:lnw}, we will first generate a graph corresponding to possible working tapes of the streaming algorithm $\calA$. In a streaming algorithm, the number of possible configurations of working tapes may be infinite (as the algorithm is meant to support arbitrarily long streams); however, it will be helpful for us to define a finite graph, which is where we will use the fact that we have a finite bound. Namely, the fact that distribution $\calD$ is finitely-supported means there is a large enough $m' \in \N$ such that $\supp(\calD) \subset \interval{-m',m'}^n$. Then, our construction will be parametrized by any fixed bound $m \geq m'$. As mentioned earlier, while the fact $m$ is finite is important for the construction, it will not affect the complexity analysis, as we have the space complexity bound $\calS(\calB, \ell) \leq \calS(\calA, \ell)$ and $\calS^+(\calB, \ell) \leq \calS^{+}(\calA, \ell)$ for all $\ell \leq m$. The subsequent sections give the construction of $\calB$; we note that multiple intermediate lemmas and definitions must be established to ensure the objects are well-defined and that they satisfy the desired bounds on space complexity and failure probability. The overview of the proof, which solely connects the various lemmas and sections, appears in Subsection~\ref{sec:putting-everything-together}.

\subsection{Zero-Frequency Transitions and State Transition Graphs}\label{sec:encoding}

We first define the following two families of directed graphs. Both families are parameterized by the contents of a randomness tape $\rho$. The first, which we call a \emph{zero-frequency transition graph}, encodes transitions of $\calA$ corresponding to streams with frequency vector $0^n$, henceforth called \emph{zero-frequency streams}. Intuitively, adding streams with frequency vector $0^n$ does not change the frequency vector, and thus transitions in a zero-frequency graph should not affect the final output of the algorithm. The second family of graphs is the \emph{state transition graph}, which uses the zero-frequency transition graph to define equivalence classes. These equivalences classes will be used to define the path-independent streaming algorithm $\calB$.

\begin{definition}[Zero-Frequency Transition Graph]\label{def:zero-freq}
For a randomized streaming algorithm $\calA$ of randomness complexity $r \in \N$ and any $m \in \N$, we consider the following family of directed graphs $G(\rho)$, which are parameterized by a randomness tape $\rho \in \{0,1\}^{r}$. 
\begin{itemize}
\item \textbf{Vertices}. Each vertex of the graph is a final working tape in $\calW(\calA, m)$. 
\item \textbf{Edges}. There is a directed edge $(o_1, o_2)$ for $o_1, o_2 \in \calW(\calA, m)$ if there exists a zero-frequency stream $\sigma$ such that $\calA$ starting at working tape $o_1$ transitions to working tape $o_2$ after processing stream $\sigma$.
\end{itemize}
\end{definition}

Since we have imposed a finite bound on $m \in \N$, we ensure that $G(\rho)$ is finite graph. An important aspect of this definition, which will also be useful for bounding the space complexity of $\calB$, is that the vertices of $G(\rho)$ can be partitioned into a nested family of sets: \[\calW(\calA, 0) \subset \calW(\calA, 1) \subset \calW(\calA, 2) \subset \dots \subset \calW(\calA, m).\] Using $G(\rho)$, we now define the vertices and edges of the state transition graph $H(\rho)$.


\begin{definition}[Vertices of the State Transition Graph]\label{def:vert-state-transition}
For $\rho \in \{0,1\}^{r}$, the vertices of the state transition graph $H(\rho)$ are denoted $T(\rho)$ and constructed in the following manner:
\begin{itemize}
\item Partition the vertices of $G(\rho)$ (in Definition~\ref{def:zero-freq}) into strongly connected components of $G(\rho)$, and let $C(\rho)$ be the set of components. Let $\rep \colon C(\rho) \to \calW(\calA, m)$ be an injective map\footnote{The function $\rep$ depends on $\rho$, but we will drop it since it will be clear from context.}
\[ \rep(v) = \begin{array}{c} \text{arbitrary $o \in \calW(\calA, \ell)$ with $o \in v$,} \\
					\text{for the minimum possible $\ell\leq m$.} \end{array} \] 
\item We denote the terminal components by the set $T(\rho) \subset C(\rho)$ given by
\[ T(\rho) = \left\{ v \in C(\rho) : \forall o_1 \in v, \text{ all paths from $o_1$ in $G(\rho)$ end at a vertex $o_2 \in v$}\right\}.\]
\end{itemize}
\end{definition}
Intuitively, contracting strongly connected components of $G(\rho)$ defines a directed acyclic graph whose vertices are strongly connected components, and directed edges between components are present if there exists an edge between two vertices in each component---the graph is a acyclic as otherwise, the cycle would have generated a strongly connected component. The set $T(\rho)$ corresponds to the ``terminal'' vertices in the directed acyclic graph (i.e., one cannot reach any other component via edges of $G(\rho)$). In order to define the edges of the state transition graph $H(\rho)$ among vertices $T(\rho)$, it will be helpful to first define the following map.
\begin{definition}
For $\rho \in \{0,1\}^r$, let $\alpha \colon C(\rho) \to T(\rho)$ map each component $v \in C(\rho)$ to an arbitrary terminal component reachable from $v$ in $G(\rho)$.\footnote{Here, ``reachability'' among components simply refers to reachability of underlying vertices: a component $v$ is reachable from a component $u$ in $G(\rho)$ if there exists a path in $G(\rho)$ from a vertex $o_1 \in u$ to a vertex $o_2 \in v$. The map $\alpha$ depends on $\rho$ but we drop it since $\rho$ will be clear from context.}
\end{definition} 

\begin{definition}[Edges of the State Transition Graph]\label{def:state-transition-edges}
For $\rho \in \{0,1\}^r$, and any two vertices $v, w \in T(\rho)$ and $\xi \in \{-1,1\}$ and $i \in [n]$,
\[ (v, w) \in H(\rho) \text{ with label $\xi e_i$}\iff \begin{array}{c} 
\text{$\calA$ with randomness tape $\rho$ transitions $\rep(v)$ to}\\ 
\text{a working tape $o_2 \in \calW(\calA, m)$ with update $\xi e_i$} \\ 
\text{where $o_2 \in u \in C(\rho)$ with $\alpha(u) = w$.} \end{array} \]
\end{definition}

Here, the label $\xi e_i$ is associated with the update $\Add(i)$ when $\xi = 1$ and $\Remove(i)$ when $\xi = -1$. An important aspect to notice (and the reason that the condition of path-independence will be restricted to vectors in $\interval{-(m-1), m-1}^n$) is that every vertex $v\in T(\rho)$ with $\rep(v) \in \calW(\calA, m-1)$ has outgoing edges corresponding to each of the $2n$ distinct updates. There may be vertices $v \in T(\rho)$ where $\rep(v)$ lies in $\calW(\calA, m)$ and does not have all $2n$ outgoing edges, and this occurs whenever the final working tape (after updating $\rep(v)$) can only be reached via streams of frequency vectors larger than $m$.

We now give the main lemma (Lemma~7 of~\cite{LNW14}) concerning the state transition graph $H(\rho)$. This lemma will be used to argue correctness, by relating transitions in $H(\rho)$ back to the inner workings of $\calA$. The lemma, as well as subsequent arguments, make heavy use of the following simple fact regarding terminal components.
\begin{fact}\label{fact:terminal}
Let $G$ be any directed graph with strongly connected components $C$, and let $T \subset C$ be the set of terminal components. If $x$ is a vertex of $G$ within a terminal component, and $y$ is the endpoint of a path in $G$ starting from $x$, then $x$ and $y$ lie in the same component.
\end{fact}

Lemma~\ref{lem:relate} below will conclude that, so long as $\calA$ is initialized with randomness tape $\rho$ and an arbitrary working tape $o_1$ lying within a terminal component $v$ of $H(\rho)$, the following will hold for any frequency vector $x \in \interval{-m,m}^n$: any stream $\sigma$ of frequency vector $x$ will result in a working tape $o_2$ within a component which can reach a \emph{unique} terminal component of $H(\rho)$ (via edges of $G(\rho)$). Importantly, the uniquely reachable terminal component $u$ will only depend on $\rho$, $v$, and $x$. Roughly speaking, this will allow us to design an algorithm which solely maintains the unique terminal reachable component from a particular frequency vector. Since this terminal component is invariant under zero-frequency stream updates, path-independence follows, which is formally shown in Lemma~\ref{lem:path-ind}.

\begin{lemma}\label{lem:relate}
For any $\rho \in \{0,1\}^r$, any terminal component $v \in T(\rho)$ with $\rep(v) \in \calW(\calA, \ell)$, and any frequency vector $x \in \interval{-(m-\ell),m-\ell}^n$, there is a unique $u \in T(\rho)$ where the following holds:
\begin{itemize}
\item For any $o_1 \in v$, let $o_2$ be the working tape after transitioning $\calA$ from $o_1$ with an arbitrary stream $\sigma$ of frequency vector $x$ and randomness tape $\rho$.
\item Then, $o_2 \in \calW(\calA, m)$, and any path in $G(\rho)$ which starts at $o_2$ and ends in any $o_3 \in u' \in T(\rho)$ satisfies $u' = u$. 
\end{itemize}
\end{lemma}

\begin{proof}
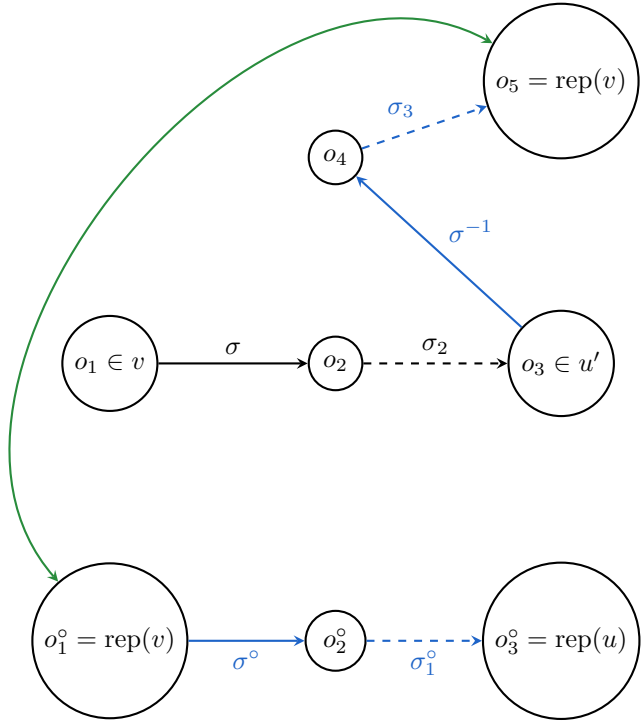
\begin{figure}[h]
\centering
\begin{tikzpicture}[->, >=stealth, node distance=3cm, thick]

\definecolor{niceblue}{RGB}{30, 100, 200} 
\definecolor{nicegreen}{RGB}{40, 140, 60}   

\node[circle, draw] (o1) {$o_1 \in v$};
\node[circle, draw, right of=o1] (o2) {$o_2$};
\node[circle, draw, right of=o2] (o3) {$o_3 \in u'$};

\node[circle, draw, below=2cm of o1] (o1c) {$o_1^{\circ} = \rep(v)$};
\node[circle, draw, right of=o1c] (o2c) {$o_2^{\circ}$};
\node[circle, draw, right of=o2c] (o3c) {$o_3^{\circ} = \rep(u)$};

\node[circle, draw, above=2cm of o2] (o4) {$o_4$};
\node[circle, draw, above=2cm of o3] (o5) {$o_5 = \rep(v)$};

\draw[->] (o1) -- node[above] {$\sigma$} (o2);
\draw[->, dashed] (o2) -- node[above] {$\sigma_2$} (o3);

\draw[niceblue, ->] (o1c) -- node[below] {$\sigma^{\circ}$} (o2c);
\draw[niceblue, ->, dashed] (o2c) -- node[below] {$\sigma_1^{\circ}$} (o3c);

\draw[niceblue, ->] (o3) -- node[above right] {$\sigma^{-1}$} (o4);
\draw[niceblue, ->, dashed] (o4) -- node[above left] {$\sigma_3$} (o5);

\draw[<->, nicegreen, bend right = 80] (o5) to (o1c);

\node[below=0.2cm of o1c] {};
\node[above=0.2cm of o1] {};

\end{tikzpicture}
\caption{Diagram of the proof of Lemma~\ref{lem:relate}. Dashed arrows denote zero-frequency transitions, blue arrows highlight the zero-frequency path from $o_3$ to $o_3^{\circ}$, and the green bidirectional arrow indicates that $o_1^{\circ}$ and $o_5$ are identical.}
\label{fig:transition-diagram}
\end{figure}

We first define $u$ in the following way. First, fix an arbitrary stream $\sigma^{\circ}$ of frequency vector $x$ (say, which performs updates in order of increasing $i$), and let $o_2^{\circ}$ be the transition from $o_1^{\circ} = \rep(v)$ with the stream $\sigma^{\circ}$ and randomness tape $\rho$. Notice that $o_2^{\circ} \in \calW(\calA, m)$ for the following reason: $\rep(v) \in \calW(\calA, \ell)$ implies there exists a frequency vector $y \in \interval{-\ell,\ell}^n$ and a stream which leads to a final working tape $\rep(v)$; concatenating $\sigma^{\circ}$ implies that $o_2^{\circ}$ is the final working tape of a stream with frequency vector $y + x \in \interval{-m,m}^n$. Let
\[ v_2^{\circ} \in C(\rho) \text{ with $o_2^{\circ} \in v_2^{\circ}$} \qquad v_3^{\circ} = \alpha(v_2^{\circ}) \qquad \text{and} \qquad o_3^{\circ} = \rep(v_3^{\circ}), \]
and let $\sigma_1^{\circ}$ be an arbitrary zero-frequency stream corresponding to an arbitrary path in $G(\rho)$ which goes from $o_2^{\circ}$ to $o_3^{\circ}$. Note such a stream must exist since, in $G(\rho)$, $o_2^{\circ}$ can reach $\rep(v_2^{\circ})$ (by strong connectivity) and $\rep(v_2^{\circ})$ can reach $o_3^{\circ}$ (by definition of $\alpha$). We define $u = v_3^{\circ} \in T(\rho)$.

We now argue the conclusions of the lemma by considering an arbitrary $o_1 \in v$, an arbitrary stream $\sigma$ of frequency vector $x$ which transitions $\calA$ from $o_1$ to $o_2$ with randomness tape $\rho$, and an arbitrary path in $G(\rho)$ which ends in any $o_3 \in u' \in T(\rho)$, corresponding to the zero-frequency stream $\sigma_2$. We will construct a path in $G(\rho)$ which connects $o_3$ to $o_3^{\circ} = \rep(u)$ and use Fact~\ref{fact:terminal} to conclude the claim (see Figure~\ref{fig:transition-diagram}). The path is obtained as follows:
\begin{itemize}
\item Let $o_4$ be the working tape obtained from transitioning $\calA$ from $o_3$ with stream $\sigma^{-1}$ and randomness tape $\rho$ (i.e., $\sigma^{-1}$ is that which inverts $\sigma$). Note, $o_4 \in \calW(\calA, \ell)$ since it is the result of a stream with frequency vector $y \in \interval{-\ell,\ell}^n$, and hence $v_4 \in C(\rho)$ with $o_4 \in v_4$. 
\item Let $\sigma_3$ an arbitrary zero-frequency stream corresponding to a path which begins at $o_4$ and ends at $o_5 \in \calW(\calA, \ell)$ given by first letting $v_5 = \alpha(v_4)$ and then $o_5 = \rep(v_5)$. Note, the concatenation $\sigma, \sigma_2, \sigma^{-1}, \sigma_3$, has frequency vector $0$ and leads to a path in $G(\rho)$, $o_1 \to o_2 \to o_3 \to o_4 \to o_5$, of frequency vector $0$, where $o_1 \in v$ and $o_5 = \rep(v_5)$ both lie in a terminal component.
\item Since $o_1$ and $o_5$ have a zero-frequency path and lie in terminal components, Fact~\ref{fact:terminal} implies they lie in the same component, and this means $o_5 = o_1^{\circ} = \rep(v)$. 
\end{itemize}
We now consider the zero-frequency stream given by the concatenation of $\sigma^{-1}, \sigma_3, \sigma^{\circ}, \sigma_1^{\circ}$ (note that it has zero-frequency since $\sigma^{-1}$ and $\sigma^{\circ}$ invert each other). The stream transitions $\calA$ with randomness tape $\rho$ by: $o_3 \to o_4 \to o_5 = o_1^{\circ} \to o_2^{\circ} \to o_3^{\circ}$, but we have thus found a zero-frequency stream which connects two working tapes $o_3$ and $o_3^{\circ}$ which lie in terminal components. By Fact~\ref{fact:terminal} once more, $o_3$ and $o_3^{\circ}$ lie in the same terminal component, and hence $u = u'$.
\end{proof}

We are now ready to define a collection of encoding maps. The deterministic, path-independent algorithm $\calB$ will be given by choosing a fixed element from the collection.

\begin{definition}\label{def:initial-states}
For any $\rho \in \{0,1\}^r$, let $T_{0}(\rho) \subset T(\rho)$ denote the set of terminal components reachable from the initial working tape of $\calA$ via zero-frequency streams and randomness tape $\rho$.
\end{definition}

\begin{definition}[States]\label{def:states}
For any fixed $\rho \in \{0,1\}^r$, the set of states $\calW$ is defined to be $T(\rho)$.
\end{definition}

\begin{definition}[Family of Encoding Functions]\label{def:encodings}
For any $\rho \in \{0,1\}^r$ and any $v \in T_{0}(\rho)$, let $\Enc(\cdot; \rho, v) \colon \interval{-m,m}^n \to T(\rho)$ be given by:
\begin{align*}
\Enc(x; \rho, v) = \begin{array}{c} \text{unique terminal component $u \in T(\rho)$ reachable from an}\\ 
						   \text{arbitrary working tape $o_1 \in v$ via an arbitrary stream $\sigma$} \\
						   \text{of frequency vector $x$ in $\calA$ with randomness tape $\rho$.} \end{array} 
\end{align*}
By Lemma~\ref{lem:relate}, the map $\Enc$ is well-defined.
\end{definition}

\begin{definition}[Transitions]\label{def:trans}
For any $\rho \in \{0,1\}^r$ and $v \in T_0(\rho)$, the transition from $u \in T(\rho)$ via $\xi e_i$, denoted $u \oplus_{\rho,v} \xi e_i$ is given by traversing the outgoing edge from $u \in H(\rho)$ labeled $\xi e_i$.
\end{definition}

Note, that in Definition~\ref{def:trans}, an outgoing edge from $u$ is guaranteed to exist whenever $\Enc^{-1}(u; \rho, v)$ contains a frequency vector $x \in \interval{-(m-1),m-1}^n$. This is because there exists a stream $\sigma$ which transitions $\calA$ from the initial state to $\rep(v)$ with randomness tape $\rho$, and furthermore, Lemma~\ref{lem:relate} implies that any stream $\sigma'$ with frequency vector $x$ will transition $\calA$ from $\rep(v)$ to a working tape $o'$ which can reach $\rep(u)$. Whenever $x \in \interval{-(m-1),m-1}^n$, we have $\rep(u) \in \calW(\calA, m-1)$, and thus $u$ has $2n$ outgoing edges in $H(\rho)$.

\subsection{Path-Independence and Space Complexity}\label{sec:path-and-space}

We are now ready to prove that the encoding function in Definition~\ref{def:encodings} and transitions constructed in Definition~\ref{def:trans} satisfy path-independence with bound $m$ as well as the claimed space complexity bound, for any fixed setting of $\rho \in \{0,1\}^r$ and $v \in T_{0}(\rho)$.

\begin{lemma}[Path-Independence]\label{lem:path-ind}
Consider any fixed $\rho \in \{0,1\}^r$ any $v \in T_0(\rho)$. For any $x\in \interval{-(m-1),m-1}^n$, $\xi \in \{-1,1\}$ and $i \in [n]$, 
\begin{align*}
\Enc(x; \rho, v) \oplus \xi e_i = \Enc(x + \xi e_i ; \rho, v).
\end{align*}
\end{lemma}

\begin{proof}
Let $\sigma$ be a stream of frequency $x$ that transitions $\calA$ from $\rep(v)$ to a working tape $o_1$, and let $\sigma'$ be a zero-frequency stream which transitions $\calA$ from $o_1$ to $\rep(u)$, where $u = \Enc(x; \rho, v)$ is the unique terminal component which is reachable from $o_1$. By Definition~\ref{def:trans} and Lemma~\ref{lem:relate}, the transition $\Enc(x; \rho, v) \oplus \xi e_i$ corresponds to traversing an edge $(u, u')$ labeled $\xi e_i$. Furthermore, the fact $x \in \interval{-(m-1),m-1}^n$ means that $\rep(u) \in \calW(\calA, m-1)$, and hence such an edge must exist. Note, $u'$ is the unique terminal component which is reachable from the working tape obtained by transitioning $\rep(u)$ with $\xi e_i$, say via zero-frequency stream $\sigma''$. However, since the concatenation of $\sigma, \sigma', \xi e_i, \sigma''$ has frequency vector $x + \xi e_i$, we conclude that $u' = \Enc(x+\xi e_i; \rho, v)$.  
\end{proof}

\begin{lemma}[Space Complexity]\label{lem:space}
For any $\rho \in \{0,1\}^r$ and $v \in T_0(\rho)$, and any $\ell \leq m$,
\begin{align*}
\left| \left\{ \Enc(x; \rho, v) : x \in \interval{-\ell,\ell}^n \right\} \right| &\leq |\calW(\calA, \ell)|,
\end{align*}
and in addition, $\left| \left\{ \Enc(x; \rho, v) : x \in \interval{0, \ell}^n \right\} \right| \leq |\calW^+(\calA, \ell)|$.
\end{lemma}

\begin{proof}
We first show the first claim, that for any $\ell \leq m$, the set of distinct states $\Enc(x;\rho, v)$ obtained via $x \in \interval{-\ell,\ell}^n$ is at most $|\calW(\calA,\ell)|$. We do this by showing that the restriction of $\rep \colon C(\rho) \to \calW(\calA, m)$ to the domain of states in $\{ \Enc(x;\rho,v) : x\in \interval{-\ell,\ell}^n\}$ maps to $\calW(\calA, \ell)$ and is an injection. First, any domain restriction of an injective map is injective, so it remains to show $u = \Enc(x; \rho, v)$ satisfies $\rep(u) \in \calW(\calA, \ell)$. For that, the fact $u = \Enc(x; \rho, v)$ with $x \in \interval{-\ell, \ell}^n$ means there exists a stream which can transition $\calA$ from the initial working tape to a working tape $o_1 \in u$ where $o_1 \in \calW(\calA, \ell)$. Hence, $\rep(u) \in \calW(\calA, \ell)$, since some $o_1$ exists. 

To generalize the argument for vectors in $x \in \interval{0,\ell}^n$, we consider an arbitrary function $\rep^{+} \colon C(\rho) \to \calW(\calA, m)$ which is an injection and for $u \in C(\rho)$, sets $\rep^+(u) = o$ to an arbitrary $o \in \calW^+(\calA, \ell)$ with $o \in u$ if one exists for the smallest possible $\ell \leq m$, and an arbitrary $o \in \calW(\calA, \ell)$ with $o \in u$ otherwise. Note (similarly to $\rep$ in Definition~\ref{def:vert-state-transition}), $\rep^{+}$
 is an injection: any two $u, u' \in C(\rho)$ and any $o \in u$ and $o' \in u'$, $\{ \rep(u), o\}$ and $\{ \rep(u'), o'\}$ lie in the same strongly connected components of $G(\rho)$, so $\rep(u) = \rep(u')$ implies that $\{o,o'\}$ lie in the same strongly connected component of $G(\rho)$. Since $o, o'$ were arbitrary, $u = u'$. Hence, the output restriction to states in $\Enc(x;\rho, v)$ with $x \in \interval{0,\ell}^n$ is injective, and there is a stream with frequency vector $x\in \interval{0,\ell}^n$ which transitions $\calA$ from the initial state to $o_1 \in u$, so $\rep^+(u) \in \calW^+(\calA, \ell)$.
 
\end{proof}

\subsection{Definition of $\calB$ and Correctness}\label{sec:decode}

In order to complete the construction of the deterministic, path-independent algorithm $\calB$, it remains to choose a setting of $\rho \in \{0,1\}^r$ and $v \in T_0(\rho)$, as well as the decoding function $\Dec(\cdot; \rho, v)$, to ensure correctness. We begin with the following lemma, which constructs a certain distribution over zero-frequency streams, which will effectively allow us to assume that the algorithm $\calA$ transitions to a terminal component. This will be useful, not only for selecting the initial terminal component $v \in T(\rho)$ to be used in $\Enc(\cdot; \rho, v)$, but also for defining the decoding function.

\begin{lemma}[Finite Collection of Zero-Frequency Streams]\label{lem:finite-length}
For any randomized streaming algorithm $\calA$ with finite randomness complexity $r \in \N$, there exists a finite collection $S$ of zero-frequency streams such that, for all $o_1, o_2 \in \calW(\calA, m)$ and any $\rho \in \{0,1\}^r$,
\[ (o_1, o_2) \text{ is an edge of } G(\rho) \iff \exists \sigma \in S : \begin{array}{c} \text{transitioning $\calA$ from $o_1$ with randomness tape $\rho$} \\ \text{and input stream $\sigma$ gives working tape $o_2$.} \end{array}  \]
\end{lemma}

\begin{proof}
We use the finite bound on $\rho \in \{0,1\}^r$ and $|\calW(\calA, m)|$ to obtain a finite bound on the length of streams needed to transition edges in $G(\rho)$. First, for any $\rho \in \{0,1\}^r$ and any edge $(o_1, o_2)$ of $G(\rho)$, we let 
\[ L(\rho, o_1, o_2) = \min\left\{ |\sigma| : \begin{array}{c} \text{transitioning $\calA$ from $o_1$ with zero-frequency} \\
											\text{stream $\sigma$ with randomness tape $\rho$ gives $o_2$}\end{array}\right\}.\]
The above is well-defined, since lengths of zero-frequency streams are positive integers, and the fact $(o_1, o_2)$ is an edge of $G(\rho)$ means at least one such zero-frequency stream exists. Now, using the fact there are finitely many $\rho \in \{0,1\}^r$, and finitely many possible edges $(o_1, o_2)$ in $G(\rho))$, we let $L_0 \in \N$ be the maximum over all $L(\rho, o_1, o_2)$ given by edges of $G(\rho)$. We let $S$ denote the set of all zero-frequency streams of length at most $L_0$.
\end{proof}

\begin{lemma}[Reaching a Terminal Component]\label{lem:reach-terminal}
For any randomized streaming algorithm $\calA$ with finite randomness complexity $r \in \N$ and any $\delta > 0$, there exists a distribution $\calS$ supported on finite-length streams of zero-frequency such that the following holds:
\begin{itemize}
\item For any $o_1 \in \calW(\calA, m)$ and any $\rho \in \{0,1\}^r$. 
\item Draw $\bsigma \sim \calS$ and let $\boldsymbol{o}_2 \in \calW(\calA, m)$ be the working tape from transitioning $\calA$ from $o_1$ with randomness tape $\rho$ and input stream $\bsigma$, and let $\bu \in C(\rho)$ with $\boldsymbol{o}_2 \in \bu$.
\end{itemize}
Then,
\begin{align*}
\Prx_{\bsigma \sim \calS}\left[ \bu \in T(\rho) \right] \geq 1 -\delta. 
\end{align*}
\end{lemma}

\begin{proof}
We will let $\calS$ be the distribution obtained by concatenating sufficiently many streams sampled from $S$. First, for any fixed $\rho \in \{0,1\}^r$ and any working tape $o_1$, there exists a path in $G(\rho)$ which begins at $o_1$ and ends at some $o_2$ in a terminal component of $H(\rho)$ (existence here is because the graph $G(\rho)$ is finite), and furthermore, such a path traverses at most $|\calW(\calA, m)|$ edges. Each edge along the path may be labeled with a zero-frequency stream in $S$ from Lemma~\ref{lem:finite-length}, and hence the concatenation of these streams results in a zero-frequency stream $\sigma$ of finite length, such that transitioning $\calA$ from $o_1$ with randomness $\rho$ and stream $\sigma$ leads to a working tape $o_2$ lying in a terminal component. Hence, we let $\tilde{S}$ be the finite set of zero-frequency streams such that, for all $o_1, o_2 \in \calW(\calA, m)$, where $o_1$ can reach $o_2$ in $G(\rho)$,
\[ \exists \sigma \in \tilde{S} \qquad :\qquad \begin{array}{c} \text{transitioning $\calA$ from $o_1$ with randomness tape $\rho$} \\ \text{with input stream $\sigma$ gives working tape $o_2$.} \end{array}\]  
We define $\calS$ as the distribution supported on zero-frequency streams which concatenates $|\tilde{S}| \log(1/\delta)$ streams drawn uniformly from $\tilde{S}$. We now show that for any $\rho \in \{0,1\}^r$ and any $o_1 \in \calW(\calA, m)$, with probability at least $1-\delta$ over the draw of $\bsigma \sim \calS$, if $\boldsymbol{o}_2$ is the resulting working tape after transitioning $\calA$ from $o_1$ with randomness tape $\rho$ and input stream $\bsigma$, then $\boldsymbol{o}_2$ lies in a terminal component of $H(\rho)$. This is true because: 
\begin{itemize}
\item Given any working tape $o'$, there is at least one stream in $\tilde{S}$ which, if sampled from $\tilde{S}$, will transition $\calA$ to a working tape in a terminal component. 
\item In addition, once $\calA$ has a working tape in a terminal component, updating via zero-frequency streams leads to working tapes in the same terminal component.
\end{itemize}
Thus, the probability that $\calA$ is not in a terminal component after transitioning $\bsigma \sim \calS$ is at most the probability that at each of the $|\tilde{S}| \log(1/\delta)$ draws from $\tilde{S}$, the sampled stream from $\tilde{S}$ avoided the special streams for the working tape at the moment of sampling. Since the probability of that we sample a special stream at any moment is at least $1/|\tilde{S}|$, the probability that all $|\tilde{S}| \log(1/\delta)$ draws avoid the special stream is at most
\[ \left(1 - 1/|\tilde{S}| \right)^{O(|\tilde{S}| \log(1/\delta))} \leq \delta. \]
\end{proof}

We may now define the distribution $\pi(\rho)$ which is used to select an initial terminal component $\bv \in T(\rho)$ to be used in the choice of $\Enc(\cdot; \rho, \bv)$.

\begin{definition}[Initialization Distribution]\label{def:initialization}
For any $\rho \in \{0,1\}^r$, let $\pi(\rho)$ be the distribution supported on terminal components $T(\rho)$ of $H(\rho)$ given from the following procedure:
\begin{itemize}
\item Let $o \in \calW(\calA, 0)$ be the initial working tape of $\calA$. 
\item Let $\boldsymbol{o}_2$ be the working tape after transitioning $\calA$ from $o$ with randomness tape $\rho$ and input stream $\bsigma \sim \calS$ (from Lemma~\ref{lem:reach-terminal}).
\end{itemize}
We let $\bv$ be the component of $H(\rho)$ which contains $\boldsymbol{o}_2$. The distribution $\pi(\rho)$ outputs $\bv$ if $\bv \in T(\rho)$ and otherwise a uniform choice in $T(\rho)$.\footnote{We note that from~Lemma~\ref{lem:reach-terminal}, the second case, that $\pi(\rho)$ samples a uniform element from $T(\rho)$ occurs with probability at most $\delta$.} 
\end{definition}

We now give the following lemma which constructs a distribution over zero-frequency streams (once more, used in the correctness analysis) to effectively assume that the algorithm $\calA$ transitions to a working tape whose distribution is close to the unique stationary distribution obtained from a fixed random walk over working tapes in terminal components. The following claim simply follows from the fact that a lazy random walk on a strongly connected component gives us an irreducible and aperiodic Markov chain which converges to a unique stationary distribution (see Theorem~4.9 of~\cite{LP17}).

\begin{claim}\label{cl:convergence}
For any $\rho \in \{0,1\}^r$ and $u \in T(\rho)$, the lazy random walk given by transitioning $\calA$ from any $o \in u$ with $\bsigma \sim S$ (from Lemma~\ref{lem:finite-length}) has a unique stationary distribution $\xi(\rho;u)$ supported on the vertices of $u$ in $G(\rho)$.
\end{claim}

\begin{lemma}[Mixing in Component]\label{lem:mix-component}
For any randomized streaming algorithm $\calA$ with finite randomness complexity $r \in \N$ and any $\delta > 0$, there exists a distribution $\calM$ supported on finite-length streams of zero-frequency such that the following holds:
\begin{itemize}
\item For any $\rho \in \{0,1\}^r$ and any $o_1 \in \calW(\calA, m)$ with a unique terminal component $u \in T(\rho)$ reachable from $o_1$ in $G(\rho)$.
\item Let $\boldsymbol{o}_2 \in \calW(\calA, m)$ be the working tape after transitioning $\calA$ from $o_1$ with randomness tape $\rho$ and input stream $\bsigma \sim \calM$. 
\end{itemize}
Then, the distribution of $\boldsymbol{o}_2$ is $2\delta$-close to the (unique) stationary distribution $\xi(\rho; u)$. 
\end{lemma}

\begin{proof}
From Claim~\ref{cl:convergence}, for any $\rho \in \{0,1\}^r$ and any $u \in T(\rho)$, the lazy random walk on $u$ from transitioning according to $\bsigma \sim S$ converges to the unique stationary distribution $\xi(\rho; u)$, and hence there is a finite $t(\rho, u)$ such that $t(\rho, u)$ steps of the random walk suffice for the random walk which starts at an arbitrary working tape $o \in u$ is within total variation distance at most $\delta$ from $\xi(\rho; u)$. Let $t_0$ be the maximum over all $\rho \in \{0,1\}^r$ and all $u \in T(\rho)$ of $t(\rho, u)$, and consider the distribution $\calM$ which is given by (i) drawing a sample $\bsigma_0 \sim \calS$ from Lemma~\ref{lem:reach-terminal}, (ii) taking $t_0$ samples $\bsigma_1, \dots, \bsigma_{t_0} \sim S$, and then (iii) concatenating the streams $\bsigma_0, \bsigma_1, \dots, \bsigma_{t_0}$. 

For any $\rho \in \{0,1\}^r$ and any $o_1 \in \calW(\calA, m)$, Lemma~\ref{lem:reach-terminal} implies that the distribution over working tapes $\boldsymbol{o}_2$ obtained after transitioning $\calA$ from $o_1$ with randomness tape $\rho$ and stream $\bsigma_0$ will be in a a terminal component component with probability at least $1-\delta$, and if the unique reachable terminal component from $o_1$ is $u$, then $\boldsymbol{o}_2 \in u$ with probability at least $1-\delta$. Whenever that is the case, the subsequent transitions $\bsigma_1, \dots, \bsigma_{t_0} \sim S$ implement $t_0$ steps of a random walk, which is $\delta$-close in total variation distance to $\xi(\rho; u)$. By a union bound, the resulting distribution is $2\delta$-close in total variation distance to $\xi(\rho; u)$.
\end{proof}

\begin{definition}[Family of Decoding Functions]\label{def:decode}
For any $\rho \in \{0,1\}^r$ and any $v \in T_0(\rho)$, let $\Dec(\cdot; \rho, v) \colon T(\rho) \to \{0,1\}$ be given by:
\begin{align*}
\Dec(u; \rho, v) = \ind\left\{ \Prx_{\boldsymbol{o} \sim \xi(\rho; u)}\left[ \calA(\boldsymbol{o}, \rho) = 1 \right] \geq \frac{1}{2} \right\},
\end{align*}
where $\calA(\boldsymbol{o}, \rho)$ is the output function of $\calA$ on working tape $\boldsymbol{o}$ and randomness tape $\rho$.
\end{definition}

\begin{lemma}[Simulation of $\calA$]\label{lem:simulate}
For any $\rho \in \{0,1\}^r$, any frequency vector $x \in \interval{-m,m}^n$ and any stream $\sigma$ with frequency vector $x$, the following two distributions supported on $\calW(\calA, m)$ are within total variation distance at most $3\delta$:
\begin{itemize}
\item The first distribution produces a working tape $\boldsymbol{o}$ by:
\begin{enumerate}
\item Sampling $\bsigma_1 \sim \calS$ from Lemma~\ref{lem:reach-terminal}, sampling $\bsigma_2 \sim \calM$ from Lemma~\ref{lem:mix-component}.
\item Executing $\calA$ with randomness tape $\rho$, initial working tape, and input stream given by concatenating $\bsigma_1, \sigma, \bsigma_2$, and outputting the final working tape $\boldsymbol{o} \in \calW(\calA, m)$.
\end{enumerate}
\item The second distribution samples produces a working tape $\boldsymbol{o}'$ by:
\begin{enumerate}
\item Sampling $\bv' \sim \pi(\rho)$, letting $\bu' = \Enc(x; \rho, \bv')$.
\item Then, sampling $\boldsymbol{o}' \sim \xi(\rho, \bu')$, and outputting $\boldsymbol{o}'$.
\end{enumerate}
\end{itemize}
\end{lemma}

\begin{proof}
Consider the following distribution $\calD_1$ which generates tuples $(\bsigma_1, \boldsymbol{o}_2, \bv, \bv', \boldsymbol{o}_3, \boldsymbol{o}_3, \bu, \bu', \bsigma_2, \boldsymbol{o})$ in the following way: 
\begin{enumerate}
\item Sample $\bsigma_1 \sim \calS$ from Lemma~\ref{lem:reach-terminal}.
\item Execute $\calA$ with the initial working tape, input stream $\bsigma_1$ and randomness tape $\rho$ to obtain working tape $\boldsymbol{o}_2$. 
\item Let $\bv \in C(\rho)$ be that which satisfies $\boldsymbol{o}_2 \in \bv$. If $\bv \in T(\rho)$, let $\bv' = \bv$, otherwise, let $\bv'$ be a uniform draw of $T(\rho)$. (It is helpful to consider $\bv = \bv'$, since this will occur with probability at least $1-\delta$.)
\item Consider two executions (by Lemma~\ref{lem:relate}, these will lead to the same terminal component when $\bv = \bv'$):
\begin{itemize}
\item Execute $\calA$ with working tape $\boldsymbol{o}_2$, input stream $\sigma$, and randomness tape $\rho$ to obtain working tape $\boldsymbol{o}_3$. Let $\bu$ be the (lexicographic) first terminal component reachable from $\boldsymbol{o}_3$ in $G(\rho)$.
\item Execute $\calA$ with an arbitrary working tape $\boldsymbol{o}_2' \in \bv'$, input stream $\sigma$, and randomness tape $\rho$ to obtain working tape $\boldsymbol{o}_3'$. Let $\bu'$ be the uniquely reachable terminal component from $\boldsymbol{o}_3'$ (guaranteed to be unique by Lemma~\ref{lem:relate}, since $\bv' \in T(\rho)$
\end{itemize}
\item Let $\calC(\bu,\bu')$ denote a distribution over tuples $(\bsigma_2, \boldsymbol{o}, \boldsymbol{o}')$ which satisfies
\begin{itemize}
\item $\bsigma_2 \sim \calM$, 
\item $\boldsymbol{o}$ is the working tape after transitioning $\calA$ from working tape $\boldsymbol{o}_3$ with input stream $\bsigma_2$ and randomness tape $\rho$.
\item $\boldsymbol{o}' \sim \xi(\rho, \bu')$.
\end{itemize}
and maximizes the probability that $\boldsymbol{o} = \boldsymbol{o}'$.
\end{enumerate}
The fact that $\calA$ is a streaming algorithm means that the final working tape $\boldsymbol{o}$ is exactly distributed as that of the first distribution. To relate to the second distribution, notice that the draw $\bv \sim \pi(\rho)$ in Definition~\ref{def:initialization} is generated by Steps~1, 2, and 3. Lemma~\ref{lem:reach-terminal} implies that with probability at least $1-\delta$, $\bv \in T(\rho)$ and in this case $\bv' = \bv$. Whenever that is the case (we will lose an additive factor $\delta$), Lemma~\ref{lem:relate} implies that the choice of $\boldsymbol{o}_2, \boldsymbol{o}_2'$ is completely arbitrary, and there is a uniquely reachable terminal component from either $\boldsymbol{o}_3, \boldsymbol{o}_3'$, and it is the same uniquely reachable terminal component, so $\bu = \bu'$. In particular, Definition~\ref{def:encodings} implies that $\bu' = \Enc(x; \rho, \bv')$. Finally, Lemma~\ref{lem:mix-component} implies that, whenever $\bu = \bu'$ is the unique reachable terminal component of $\boldsymbol{o}_3$, a coupling $\calC(\bu, \bu')$ over $(\bsigma_2, \boldsymbol{o}, \boldsymbol{o}')$ is guaranteed to satisfy $\boldsymbol{o} = \boldsymbol{o}'$ except with probability at most $2\delta$.  

In summary, the projection of the above tuple to $(\bsigma_1, \bsigma_2, \boldsymbol{o})$ is generated according to the first distribution, and the projection to the tuple $(\bv', \bu', \boldsymbol{o}')$ is generated from the second distribution. By a union bound, $\boldsymbol{o} \neq \boldsymbol{o}'$ is at most $3\delta$, which proves the lemma.
\end{proof}

\begin{lemma}[Correctness]\label{lem:corretness}
For any distribution $\calD$ supported on frequency vectors $\interval{-m,m}^n$, and suppose $\calA$ computes $g \colon \interval{-m,m}^n \to \{0,1,*\}$ with error at most $\delta$, there exists a setting of $\rho \in \{0,1\}^r$, and $v \in T_0(\rho)$ satisfying
\begin{align*}
\Prx_{\bx \sim \calD}\left[ g(\bx) \in \{0,1\} \text{ and }\Dec( \Enc(\bx; \rho, v) ; \rho, v) \neq g(\bx) \right] \leq 8\delta.
\end{align*}
\end{lemma}

\begin{proof}
For any fixed $x \in \interval{-m,m}^n$ and consider an arbitrary stream $\sigma(x)$ with frequency vector $x$. Let $\calS_f(x)$ denote the distribution supported on streams given by (i) sampling $\bsigma_1 \sim \calS$, appending the stream $\sigma(x)$, and (iii) appending another stream $\bsigma_2 \sim \calM$ (note that $\bsigma$ has frequency vector $x$). Let $\boldsymbol{o}$ denote the final working tape after transitioning $\calA$ from the initial tape and stream $\bsigma \sim \calS_f(x)$ with randomness tape $\brho \sim \{0,1\}^r$ (i.e., for any fixed $\brho$, the draw of $\boldsymbol{o}$ is according to the ``first'' distribution of Lemma~\ref{lem:simulate}). 

It is notationally convenient to notice that the condition of the algorithm $\calA$ being \emph{incorrect} on the output from final working tape $o$ and randomness tape $\rho$ after processing a stream $\sigma$ of frequency vector $x$, i.e., $g(x) \in \{0,1\}$ and $\calA(o, \rho)\neq g(x)$, may be equivalent restated as $\calA(o, \rho) \notin \{0,1\} \setminus \{ g(x) \}$. With that notation, since $\calA$ computes $g$ with probability $1-\delta$, if we let $\bx \sim \calD$ and $\bsigma \sim \calS_f(\bx)$, the fact that $\calA$ computes $g$ implies
\begin{align*}
\Ex_{\brho}\left[ \Ex_{\bx \sim \calD}\left[ \Ex_{\bsigma \sim \calS_f(\bx)}\left[ \ind\left\{\calA(\boldsymbol{o}, \brho) \in \{0,1\} \setminus \{g(\bx)\} \right\} \right] \right] \right] &= \Ex_{\bx \sim \calD}\left[ \Ex_{\bsigma \sim \calS_f(\bx)}\left[ \Prx_{\brho}\left[ \calA(\boldsymbol{o}, \brho) \in \{0,1\} \setminus \{g(\bx)\}  \right] \right] \right] \\
	&\leq \delta,
\end{align*}
where the first equality is by independence of $\brho$ from $(\bx, \bsigma)$, and the second by correctness of $\calA$. Hence, there exists a fixed setting of $\rho$ where the expectation, solely over $\bx \sim \calD$ and $\bsigma \sim \calS_f$ is at most $\delta$. Consider such a fixed setting of $\rho$, and for any fixed setting of $x$, Lemma~\ref{lem:simulate} implies that the distribution over final working tapes $\boldsymbol{o}$ is $3\delta$-close in total variation distance from the distribution over $\boldsymbol{o}'$ given by $\bv' \sim \pi(\rho)$, $\bu' = \Enc(x; \rho, \bv')$ and $\boldsymbol{o}' \sim \xi(\rho, \bu')$. Since an indicator random variable lies in $\{0,1\}$, the fact the distributions are $3\delta$-close means we obtain the following bound on the expectation over $\bx \sim \calD$,
\begin{align*}
\Ex_{\bx \sim \calD}\left[ \Ex_{\bv' \sim \pi(\rho)} \left[ \Prx_{\bu',\boldsymbol{o}'}\left[ \calA(\boldsymbol{o}', \rho) \in \{0,1\} \setminus \{g(\bx)\}\right]\right] \right] \leq 4\delta,
\end{align*}
and once more, since $\bv'$ is independent of $\bx$, we may switch the order of expectations, and conclude there exists a fixed setting of $\bv' = v$ where the expectation, solely over $\bx, \bu'$ and $\boldsymbol{o}'$ is at most $4\delta$. Note that, $\bu' = \Enc(\bx; \rho, v)$ depends solely on $\bx$, so we will drop the expectation over $\bu'$ whenever the expectation over $\bx$ is present. Notice that $\Dec(\bu'; \rho, v) \in \{0,1\} \setminus \{ g(\bx) \}$ implies that at least $1/2$ fraction of draws $\boldsymbol{o}'$ have $\calA(\boldsymbol{o}', \rho) \in \{0,1\} \setminus \{g(\bx)\}$, and therefore,
\begin{align*}
\Prx_{\bx \sim \calD}\left[ \Dec(\bu'; \rho, v) \in \{0,1\} \setminus \{ g(\bx) \}\right] \leq 2 \cdot \Ex_{\bx \sim \calD}\left[ \Prx_{\boldsymbol{o}'} \left[ \calA(\boldsymbol{o}', \rho) \in \{0,1\} \setminus \{g(\bx)\} \right]\right] \leq 8\delta.
\end{align*}
\end{proof}

\subsection{Putting Everything Together: Proof of Theorem~\ref{thm:lnw}}\label{sec:putting-everything-together}

All elements of the proof have their components appearing in previous sections, and the goal of this section is to provide a guide to the corresponding definitions and lemmas. We select any large enough setting of $m \in \N$ such that the distribution $\calD$ is supported on frequency vectors in $\interval{-m,m}^n$. 

The deterministic, path-independent algorithm $\calB = (\calW, \Enc,\Dec, \oplus)$ is defined in several parts, with certain parameters left unspecified initially. Specifically, the parameters $\rho \in \{0,1\}^r$ and $v$ are fixed later according to Lemma~\ref{lem:corretness}. The first part of the definition, which specifies $\calW$ and the encoding function $\Enc$, and the transitions $\oplus$, appears in Subsection~\ref{sec:encoding}, where

\begin{itemize}
\item The set of states $\calW$, parametrized by $\rho$, is defined in Definition~\ref{def:states} as $T(\rho)$. This set corresponds to the set of vertices in the state-transition graph $H(\rho)$ (Definition~\ref{def:vert-state-transition}).
\item Within $T(\rho)$, a distinguished subset of states $T_0(\rho)$ is defined in Definition~\ref{def:initial-states}. The parameter $v$ is chosen from $T_{0}(\rho)$, though left unspecified in this subsection. 
\item The encoding function $\Enc(\cdot; \rho, v)$, as well as the transition rule $\oplus_{\rho, v}$ are specified in Definition~\ref{def:encodings} and Definition~\ref{def:trans}.
\end{itemize}
Subsection~\ref{sec:encoding} proves these definitions are well-formed, and Subsection~\ref{sec:path-and-space} establishes the path-independence property (with bound $m$), as well as the claimed space complexity. In particular, Lemma~\ref{lem:path-ind} proves that, for any parametrization $\rho \in \{0,1\}^r$ and $v \in T_0(\rho)$, the deterministic encoding function $\Enc(\cdot; \rho, v) \colon \interval{-m,m}^n \to \calW$ satisfies the ``path-independence'' property of Definition~\ref{def:path-ind}. Then, Lemma~\ref{lem:space} proves the desired space complexity bound, that $\calS(\calB, \ell) \leq \calS(\calA, \ell)$ and $\calS^+(\calB, \ell) \leq \calS^+(\calA, \ell)$ for all $\ell \leq m$, since $\calS(\calA, \ell)$ is the logarithm of $|\calW(\calA, \ell)|$, and $\calS(\calB, \ell)$ is the analogous quantity for path-independent algorithms (again, in Definition~\ref{def:path-ind}); the bound for $\calS^+(\calB, \ell) \leq \calS^{+}(\calA,\ell)$ follows analogously.

Finally, Subsection~\ref{sec:decode} defines the decoding function $\Dec(\cdot; \rho, v)$ (Definition~\ref{def:decode}). Lemma~\ref{lem:corretness} proves correctness: there exists a fixed choice of $\rho \in \{0,1\}^r$ and $v \in T_0(\rho)$ such that $\Enc(\cdot; \rho, v)$ and $\Dec(\cdot; \rho, v)$ compute $g$ over the distribution $\calD$ with probability $1-8\delta$.

\section{A Basis of $\R^n$ from the Path-Independent Algorithm}\label{sec:path-ind-to-matrix}

In this section, we consider a fixed deterministic and path-independent algorithm $\calB$, which is specified by $(\calW, \Enc, \Dec, \oplus)$, and the goal of this section is to extract a special $s \times n$ matrix $T$ whose complexity  (i.e., the number of rows $s$ and the magnitude of entries in $T$) depends on the space complexity of the algorithm $\calB$. The discussion follows the works of~\cite{G08, LNW14} which show that the so-called ``kernel'' of the encoding function $\Enc(\cdot)$ of a path-independent algorithm can be used to define a sub-module of $\Z^n$. Because we have imposed the condition that $\Enc(\cdot)$ satisfies the path-independence property to bound $m$, and because the function $\Enc(\cdot)$ is only defined on vectors $\interval{-m, m}^n$, we provide the arguments. For the rest of the section, we consider a fixed setting of a large enough bound $m \in \N$, as well as a fixed $\calB = (\calW, \Enc, \Dec, \oplus)$, and we refer the reader to Definition~\ref{def:zero-and-spanning} for the definition of the zero-set of $\Enc(\cdot)$, as well as the zero-set subspace.

\ignore{\begin{definition}[Zero-Set and the Spanning Set]\label{def:zero-and-spanning}
For $\ell \leq m$, we define:
\begin{itemize}
\item \textbf{Zero-Set}. The zero-set of $\Enc(\cdot)$ within $\interval{-\ell, \ell}^n$ is given by
\[ K(\ell) = \left\{ x \in \interval{-\ell, \ell}^n : \Enc(x) = \Enc(0) \right\}.\]
\item \textbf{Zero-Set Subspace}. The zero-set subspace $M(\ell) \subset \R^n$ is $M(\ell) = \Span(K(\ell))$, where the span is taken over $\R$.
\end{itemize}
We let $t(\ell) \in \N$ be the dimension of $M(\ell)$ and $b_1,\dots, b_{t(\ell)} \in K(\ell)$ be an arbitrary basis of $M(\ell)$. 
\end{definition}}

Note, we defined the zero-set as a discrete subset of $\{-\ell, \dots, \ell\}^n$ and the zero-set subspace as a subspace over $\R^n$. We avoid the discussion of modules (and the fact that an extension of $K(\ell)$ is a sub-module of $\Z^n$) so as to not confuse notions of linear independence over $\Z$ and $\R$. In particular, we will always work with subspaces of $\R^n$, so that linear combinations are always taken with respect to real numbers. 

\begin{definition}[Equivalence Classes]\label{def:equiv}
For $\ell \in \N$ with $\ell < m/2$ and $x \in \interval{0,\ell}^n$, the equivalence class $C(x;\ell)$ is 
\[ C(x;\ell) = \left\{ y \in \interval{0, \ell}^n : y - x \in K(2\ell) \right\}. \]
\end{definition}

The fact that we restrict $\ell < m/2$ in Definition~\ref{def:equiv} is due to the subsequent lemma, where we use the fact $\calB$ is path-independent with bound $m$ to connect the encoding map $\Enc(\cdot)$ with the equivalence classes $C(x;\ell)$ defined above. Since path-independence with bound $m$ (in Definition~\ref{def:path-ind}) requires us to restrict the underlying vectors to those within $\{-(m-1),\dots, m-1\}^n$, the condition $\ell < m/2$ enables us to conclude that any $x, y\in \interval{0,\ell}^n$ has $x-y \in \interval{-(m-1),m-1}^n$. 

\begin{lemma}\label{lem:enc-to-equiv}
For $\ell \in \N$ with $\ell < m/2$ and any two $x, y \in \interval{0, \ell}^n$, 
\[ \Enc(x) = \Enc(y) \iff y \in C(x;\ell) \iff x \in C(y;\ell). \]
\end{lemma}

\begin{proof}
Let $0^n = p_0, p_1, \dots, p_{h} = x$ denote a shortest path in a graph whose vertices are $\interval{0, \ell}^n$, and there is an edge connecting $(u,v)$ whenever $u - v = \xi e_i$ for some $\xi \in \{-1,1\}$ and $i \in [n]$. Notice that for all $k \in [h]$, we have $x - p_k \in \interval{-2\ell, 2\ell}^n$ and $y - p_k \in \interval{-2\ell, 2\ell}^n$ since $x, y, p_k \in \interval{0, \ell}^n$ since we selected a shortest path. By the setting of $\ell < m/2$, we always have $x-p_k, y-p_k \in \interval{-(m-1), m-1}^n$, and hence we can always apply the path-independence property (Definition~\ref{def:path-ind}) to conclude, if $(p_k, p_{k+1})$ is the edge with $p_{k} - p_{k+1} = \xi e_i$, then 
\begin{align*} 
\Enc(x-p_k) \oplus \xi e_i &= \Enc(x - p_k + \xi e_i) = \Enc(x - p_{k+1}),\qquad \text{and} \\
\Enc(y-p_k) \oplus \xi e_i &= \Enc(y - p_k + \xi e_i) = \Enc(y - p_{k+1}). 
\end{align*}
Therefore, a simple induction on all $k \in [h]$, we have 
\[ \Enc(x - p_{k}) = \Enc(y - p_k) \iff \Enc(x) = \Enc(y) \]
Setting $k = h$ implies $\Enc(0) = \Enc(y-x)$, i.e., $y - x \in K(2\ell)$ and $y \in C(x;\ell)$, iff $\Enc(x) = \Enc(y)$. The final equivalence is because the choice of $x$ and $y$ is symmetric.
\end{proof}

We now show that the size of the equivalence classes $C(x;\ell)$ is related to the dimensionality of the zero-set subspace $M(\ell)$. We will then use the lemma below in order to conclude that small-space streaming algorithms must have zero-sets whose subspace has dimensionality close to $n$.

\begin{lemma}\label{lem:size-equiv}
For $\ell \in \N$ and $\ell < m/2$, any $x \in \interval{0, \ell}^n$ satisfies $|C(x;\ell)| \leq (\ell+1)^{t(2\ell)}$. 
\end{lemma}

\begin{proof}
Let $B$ be the $n \times t(2\ell)$ matrix whose columns correspond to the vectors $b_1,\dots, b_{t(2\ell)} \in K(2\ell)$. The matrix $B$ has rank $t(2\ell)$ by definition, and therefore there exists a set of $t(2\ell)$ rows $R \subset [n]$ whose columns are linearly independent. Let $T$ be the $t(2\ell) \times n$ matrix with a single $1$ in each row such that $TB$ is the $t(2\ell) \times t(2\ell)$ full-rank matrix which selects the rows of $B$. Consider the map $f \colon C(x;\ell) \to \interval{0, \ell}^{t(2\ell)}$ given by $y \in C(x;\ell) \mapsto Ty \in \interval{0,\ell}^{t(2\ell)}$, whose image has size $(\ell + 1)^{t(2\ell)}$. We now argue that the map is injective as follows: suppose $y, y' \in C(x;\ell)$ with $Ty = Ty'$, then $y - y' \in K$, so there exists a unique $\alpha \in \R^{t(2\ell)}$ where $y - y' = B \alpha$. Then, we have $0 = T(y-y') = TB \alpha$, but since $TB$ is full-rank, $\alpha = 0$ and hence $y - y' = 0$.
\end{proof}

\begin{lemma}\label{lem:space-bound}
For $\ell \in \N$ with $\ell < m/2$, $\calS^{+}(\calB, \ell) \geq (n-t(2\ell)) \log_2(\ell+1)$. 
\end{lemma}

\begin{proof}
Recall that $\calS^{+}(\calB, \ell)$ is the logarithm of the distinct number of encoded states resulting from non-negative vectors. By Lemma~\ref{lem:enc-to-equiv}, 
\begin{align*}
\left| \left\{ \Enc(x) \in \calW : x \in \interval{0, \ell}^n \right\} \right| = \left| \left\{ C(x; \ell) : x \in \interval{0,\ell}^n \right\} \right| &\geq \frac{(\ell+1)^n}{\max\{ |C(x;\ell)| : x \in \interval{0, \ell}^n \}} \\
				&\geq (\ell+1)^{n - t(2\ell)},
\end{align*}
where the last inequality follows from~Lemma~\ref{lem:size-equiv}.
\end{proof}

\ignore{\begin{lemma}\label{lem:basis-from-enc}
Let $\calB = (\calW, \Enc,\Dec, \oplus)$ be a deterministic and path-independent algorithm. There exists a basis $b_1,\dots, b_n$ of $\R^n$ and an integer $t$ with $n - \calS^{+}(\calB, 1) \leq t \leq n$ such that:
\begin{itemize}
\item The first $t$ vectors $b_1,\dots, b_t$ span the zero-set subspace $M = M(1)$ with $\ell=1$.
\item The remaining vectors $b_{t+1}, \dots, b_n$ span $M^{\perp}$, and are integer vectors with $\|b_{t+1}\|_{\infty}, \dots, \|b_{n}\|_{\infty} \leq t^{t/2}$. 
\end{itemize}
\end{lemma}}

\begin{proof}[Proof of Lemma~\ref{lem:basis-from-enc}]
We let $\ell=1$, so that $t = t(2\ell) \leq n$ and $\calS^+(\calB, 1) \geq n - t$ from Lemma~\ref{lem:space-bound}. We let $b_1,\dots, b_t$ be the vectors from Definition~\ref{def:zero-and-spanning} giving a basis of $M$, where note their entries lie in $\{-1, 0,1\}$. The final $n-t$ vectors are constructed as follows. Let $B$ be the $t \times n$ matrix whose rows are $b_1,\dots, b_t$, and note that $x\in M^{\perp}$ if and only if $Bx = 0$. Since $B$ matrix has rank $t$, there is a set of $t$ columns $C \subset [n]$ which are linearly independent, and hence consider the $t\times t$ sub-matrix $B'$ given by the columns of $B$ in $C$, since $B'$ is a full-rank $t \times t$ matrix, it is invertible. For any $j \in [n]\setminus C$, let $B_j \in \{-1,0,1\}^t$ denote the $j$-th column of $B$, and the fact $j \notin C$ means there exists $\alpha \in \R^{C}$ (here, we index the coordinates of $\alpha$ by vectors from $C$) such that 
\[ B_j = B' \alpha = \sum_{i \in C} \alpha_i B_i, \] 
and by Cramer's rule, each $\alpha_i$ for $i \in C$ is the ratio of a determinant of $t \times t$ matrix with entries in $\{-1,0,1\}$ and $\det(B')$. Consider the vector $x^{(j)} \in \R^n$ given by: 
\[ x^{(j)}_j = -\det(B') \qquad \text{and}\qquad x^{(j)}_i = \left\{ \begin{array}{cc} 0 & i \in [n] \setminus C \setminus \{j \} \\
															\alpha_i \cdot \det(B') & i \in C \end{array} \right. .\]
Note, $x^{(j)}$ has integer entries with magnitude at most $t^{t/2}$ using Hadamard's inequality. Furthermore, each $x^{(j)}$ satisfies $Bx^{(j)} = 0$ and hence $x^{(j)} \in M^{\perp}$, and is the only one with non-zero $j$-th coordinate. Hence, the vectors $x^{(j)}$ for $j \notin [n] \setminus C$ gives us the remaining linearly independent $n-t$ vectors.  
\end{proof}